\documentclass[journal,11pt,onecolumn,draftclsnofoot,]{IEEEtran}
\IEEEoverridecommandlockouts

\newcommand{\circled}[1]{\small{\raisebox{.6pt}{\textcircled{\raisebox{-.8pt}{#1}}}}}

\usepackage{amssymb}
\usepackage{amsmath,mathrsfs,dsfont}
\usepackage{nicefrac}
\usepackage{algorithm}
\usepackage{algorithmicx}
\usepackage{algpseudocode}

\usepackage{booktabs}

\usepackage{color}
\usepackage{enumitem}

\usepackage{array,cite}
\usepackage{graphicx,tikz}
\usepackage[mathscr]{euscript}
\usepackage{amsthm}
\usepackage{cite}
\usepackage{calrsfs}

\usepackage{bm}
\usepackage{bbm}
\usepackage{color}

\usepackage{color}
\usepackage{epstopdf}
\usepackage{subcaption}
\usepackage{cleveref}
\usepackage{thmtools}
\usepackage{thm-restate}

\usepackage{bbding}

\newcounter{optproblem}

\theoremstyle{plain}
\newtheorem{theorem}{Theorem}
\newtheorem{remark}[theorem]{Remark}
\newtheorem{proposition}[theorem]{Proposition}
\newtheorem{corollary}[theorem]{Corollary}

\newtheorem*{theorem*}{Theorem}
\newtheorem*{lemma*}{Lemma}
\newtheorem*{remark*}{Remark}
\newtheorem{lemma}[theorem]{Lemma}%[section]

\newenvironment{proofoutline}
{\begin{proof}[Proof outline]}
{\end{proof}}

\theoremstyle{definition}

\theoremstyle{remark}
\DeclareMathAlphabet{\pazocal}{OMS}{zplm}{m}{n}
\DeclareMathAlphabet{\mathpzc}{OMS}{pzc}{m}{it}

\setlist[itemize]{leftmargin=*}

\renewcommand{\hat}{\widehat}

\newcommand{\bfm}[1]{\ensuremath{\mathbf{#1}}}
\newcommand{\bfsym}[1]{\ensuremath{\boldsymbol{#1}}}

   \def\bA{\bfm A}  
   \def\bB{\bfm B}  \def\BB{\mathbb{B}}

\def\be{\bfm e}     
    
\def\bg{\bfm g}     
     
   \def\bI{\bfm I}

\def\br{\bfm r}     \def\RR{\mathbb{R}}

   \def\bU{\bfm U}  
   \def\bV{\bfm V}  
   \def\bW{\bfm W}  
\def\bx{\bfm x}   \def\bX{\bfm X}  
\def\by{\bfm y}   \def\bY{\bfm Y}  
\def\bz{\bfm z}   \def\bZ{\bfm Z}

\def\calE{{\cal  E}}

\def\calH{{\cal  H}}

\def\calP{{\cal  P}}

\def\bZero{\bfm 0}

\def\bmu{\bfsym {\mu}}

\def\bpi{\bfsym {\pi}}
\def\bPi{\bfsym {\Pi}}
\def\btheta{\bfsym {\theta}}
\def\bTheta {\bfsym {\Theta}}

\def\bSigma{\bfsym \Sigma}

\def\bxi{\bfsym {\xi}}

\def\+#1{\mathcal{#1}}
\def\-#1{\textup{#1}}

\def\br#1{\overline{#1}}

\def\bracket#1{\left(#1\right)}
\def\Bracket#1{\left[#1\right]}

\def\defequal {\triangleq}

\def\wh{\widehat}
\def\wt{\widetilde}

\newcommand{\abs}[1]{\left\lvert#1\right\rvert}

\newcommand{\diag}{\textup{diag}}

\newcommand{\la}{\left \langle}
\newcommand{\ra}{\right \rangle}

\newcommand{\La}{\left\langle\kern-0.64ex\left\langle}
\newcommand{\Ra}{\right\rangle\kern-0.64ex\right\rangle}

\def\Norm#1#2{{\left\vert\kern-0.4ex\left\vert\kern-0.4ex\left\vert #1
    \right\vert\kern-0.4ex\right\vert\kern-0.4ex\right\vert}_{#2}}
\def\norm#1#2{{\left\|#1\right\|}_{#2}}

\def\ltwonorm#1{\norm{#1}{2}}
\def\fnorm#1{\norm{#1}{\textup{F}}}

\def\opnorm#1{\norm{#1}{\textup{OP}}}

\def \Proj  {\mathbb{P}}

\def \rmt   {\top}

\newcommand{\trace}{\textup{trace}}

\def\set#1{\left\{#1\right\}}

\newcommand{\argmax}{\textup{argmax}}
\newcommand{\argmin}{\textup{argmin}}

\def \Ind {\mathbbm{1}}
\def \St  {\textup{~s.t.~}}

\def \Expc {\mathbb{E}}
\def \Cov  {\textup{cov}}

\def \logdet {\log\det}
\def \normdist {\mathcal{N}}
\def \Prob {\textup{Pr}}

\def \Var  {\textup{var}}

\def \lsim {\lesssim}
\def \gsim {\gtrsim}

\def \vcap {\wedge}
\def \vcup {\vee}

\def \snr {\mathsf{\mathbf{snr}}}
\def \dh {\mathsf{d_H}}
\def \entH {\mathsf{H}}
\def \enth {\mathsf{h}}
\def \entI {\textup{I}}

\def\Proj{P}

\title{\huge \bfseries The Benefits of Diversity: Permutation Recovery in Unlabeled Sensing from
Multiple Measurement Vectors}
\author{\vspace{0.2in}\textbf{Hang Zhang, Martin Slawski, and Ping Li}
\\ Cognitive Computing Lab\\ Baidu Research\\
10900 NE 8th ST, Bellevue, WA 98004, USA
}
\date{}
\begin{document}
\maketitle

%===================================

\begin{abstract}
In~\footnote{Partial preliminary results appeared in 2019 IEEE International Symposium on Information Theory (ISIT'19), Paris, France.
}  ``Unlabeled Sensing", one observes a set of linear measurements of an underlying signal with incomplete or missing information about their ordering,
which can be modeled in terms of an unknown permutation.
Previous work on the case of a single noisy measurement vector has exposed two main challenges: 1) a high requirement concerning the \emph{signal-to-noise ratio} ($\snr$), i.e., approximately
of the order of $n^{5}$, and 2) a massive computational burden in light of NP-hardness in general.
In this paper, we study the case of \emph{multiple} noisy measurement vectors (MMVs)
resulting from a \emph{common} permutation and investigate to what extent the
number of MMVs $m$ facilitates permutation recovery by ``borrowing strength''.
The above two challenges have at least partially been resolved within our work.
First, we show that a large stable rank of the signal significantly reduces the required snr which can drop from a polynomial in $n$ for $m = 1$ to a constant for $m
= \Omega(\log n)$, where $m$ denotes the number of MMVs and
$n$ denotes the number of measurements per MV.
This bound is shown to be sharp and is associated with a phase transition phenomenon.
Second, we propose a computational scheme for recovering the unknown permutation in practice. For the ``oracle case" with the known signal, the maximum likelihood (ML) estimator reduces to a linear assignment problem whose
global optimum can be obtained efficiently. For the case in which both the signal and permutation are unknown,
the problem is reformulated as a bi-convex optimization problem with an auxiliary variable, which can be solved by the Alternating Direction Method of Multipliers (ADMM). Numerical experiments based on the proposed computational scheme confirm the tightness of our theoretical analysis.
\end{abstract}

\newpage

\section{Introduction}
Noisy linear sensing with $m$ measurement vectors is described by the relation
\begin{equation}
\label{eq:limo}
\bY = \bX\bB^{*} + \bW,
\end{equation}
where $\bY \in \RR^{n\times m}$ represents the observed $m$ measurements, $\bX
\in \RR^{n\times p}$ represents the sensing matrix, and the columns of $\bB^{*}
\in \RR^{p \times m}$ contain $m$ signals of interest with dimension $p$
each, and $\bW \in \RR^{n\times m}$ represents additive noise. Model~\eqref{eq:limo} also arises in linear regression modeling with $m$ response
variables and $p$ explanatory variables~\cite{Mardia1979}. Least squares
regression yields the estimator $\hat{\bB} = (\bX)^{\dagger}\bY$, where
$(\cdot)^{\dagger}$ denotes the Moore-Penrose inverse. The properties of
$\hat{\bB}$ under various assumptions on the noise $\bW$ are well-known. In this
paper, we consider the more challenging situation in which we observe $m$
measurements with missing or incomplete information about their ordering, i.e.,
the correspondence between the rows of $\bY$ and the rows of $\bX$ has been lost.
Put differently, we observe data according to~\eqref{eq:limo} up to an unknown
permutation:
\begin{equation}\label{eq:limo_perm}
\bY = \bPi^{*} \bX\bB^{*} + \bW,
\end{equation}
where $\bPi^{*}$ is an $n$-by-$n$ permutation matrix. Ignoring the unknown
permutation can significantly impair performance with regard to the estimation of
$\bB^{*}$. We herein consider recovery of $\bPi^{*}$ given $(\bX, \bY)$. The latter
suffices for signal recovery since with restored correspondence the setup
becomes standard. In addition, recovery of $\bPi^{*}$ may be of its own interest,
as can be seen from selected example applications sketched below that motivate
the setting~\eqref{eq:limo_perm}. It is worth emphasizing that the latter assumes
that the permutation is shared across the $m$ sets of measurements, and hence does
not apply to situations in which each of those involves its individual permutation.

\emph{Header-Free Communication}. As discussed, e.g., in~\cite{pananjady2017denoising, Shi2019HFC}, in sensor networks with stringent requirements concerning latency and communication footprint, it can be beneficial to omit sensor metadata when transmitting measurements to the fusion center in an effort to minimize latency and communication cost. In this case, signal recovery without metadata such as sensor identifiers involves an unknown permutation.

\emph{Post-Linkage Data Analysis}.
It is often much more cost-efficient to combine data from existing databases rather than collecting new data containing all variables of interest. Due to data formatting and data quality issues, linkage of records pertaining to the same entity can be error-prone. As a result, downstream data analysis such as linear regression or estimation of the cross-covariance between $\bX$ and $\bY$ can be affected, and modeling mismatches via a permutation has been studied recently as a mitigation strategy~\cite{slawski2017linear}. %The case
%of multiple dependent variables is also relevant to the estimation of the cross-covariance between  %given its explicit relation to the associated regression coefficients.

\emph{Data Privacy}. In linkage attacks, intruders aim at the disclosure of sensitive data by using external data and record linkage. There is a long history of attacks in which public data was combined with de-identified data to reveal sensitive information, e.g.,~\cite{Sweeney2001, narayanan2008robust}. Those examples involve direct comparison of two data sets $\bPi^* \bX$ and $\bY$; the regression setup~\eqref{eq:limo_perm} arises as a natural generalization.

\emph{Unsupervised Alignment}. Aligning two sets of points is a fundamental task with applications in computer vision, curve registration, and natural language processing. A problem of recent interest is the alignment of embeddings of text corpora into the unit sphere $\mathbb{S}^{p-1}$ in $\mathbb{R}^p$~\cite{Shi2018, Grave19}. For example,~\cite{Shi2018} {formulate} the automated translation between different versions of medical diagnosis codes used in electronic health systems as a problem of the form~\eqref{eq:limo_perm} given two sets of vectors $\bX$ and $\bY$ in $\mathbb{S}^{p-1}$ representing embeddings of two different versions of medical diagnosis codes.

Additional examples can be found among the references provided in the next section.
\subsection{Related work}
The work~\cite{unnikrishnan2015unlabeled} discusses \emph{signal recovery} under
setup~\eqref{eq:limo_perm} dubbed `Unlabeled Sensing'' therein for the
case of a single measurement vector ($m = 1$) and no noise ($\bW = 0$). It is
shown that if the entries of the sensing matrix $\bX$ are drawn from a continuous
distribution over $\RR$, the condition $n \geq 2p$ is required for signal
recovery by means of exhaustive search over all permutation matrices. The
authors also motivate the problem from a variety of applications, including the
reconstruction of spatial fields using mobile sensors, time-domain sampling in
the presence of clock jitter, and multi-target tracking in radar. Alternative
proofs of the main result in~\cite{unnikrishnan2015unlabeled} are shown in
\cite{Domankic2018, Tsakiris2018b}.

A number of recent papers discuss the case $m = 1$ and Gaussian $\bW$. The paper
\cite{pananjady2016linear} {establishes} the statistical limits of exact and
approximate permutation recovery based on the ratio of signal energy and
noise variance henceforth referred to as ``$\snr$". In~\cite{pananjady2016linear}, it is also demonstrated that the least squares estimation of $\bPi^{*}$ is NP-hard in
general. In~\cite{hsu2017linear}, a polynomial-time approximation algorithm is
proposed, and lower bounds on the required $\snr$ for approximate signal recovery
in the noisy case are shown; related results can be found in
\cite{abid2017linear, slawski2017linear}. The works~\cite{slawski2017linear,slawski2019two,SlawskiRahmaniLi2018,Shi2018} discuss both signal and permutation recovery if $\bPi^{*}$ only permutes a
small fraction of the rows of the sensing matrix. An interesting variation of
\eqref{eq:limo_perm} in which $\bPi^{*}$ is an unknown selection matrix that
selects a fraction measurements in an order-preserving fashion is studied in~\cite{Haghighatshoar2017}. The papers~\cite{TsakirisICML19, Peng2020} develop the approach in~\cite{Haghighatshoar2017} further by combining it with a careful branch-and-bound scheme to solve general unlabeled sensing problems.

Several papers~\cite{Emiya2014, pananjady2017denoising,slawski2019two,SlawskiRahmaniLi2018} have studied the setting of multiple
measurement vectors ($m\geq 2$) and associated potential benefits for permutation
recovery. The paper~\cite{Emiya2014} {discusses} a practical
branch-and-bound scheme for permutation recovery but does not provide
theoretical insights. The work~\cite{pananjady2017denoising} analyzes the
\emph{denoising problem}, i.e., recovery of $\bPi^{*} \bX \bB^{*}$, rather than
individual recovery of $\bPi^{*}$ and $\bB^{*}$. In~\cite{slawski2019two,SlawskiRahmaniLi2018}, the number of permuted rows
in the sensing matrix is assumed to be small, and are treated as outliers. Methods for
robust regression and outlier detection are proposed to perform signal recovery. While
both~\cite{slawski2019two,SlawskiRahmaniLi2018} also contain achievability
results for permutation recovery given an estimate of the signal, none of these
works provides information-theoretic lower bounds to assess the sharpness of the results.
Moreover, the method in~\cite{slawski2019two} limits the fraction of permuted rows to a constant multiple of the reciprocal of the signal dimension $p$, while the method in~\cite{SlawskiRahmaniLi2018} requires the number of MMVs $m$
to be of the same order of $p$ and additionally exhibits an unfavorable
running time that is cubic in the number of measurements. In the present paper, we eliminate
the limitations in~\cite{slawski2019two,SlawskiRahmaniLi2018} to a good extent.

\subsection{Summary of contributions}
Results in~\cite{pananjady2016linear} on the case $m = 1$ indicate that the \emph{maximum likelihood} (ML) estimator in~\eqref{eq:sys_ml_estimator} can be regarded as impractical from both statistical and computational viewpoints. On one hand, exact recovery of $\bPi^{*}$
requires $\snr = \Omega(n^c)$, where $c >
0$ is a constant that is approximately equal to $5$ according to simulations. As $n$ grows, this requirement becomes prohibitively strong. On the other hand,
the ML estimator~\eqref{eq:sys_ml_estimator} has been proven to be
NP-hard except for the special case $m = 1$ and $p=1$. To the best
of our knowledge, no efficient algorithm has been proposed yet. %to solve this
%problem.
In this paper, by contrasting $m =1$ and $m \gg 1$,  our goal is  to tackle both obstacles.
Before giving a detailed account of our contribution, we first define
a crucial quantity, the \emph{signal-to-noise-ratio} ($\snr$)
\begin{equation}\label{eq:snr}
    \snr = {\fnorm{\bB^{*}}^2}/{(m \cdot \sigma^2)},
\end{equation}
where $\fnorm{\bA} = \sqrt{\sum_{i,j} \bA_{ij}^2}$ denotes the Frobenius of a matrix $\bA$ of arbitrary dimension.

\begin{itemize}
\item
We improve the requirement $\snr=\Omega(n^c)$ to roughly $\snr=
\Omega(n^{c/\varrho(\bB^{*})})$, where
$\varrho(\bB^{*}) = \frac{\fnorm{\bB^*}^2}{\opnorm{\bB^*}^2}$ is the
so-called stable rank of $\bB^{*}$, which is given by the squared ratio of the Frobenius norm and the operator norm $\opnorm{\cdot}$ of a matrix and constitutes a lower bound on its rank (e.g.,  Section 2.1.15 in~\cite{tropp2015introduction}). Once $\varrho(\bB^{*})$ is of the order
$\Omega(\log n)$, we notice that the $\snr$ is only required to be
of the order $\Omega(1)$ and hence does no longer need to increase with $n$.
The underlying intuition is that larger values of $m$ lead to relaxed
requirements on the $\snr$ since
$1)$ the overall signal energy increases, $2)$ all MMVs result
from the same permutation matrix $\bPi^{*}$, which is expected to yield extra
information. In our analysis, $1)$ is reflected by conditions on permutation
recovery involving dependence on the overall signal energy, while $2)$ enters
via a dependence on the stable rank $\varrho(\bB^{*})$ of the signal matrix $\bB^{*}$.
%========================================
\item
We verify that the theoretical results can be attained in practice. For this purpose, we develop a practical algorithm for recovery
of $\bPi^{*}$ and $\bB^{*}$ via least squares fitting, which
is an NP-hard problem except for the special case with $p = m =1$. Our computational approach is based
on a reformulation as a bi-convex optimization problem involving an auxiliary variable that can be tackled via a scalable \emph{alternating direction method of multipliers} (ADMM) scheme~\cite{boyd2011distributed}.
Extensive numerical results based on this approach align with our theorems and confirm significant reductions of the  $\snr$ required for recovery of $\bPi^*$ as the stable rank $\varrho(\bB^{*})$ of the signal matrix $\bB^*$ increases.
\end{itemize}

We conclude this summary of contributions with an overview presented in Table~\ref{tab:relate_work} that compares the results herein to those obtained in related work.
\begin{table}[h!]
\centering
\caption{Overview on results in related work in comparison to those shown herein. The column ``$h_{\max}$" refers to the maximum Hamming distance between $\bPi^*$ and the identity
matrix with $h_{\max} = n$ referring to the fully shuffled case. ``Computable" refers
to the availability of practical computational schemes that achieve the theoretical guarantees established in each work.}
\begin{tabular}{c|ccccc}
\noalign{\hrule height 0.7pt}
Related work & $\snr$ & $n/p$ & $h_{\mathsf{max}}$ & Minimax Analysis & Computable \\
\hline
%\textbf{Case $m=1$} \\
\cite{unnikrishnan2015unlabeled, Domankic2018, Tsakiris2018b}& $\infty$ & $\Omega(1)$ &  $n$ & $\boldsymbol{\times}$ &  $\boldsymbol{\times}$ \\
\cite{pananjady2016linear}  & $\Omega(n^c)$ & $\Omega(1)$ & $n$ & $\checkmark$ & $\boldsymbol{\times}$  \\
\cite{hsu2017linear} & $\infty$ & $\Omega(1)$ & $n$ & $\boldsymbol{\times}$ &  $\checkmark$ \\
\cite{slawski2017linear}  & $\Omega(n^c)$ & $\Omega(1)$ & $O\Bracket{\frac{n-p}{\log(n/h_{\mathsf{max}})}}$ & $\boldsymbol{\times}$ & $\checkmark$\\
\textbf{This work} & \textbf{$\Omega(n^{c/\varrho(\bB^{*})})$} & \textbf{$\Omega(1)$} & \textbf{$\Omega(\frac{n}{\log n})$} & $\checkmark$ & $\checkmark$ \\
\cite{slawski2019two} &$\Omega(n^{c/\varrho(\bB^{*})})$ & $\Omega(p\vcup h_{\mathsf{max}} )$ & $\Omega\Bracket{\frac{n}{\log(n/h_{\mathsf{max}})}}$ & $\boldsymbol{\times}$ & $\checkmark$   \\
\noalign{\hrule height 0.7pt}
\end{tabular}
\label{tab:relate_work}
\end{table}

\subsection{Outline}

The rest of the paper is organized as follows. The underlying sensing model is reviewed in Section~\ref{sec:sys_mdl}.
In Section~\ref{sec:oracle_recover},  we establish conditions that imply failure of recovery (inachievability).
This is followed by achievability results presented in Section~\ref{sec:prac_recover} and a discussion
of their tightness in relation to the corresponding inachievability results. The empirical evaluation and concluding remarks are
provided in Section~\ref{sec:num_simul} and Section~\ref{sec:conclusion},
respectively. A graphical representation of the structure of this paper is provided in
Fig.~\ref{fig:results}.
%tightness of the bounds are confirmed in
%, which suggests
%that the correct recovery requires the $\snr$ to be of the same order.
%Practical algorithm for the recovery of $\bPi^{*}$ is
%developed in Section~\ref{sec:comp_mtd}.

\begin{figure}[h!]
\begin{center}
\mbox{
\includegraphics[width=2.7in]{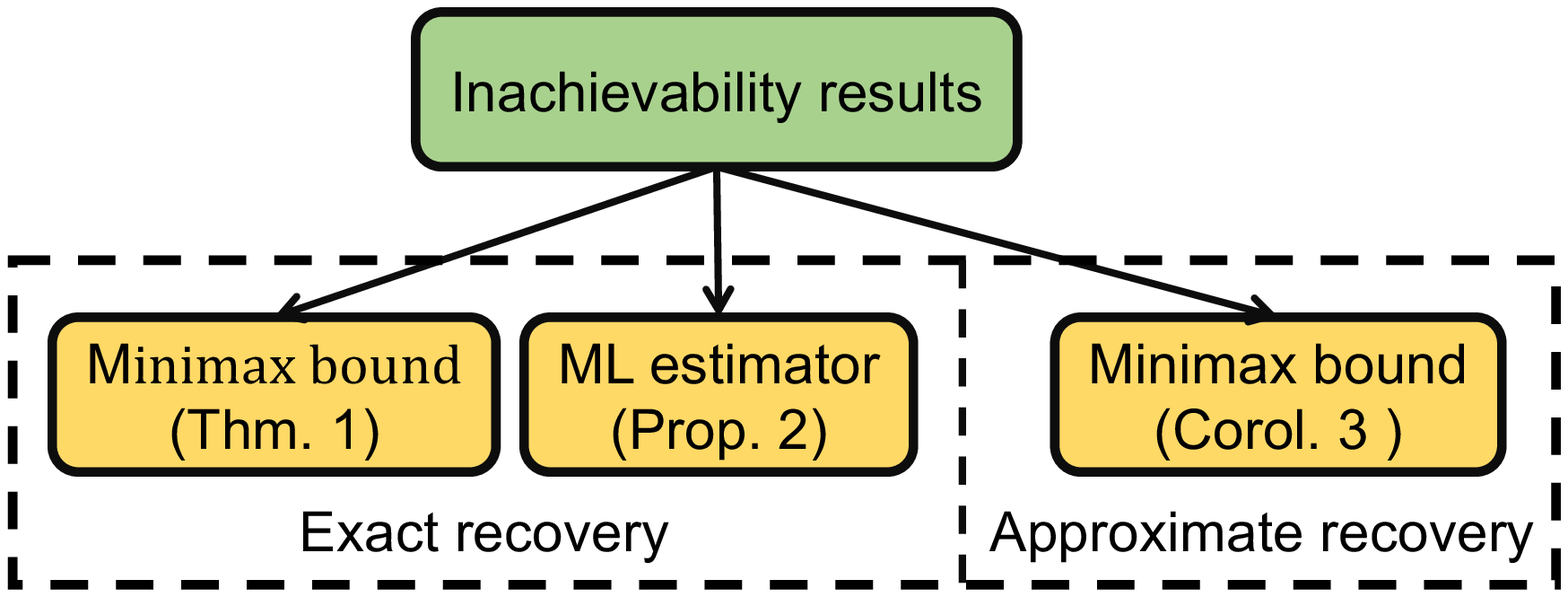}
\includegraphics[width=3.7in]{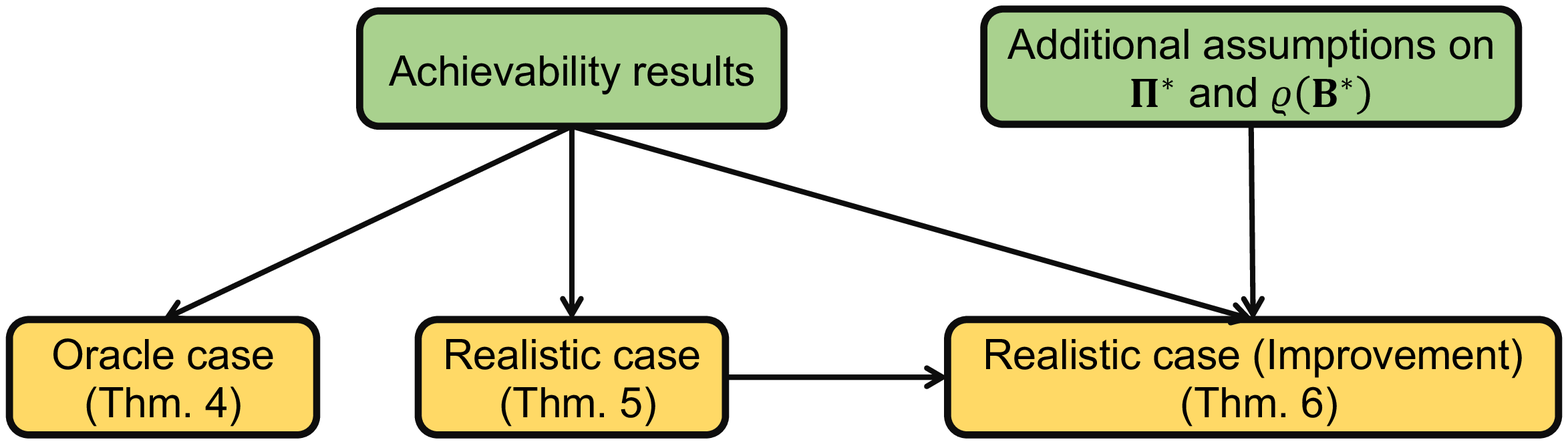}

}
\end{center}
\caption{A roadmap of the main results to be presented in the paper. \textbf{Left panel}:
inachievability results (Section~\ref{sec:oracle_recover});
\textbf{Right panel}: achievability results (Section~\ref{sec:prac_recover}).}
\label{fig:results}
\end{figure}

\section{System Model}\label{sec:sys_mdl}

We recall that the sensing model under consideration reads
\begin{equation}
\label{eq:sys_sense_relation}
\bY = \bPi^{*} \bX \bB^{*} + \bW, 	
\end{equation}
where $\bY \in \RR^{n\times m}$ represents the
results of the sensing process, $\bPi^{*} \in \RR^{n\times n}$ denotes
the unknown permutation matrix, $\bX \in \RR^{n\times p}$
($n\geq 2p$) is the sensing matrix, $\bB^{*} \in \RR^{p\times m}$
is the matrix of signals, and $\bW \in \RR^{n\times m}$
is the sensing noise. For what follows, we assume that the entries
$(X_{ij})$ of $\bX$ are i.i.d.~standard Gaussian random variables, i.e.,
$X_{ij} \sim \normdist(0,1)$, $1 \leq i \leq n, \; 1 \leq j \leq p$. Likewise, we assume that the entries
of $\bW$ are i.i.d.~$\normdist(0, \sigma^2)$-random variables,
where $\sigma^2 > 0$ denotes the noise variance.
% with $W_{ij}\sim \normdist(0, \sigma^2)$.
%======================================
The ML estimator of
$(\bPi^{*}, \bB^{*})$ then results as the least squares solution
%======================================
\begin{equation}
\label{eq:sys_ml_estimator}	
(\hat{\bPi}, \wh{\bB}) = \argmin_{(\bPi,\bB)}~\fnorm{\bY - \bPi \bX \bB}^2.
\end{equation}
Note that for a fixed permutation matrix $\bPi$, we obtain
\begin{equation}
\label{eq:sys_b_estimator}	
\wh{\bB}(\bPi) =(\bPi \bX)^{\dagger} \bY,
\end{equation}
where the superscript $\dagger$ denotes the generalized inverse. From the above, we can see the importance of accurate estimation of $\bPi^{*}$ in a least
squares approach since errors may significantly degrade the quality of the corresponding
estimator $\wh{\bB}$, while exact permutation recovery, i.e., $\wh{\bPi} = \bPi^{*}$ yields the
usual least squares estimator as in the absence of $\bPi^{*}$.
In the following, we put estimation of $\bB^{*}$ aside and concentrate on analyzing the determining factors
for recovery of $\bPi^{*}$. Broadly speaking, this task involves two main sources
of difficulty.
\begin{itemize}
\item
\emph{Sensing noise $\bW$}. In the \emph{oracle case} in which $\bB^{*}$ is known,
computation of the ML estimator of $\bPi^{*}$ reduces to the \emph{linear assignment problem}~\cite{burkard2012assignment} %in~\eqref{eq:sys_ml_estimator}
\begin{equation}\label{eq:LAP}
\hat{\bPi} = \argmax_{\bPi}~\la \bPi,~\bY \bB^{*\rmt}\bX^{\rmt}\ra,
\end{equation}
where $\la \bU, \bV \ra = \trace({\bU^{\rmt} \bV})$ here refers to the inner product between matrices $\bU$ and $\bV$ that induces the Frobenius norm. Even though the solution of \eqref{eq:LAP} can be obtained efficiently by solving a linear program, recovery of $\bPi^{*}$ is still likely to fail if the noise
 level $\sigma^2$ is large enough.
\item \emph{Unknown $\bB^{*}$}. In contrast to the oracle case above,
we have no access to $\bB^{*}$ in practice, which suggests that recovery becomes more challenging.
\end{itemize}
In the sequel, we will show that the
sensing noise $\bW$ constitutes the major difficulty
in recovering $\bPi^{*}$ rather than {the} missing knowledge of
$\bB^{*}$. Before delving into our main results, we first define the following notations. \\

\noindent\textbf{Notations}: \ Positive constants are denoted by $c$, $c'$, $c_0$, $c_1$,
etc.
We write $a \lsim b$ if there is a constant
$c_0$ such that $a \leq c_0 b$. Similarly, we
define $\gsim$.
If both $a \lsim b$ and $a \gsim b$ hold, we write $a \asymp b$. For two numbers $a$ and $b$,  we let $a \vcup b = \max\{a,b\}$ and $a \vcap b = \min\{a,b\}$. %is denoted by
%$a\vcup b$
%while the minimum of $a$ and $b$ is denoted by $a\vcap b$.
For a matrix $\bA\in \RR^{m\times n}$, we denote $\bA_{:, i} \in \RR^n$
as the $i^{\mathsf{th}}$ column of $\bA$ while $\bA_{i, :}$ denotes its $i^{\mathsf{th}}$ row,
viewed as a column vector.
The Frobenius norm of a matrix is represented as
$\fnorm{\cdot}$ while the  operator norm
is denoted as $\opnorm{\cdot}$ whose definition can be found
in~\cite{golub2012matrix} (Section 2.3, P71).
The ratio $\varrho(\cdot) = \fnorm{\cdot}^2/\opnorm{\cdot}^2$ represents the stable-rank
while $r(\cdot)$ represents the (usual) rank.
We denote the \emph{singular value decomposition} (SVD)
of the matrix $\bA$
as $\textup{SVD}(\bA)$, whose definition
can be found in ~\cite{golub2012matrix} (Section 2.4, P76)
and is listed in Appendix~\ref{sec:appendix_notation} as well.
We let $\calP_n$ denote the set of permutation matrices of size $n$. Associating each
$\bPi \in \calP_n$ with a mapping $\bpi$ of
$\set{1, 2,\ldots, n}$
which moves index $i$ to $\bpi(i)$, $1 \leq i \leq n$,
we define the Hamming distance $\dh(\cdot; \cdot)$
between two permutation matrices
as $\dh(\bPi_1; \bPi_2) \defequal \sum_{i=1}^n \Ind\bracket{\bpi_1(i)\neq \bpi_2(i)}$.
The \emph{signal-to-noise-ratio} ($\snr$) is defined as
$\snr = {\fnorm{\bB^{*}}^2}/({m\sigma^2})$.
Additional notation
can be found in Appendix~\ref{sec:appendix_notation}.

\section{Inachievability results}\label{sec:oracle_recover}
In this section, we present conditions under which exact and approximate recovery of $\bPi^{*}$ would \emph{fail} with high probability. To be specific, \emph{exact recovery} refers to
the event $\{ \wh{\bPi} = \bPi^{*} \}$, and \emph{approximate recovery} of  $\bPi^{*}$ within a Hamming ball of radius $\mathsf{D} \in \{0, 1,\ldots,n\}$ refers to the event
%\begin{equation*}
$\{ \dh(\bPi^{*}; \wh{\bPi}) = \sum_{i = 1}^n \Ind\bracket{\bpi^*(i)\neq \widehat{\bpi}(i)}\leq \mathsf{D} \}$,
%\end{equation*}
where $\bpi^*$ and $\wh{\bpi}$ denote the mappings associated with $\bPi^{*}$ and $\wh{\bPi}$, respectively, $1 \leq i \leq n$. The investigation of these cases is intended to provide valuable insights into the fundamental
statistical limits.
In order to establish inachievability results, it suffices to consider the \emph{oracle case}
with $\bB^{*}$ known. The resulting limits apply
to the case of unknown $\bB^{*}$ as well, since
it is hopeless to recover $\bPi^{*}$  even if
knowledge of $\bB^{*}$ does not suffice for recovery.

Compared with the case  $m=1$ in which
$\snr$ is the only prominent factor in determining the recovery performance~\cite{pananjady2016linear},
our analysis uncovers another crucial factor, namely,
the energy distribution over singular values of $\bB^*$.
Our work shows that
a more uniform spread of the signal energy over singular values can
greatly facilitate the recovery of $\bPi^{*}$. %Before we proceed,
%we give the definitions of exact recovery and approximate recovery.

%\begin{definition}[Exact recovery]
%We say $\wh{\bPi}$ is an exact recovery of
%$\bPi^{*}$ if $\wh{\bPi} = \bPi^{*}$, in other words,
%$\wh{\bPi}$ is with zero Hamming distance from the ground truth $\bPi^{*}$,
%i.e., $\dh(\bPi^{*}; \wh{\bPi}) = 0$.
%\end{definition}

%\begin{definition}[Approximate recovery]
%We say $\wh{\bPi}$ is an approximate recovery of
%$\bPi^{*}$ if $\wh{\bPi} = \bPi^{*}$ with Hamming distance $D$
%if $\dh(\bPi^{*}; \wh{\bPi}) \leq \mathsf{D}$, where $ \mathsf{D}\geq 0$ is an
%non-negative integer. When $ \mathsf{D}= 0$, the approximate recovery $\wh{\bPi}$
%reduces to the exact recovery.
%\end{definition}

\subsection{Exact recovery of $\bPi^{*}$}

%Using Fano's inequality (c.f. Theorem~$2.10.1$ in~\cite{cover2012elements}),
We start by presenting an inachievability result concerning exact recovery.
%===================================
\begin{theorem}
\label{thm:exact_minimax}
Let $\calH$ be any subset of $\calP_n$. Assuming that $\bB^{*}$ is known, we
have
\begin{equation}
\label{eqn:exact_minimax}
\inf_{\wh{\bPi}}\sup_{\bPi^{*} \in \calH}\Expc\Ind(\wh{\bPi}\neq \bPi^{*}) \geq
\frac{1}{2} \qquad \text{if} \;\;\, \logdet\bracket{\bI + \frac{\bB^{*\rmt}\bB^{*}}{\sigma^2}}<\
\frac{\log(|\calH|) - 2}{n},
%\inf_{\wh{\bPi}}\sup_{\bPi^{*}}\Expc\Ind(\wh{\bPi}\neq \bPi^{*}) \geq
%\frac{2(1-c_0)\log(|\calH|) - 2}{n}
%1 - \frac{1 + (n/2)\logdet(\bI +\bB^{*\rmt}\bB^{*}/\sigma^2 )}{\log\abs{\calH}},
\end{equation}
where the expectation $\Expc$ is taken w.r.t. $\bX$ and $\bW$,
and the infimum is over all estimators $\wh{\bPi}$.
%$\calH$ is the support of the permutation matrices $\bPi^{*}$.
\end{theorem}

\begin{proofoutline} %(Outline)
Given knowledge of $\bB^{*}$, we may view
the sensing relation \eqref{eq:sys_sense_relation} as a process
such that $1)$ $\bPi^{*}$ is encoded via the codeword
$\bPi^{*}\bX\bB^{*}$ and $2)$ is passed through a Gaussian
channel with additive noise $\bW$.
We complete the proof based on Fano's inequality
following the
procedure in~\cite{cover2012elements} (cf. Section~$7.9$, P$206$).
The key technical contribution is
the derivation of a tight upper bound on the conditional mutual information between $\bPi^*$ and $\bY$ given
$\bX$ when $\bPi^*$ is drawn uniformly at random from $\calH$.
% The proof details are deferred to Appendix~\ref{thm_proof:oracle_fail_general}.
%$\hfill\Box$
\end{proofoutline}

%Based on Theorem~\ref{thm:exact_minimax}, we conclude
%the wrong permutation recovery to be inevitable provided that \newline
%$\inf_{\wh{\bPi}}\sup_{\bPi^{*}}\Expc\Ind(\wh{\bPi}\neq \bPi^{*}) \geq c_0$,
%which means
%\begin{equation}
%\label{eq:snr_general_lb}
%\logdet\bracket{\bI + \frac{\bB^{*\rmt}\bB^{*}}{\sigma^2}}<\
%%\sum_i \log\left(1 + \dfrac{\lambda_i^2}{\sigma^2}\right) < \
%\frac{2(1-c_0)\log(|\calH|) - 2}{n} \asymp \frac{\log(|\calH|)}{n},
%\end{equation}
%according to~\eqref{eqn:exact_minimax}.
Let us point out important implications of Theorem \ref{thm:exact_minimax}. When $\calH = \calP_n$, we have $\log|\calH| = \log n! \approx n\log n$
and the condition in~\eqref{eqn:exact_minimax} simplifies as
$\logdet(\bI +\bB^{*\rmt}\bB^{*}/\sigma^2 ) \lsim \log n$.
With a smaller set $\calH$,
the inachievability
condition~\eqref{eqn:exact_minimax} is less likely to be fulfilled.
For example, consider the special case in which $\calH$ is a Hamming ball
around the identity, i.e., $\calH = \{\bPi \in \calP_n: \; \dh(\bI; \bPi) \leq \mathsf{D} \}$
for some fixed non-negative integer $\mathsf{D}$. Then the condition in
\eqref{eqn:exact_minimax} reduces to
$\logdet(\bI +\bB^{*\rmt}\bB^{*}/\sigma^2 ) \lsim  \mathsf{D}/(n- \mathsf{D}) \ll \log n$
when $n$ is sufficiently large.

The second major ingredient in condition~\eqref{eqn:exact_minimax} is the
term $\logdet(\bI +\bB^{*\rmt}\bB^{*}/\sigma^2) =
\sum_i \log(1 + \lambda_i^2/\sigma^2)$, where $\lambda_i$ denotes
the $i^{\mathsf{th}}$ singular value of $\bB^{*}$.
Since each singular value $\lambda_i$ is determined by the matrix $\bB^{*}$ as whole rather than by individual columns, we conclude that linear independence among multiple measurements
can positively impact the recovery of $\bPi^{*}$, which
implies extra benefits apart from mere energy accumulation.

When maximizing the term $\sum_i \log\bracket{1 + \lambda_i^2/\sigma^2}$
given fixed signal energy $\fnorm{\bB^{*}}^2 = \sum_i \lambda_i^2$,
it is easy to determine the most favorable configuration to
avoid failure of recovery: the signal energy is evenly
spread over all singular values.
In contrast, if
$\bB^{*}$ has rank one with all signal energy concentrated on the principal
singular value, condition~\eqref{eqn:exact_minimax}
reduces to the same as for a single MV ($m = 1$) with signal energy
$\fnorm{\bB^{*}}^2$ since
\begin{equation}\label{eq:logsignal}
\logdet\bracket{\bI +\frac{\bB^{*\rmt}\bB^{*}}{\sigma^2}}
% \sum_{i}\log\left(1 + \frac{\lambda_i^2}{\sigma^2}\right)
\hspace{-1pt}=\hspace{-1pt}
\log \left(1 + \frac{\lambda_1^2}{\sigma^2}\right) =
\log \left(1 + \frac{\fnorm{\bB^{*}}^2}{\sigma^2}\right).
\end{equation}
This indicates that in accordance with the intuition of ``borrowing strength'' across
different sets of measurements, performance is expected to improve as the stable rank of $\varrho(\bB^{*})$ of $\bB^*$ increases.
To give an illustration of the
benefits brought by large stable rank $\varrho(\bB^{*})$, we
numerically evaluate the required $\snr = \fnorm{\bB^*}^2/(m\sigma^2)$ for the leftmost quantity in \eqref{eq:logsignal} to exceed specific thresholds in dependence of selected choices of $\varrho(\bB^{*})$. The results are listed in Table~\ref{tab:snr_require}.

\newpage

\begin{table}[h!]
\centering
\caption{The required values of $\snr = \fnorm{\bB^*}^2/(m\sigma^2)$ for the condition $\logdet\bracket{\bI +\frac{\bB^{*\rmt}\bB^{*}}{\sigma^2}} > c \cdot \log n$ to hold, $c \in \{ 1,2,\ldots,6\}$, when $n = 1000$, $p = 100$,
and $\bB_{:, i}^{*} = \be_i$, where
$\be_i$ denotes the $i^{\mathsf{th}}$ canonical basis vector, $1 \leq i \leq m$, $m \in \{ 1,10,20,50,100 \}$ (leftmost column).}
\label{tab:snr_require}
\begin{tabular}{l|cccccc}
\noalign{\hrule height 0.7pt}
%\hline
$\frac{\log\det\left(\bI + \bB^{*\rmt}\bB^{*}/\sigma^2\right)}{\log n}$ &
\hspace{-1mm} $1$ & $2$ & $3$ & $4$ & $5$ & $6$ \\ \hline
$\varrho(\bB^{*}) =1$ & $10^3$ & $10^6$ &
$10^9$ & $10^{12}$ &  $10^{15}$ & $10^{18}$\\
$\varrho(\bB^{*}) = 10$& $1$ & $2.98$ & $6.94$ &  $14.85$  & $30.62$   & $62.10$  \\
$\varrho(\bB^{*}) = 20$ & $0.41$ & $1.00$  &  $1.82$ &  $2.98$ & $4.62$ &$6.94$ \\
$\varrho(\bB^{*}) = 50$ & $0.15$ & $0.32$  &  $0.51$ &  $0.74$ & $1.00$ &$1.29$ \\
$\varrho(\bB^{*}) = 100$ & $0.07$& $0.15$  &  $0.23$ &  $0.32$ & $0.41$ &$0.51$ \\
\noalign{\hrule height 0.7pt}
\end{tabular}%\vspace{-0.1in}
\end{table}

The statement below provides a condition for failure of recovery when using a specific
estimator of $\bPi^*$, namely the
ML estimator in~\eqref{eq:sys_ml_estimator}. The latter
is computationally feasible if $\bB^{*}$ is known.
%In this statement, $\bPi^{*}$ is considered as fixed (non-random) as is the case throughout the paper with the exception %of
%Theorem~\ref{thm:exact_minimax} and Corollary~\ref{corollary:oracleapprox}.
%
\begin{proposition}
\label{thm:fail_recover_optim}
Let $\bPi^*$ be an arbitrary element of $\calP_n$. The ML estimator $\wh{\bPi}$ with knowledge of $\bB^*$ given in
(\ref{eq:LAP})
satisfies $\Prob(\wh{\bPi}\neq {\bPi^{*}}) \geq \frac{1}{2}$
for $n\geq 10$ if
\begin{equation}
\label{eqn:fail_recover_rank_assump}
\dfrac{\fnorm{\bB^{*}}^2}{\sigma^2} \leq
\dfrac{2\log n}{4\left(1 + c \varrho^{-1/2}(\bB^{*})\right)^2},
\end{equation}
where $\varrho(\bB^{*}) ={\fnorm{\bB^{*}}^2}/{\opnorm{\bB^{*}}^2}$ is the
\emph{stable rank} of $\bB^{*}$.
\end{proposition}

\begin{proofoutline} Without loss of generality,
we may work with $\bPi^{*} = \bI$. We then use a direct argument involving corresponding rows of $\bY$ and $\bPi^*\bX$ and concentration of measure results to show that if \eqref{eqn:fail_recover_rank_assump} holds, $\wh{\bPi}$ cannot be $\bI$ with the stated probability.
%The technical bottleneck lies in properly
%lower-bounding the error probability $\Prob(\wh{\bPi}\neq {\bPi^{*}})$.
%In our work, we adopt the relation
%\begin{align*}
%\Prob(\wh{\bPi}\neq \bPi^{*})\geq
%\Prob\Bracket{\exists~(i,j),~\textup{s.t.}~\langle \bW_{j, :} - \bW_{i, :},~\bB^{*\rmt}(\bX_{i, :} - \bX_{j, :})\rangle %\geq \norm{\bB^{*\rmt}\bracket{\bX_{i, :} - \bX_{j, :}} }{2}^2},
%\end{align*}
%and lower bound the right-hand side by
%$1/2$ assuming~\eqref{eqn:fail_recover_rank_assump}.
\end{proofoutline}

The proposition states that the total signal energy given by
$m\cdot\snr$ should
be at least
of the order $\log n$ to avoid failure in recovery.
This is in agreement with
Theorem~\ref{thm:exact_minimax}
in the full-rank case.

\subsection{Approximate recovery of $\bPi^{*}$}
Exact recovery of $\bPi^{*}$ may not always be necessary. The following corollary of
Theorem~\ref{thm:exact_minimax} yields a condition under which
even approximate recovery of $\bPi^*$ within Hamming distance $ \mathsf{D}$, i.e.,
$\dh(\bPi^{*}; \wh{\bPi})\leq  \mathsf{D}$,  cannot be guaranteed.
%For the clarify of presentation, we assume
%$\calH$ contains
%no prior knowledge of $\bPi^{*}$ and
%consider $\bPi^{*}$ to be uniformly distributed among the set of
%all possible permutation matrices, i.e.,
%$|\calH| = n!$.
Specifically, we state the following corollary of Theorem \ref{thm:exact_minimax}.

\begin{corollary}
\label{corollary:oracleapprox}
Assuming that $\bB^{*}$ is known, we have
\begin{equation}\label{eq:fail_approx_recover}
\inf_{\wh{\bPi}}\sup_{\bPi^{*} \in \calP_n} \Expc \Ind(\dh(\wh{\bPi};~\bPi^{*}) \geq  \mathsf{D}) \geq \frac{1}{2}
\qquad \text{if} \;\;\, \logdet\bracket{\bI + \frac{\bB^{*\rmt}\bB^{*}}{\sigma^2}}
\leq \dfrac{\log(n- \mathsf{D}+1)! - \log 4}{n},
\end{equation}
%\begin{align}
%\inf_{\wh{\bPi}}\sup_{\bPi^{*}}
%\Expc \Ind(\dh(\wh{\bPi},~\bPi^{*}) \geq D) \geq
%1- \dfrac{({n}/{2})\logdet\bracket{\bI + \bB^{*\rmt}\bB^{*}/\sigma^2}+
%\log 2}{\log (n-D+1)!},
%\end{align}
where the expectation $\Expc$ is taken w.r.t. $\bX$ and $\bW$,
and the infimum is over all estimators $\wh{\bPi}$. %$\lambda_i$ denotes the $i^{\mathsf{th}}$ singular value of %$\bB^{*}$.
\end{corollary}

%\begin{proofoutline}
%Its proof procedure is similar to what is used in
%Theorem~\ref{thm:exact_minimax}. The gist is to
%prove the inequality
%\begin{align*}
%\Prob\bracket{\dh(\wh{\bPi}; \bPi^{*}) \geq  \mathsf{D}} \leq
%\dfrac{({n}/{2})\logdet\bracket{\bI + \bB^{*\rmt}\bB^{*}/\sigma^2}+
%\log 2}{\log (n- \mathsf{D}+1)!},
%\end{align*}
%which is less than $1/2$ given~\eqref{eq:fail_approx_recover}.
%\end{proofoutline}

%According to~\eqref{eq:fail_approx_recover}, we cannot consistently
%obtain the desirable approximate recovery $\wh{\bPi}$ for all possible $\bPi^{*}$
%provided that
%\[
%\dfrac{1}{2}\logdet\bracket{\bI + \frac{\bB^{*\rmt}\bB^{*}}{\sigma^2}} +
%\dfrac{\log 2}{n} \leq \dfrac{\log(n-D+1)!}{2n}.
%\]
Comparing the above result with Theorem~\ref{thm:exact_minimax}, one can see that
the essentially only difference is the replacement of the term
$\log\abs{\calH}$ by $\log (n- \mathsf{D}+1)!$.
An intuitive interpretation is as follows:

\begin{itemize}
\item

The set of $n$-by-$n$ permutation matrices under consideration  can be covered by a subset\\ $\{\bPi^{(1)}, \bPi^{(2)}, \cdots,~\bPi^{((n- \mathsf{D}+1)!)}\}$  such that for any permutation matrix $\bPi$, there exists an element\\
$\bPi^{\dagger} \in \{\bPi^{(1)}, \bPi^{(2)}, \cdots,~\bPi^{((n- \mathsf{D}+1)!)}\}$ such that
$\dh(\bPi;~\bPi^{\dagger}) \leq  \mathsf{D}$.

\vspace{0.1in}

\item
We would like to recover
$\bPi^{\dagger}$ from data $(\bX,\bY)$. \\
\end{itemize}

Consequently, since the cardinality of the covering is
$(n- \mathsf{D}+1)!$, we encounter the term $\log(n -  \mathsf{D} + 1)!$
in place of $\log\abs{\calH}\leq \log(n!)$; setting $ \mathsf{D} = 0$ or $1$ gives back
Theorem~\ref{thm:exact_minimax}.
Additionally, we can obtain

\noindent a lower bound on the minimax risk with respect to $\dh(\cdot; \cdot)$ effortlessly from the proof of
Corollary~\ref{corollary:oracleapprox}, as
\begin{align}
\label{eq:dh_minimax}
\inf_{\wh{\bPi}}\sup_{\bPi^{*}}\Expc\dh(\wh{\bPi}; \bPi^{*})
\geq \max_{d \in \{0,1,\ldots,n\}}  \; \, (d+1)\bracket{1-\dfrac{({n}/{2})
\logdet\bracket{\bI + \bB^{*\rmt}\bB^{*}/\sigma^2}+
\log 2}{\log (n-d+1)!}}.
\end{align}

\vspace{0.1in}

A unified proof for~\eqref{eq:fail_approx_recover} and~\eqref{eq:dh_minimax}
can be found in
Appendix~\ref{corollary_proof:oracleapprox}. To an extent, Eq.~(12) strengthens the assertion of Theorem \ref{thm:exact_minimax} in the sense that if $\logdet\bracket{\bI + \bB^{*\rmt}\bB^{*}/\sigma^2} \ll \log n$, $\wh{\bPi}$ can be rather far from recovering $\bPi^*$ in the sense that $\dh(\wh{\bPi}; \bPi^{*}) = \Omega(n)$.

To conclude this section, we would like to emphasize that the
above conditions reflect the price to compensate for the
uncertainty induced by the sensing noise $\bW$, as
there is no uncertainty in $\bB^{*}$ involved.
%In the next section,
%we will study conditions for the successful recovery.

\section{Successful Recovery}\label{sec:prac_recover}
In the previous section, we have studied conditions under which recovery is expected to fail.
In this section, we state
conditions under which the true permutation $\bPi^{*}$ can be recovered
with high probability, for both the oracle case with known $\bB^{*}$ as well as the  ``realistic case" with unknown $\bB^{*}$.
For the conciseness of presentation, we hide explicit values for numerical constants in most cases and provide them in the
appendix for interested readers.
We believe that those values can be improved further
since no specific effort was made to obtain optimal constants.

\subsection{Oracle case: known $\bB^{*}$}
As previously mentioned, in this case the ML estimator in~\eqref{eq:sys_ml_estimator}
is given by~\eqref{eq:LAP}.
The condition on the $\snr$ in the following statement can serve both as an
upper bound for the failure of permutation recovery and as a lower bound
for the more challenging case with unknown $\bB^{*}$.

\begin{theorem}
\label{thm:oracle_succ_optim}
Given the knowledge of $\bB^{*}$, if
\begin{align}
\label{eqn:oracle_succ_optim_assump}
\log\bracket{\frac{\fnorm{\bB^{*}}^2}{\sigma^2}} \geq \dfrac{8\log n}{\kappa \varrho(\bB^{*})} + \log\bracket{\kappa\varrho(\bB^{*})\vcup \alpha_1\log n }
+ \alpha_2,
\end{align}
then the ML estimator in~\eqref{eq:sys_ml_estimator} satisfies
\[
\Prob(\hat{\bPi} \neq \bPi^{*})
\leq \dfrac{2\alpha_0^{2 \kappa \varrho(\bB^{*})}}{n^2} \
\stackrel{(\alpha_0 < 1)}{<} \frac{2}{n^2},
\]
where $0 < \alpha_0 < 1$, $\kappa > 0$ are universal constants,
$\alpha_1 = 2/\log(\alpha_0^{-1})$, and
$\alpha_2 = \log(64\alpha_0^{-4}\log \alpha_0^{-1})$.
\end{theorem}

\begin{proofoutline} We show that each row of $\bY$ is closest in
Euclidean distance to its matching row in $\bX \bB^*$ with the stated
probability, which implies the desired event $\{ \wh{\bPi}\neq \bPi^{*} \}$. The key ingredient is a careful probabilistic lower bound on the minimum distance between any pairs of rows in $\bX \bB^*$ based on a small ball probability result
in high-dimensional geometry due to Latala et al.~\cite{latala2007banach}.
\end{proofoutline}

To illustrate the tightness of conditions in~\eqref{eqn:oracle_succ_optim_assump}, we would like to consider two special cases for $\bB^*$, namely, the full-rank case and the rank-one case,
and compare it with the condition for failure of recovery in
Theorem~\ref{thm:exact_minimax}.
First, we consider the full-rank case with constant singular values,
i.e., $\bB^{*\rmt}\bB^{*} = \gamma \bI$, where $\gamma > 0$ is a positive constant; in particular,
$\varrho(\bB^{*}) = m$. Then a simple term re-arrangement of~\eqref{eqn:oracle_succ_optim_assump} suggests that having
\begin{equation}\label{eq:cond_success}
\log\left(\frac{\fnorm{\bB^{*}}^2}{\varrho(\bB^{*}) \sigma^2} \right) = \log\left(\frac{\fnorm{\bB^{*}}^2}{m \sigma^2} \right) = \log\bracket{\frac{\gamma}{\sigma^2}} \gsim \frac{\log n}{\varrho(\bB^{*})}
\end{equation}
ensures success, while Theorem~\ref{thm:exact_minimax} suggests
\begin{equation}
\label{eq:cond_failure}
\log\bracket{1 + \dfrac{\fnorm{\bB^{*}}^2}{m\sigma^2}} = \log\bracket{1+\frac{\gamma}{\sigma^2}}\lsim \dfrac{\log n}{\varrho(\bB^{*})}
\end{equation}
implies failure. Conditions~\eqref{eq:cond_success} and~\eqref{eq:cond_failure}
thus match up to multiplicative factors.  \\

Next, we consider the rank-one case. Without loss of generality, we set $\bB^* = \bB^{*}_{:, 1} = \gamma \be_1$.
Theorem~\ref{thm:oracle_succ_optim} (cf.~\eqref{eqn:oracle_succ_optim_assump})
suggests that
$\log\bracket{\frac{\fnorm{\bB^{*}}^2}{\sigma^2}} = \log\bracket{\frac{\gamma}{\sigma^2}}\gsim \log n$ ensures success, while
Theorem~\ref{thm:exact_minimax} suggests
that $\log\bracket{1 + \frac{\fnorm{\bB^{*}}^2}{\sigma^2}} =
\log\bracket{1 + \frac{\gamma}{\sigma^2}} \lsim \log n$ leads to failure.
Putting things together, we conclude tightness for this case.
For a more clear view of
\eqref{eqn:oracle_succ_optim_assump},
we omit the non-dominant terms, and state the implications in the following remark.

%===============================
\begin{remark*}
%\label{thm:oraclesuccoptim}
Consider the oracle case with known $\bB^{*}$.
The ML estimator of $\bPi^*$ achieves permutation recovery with
high probability (w.h.p.) in the following situations:
\begin{align*}
\{ \wh{\bPi} = \bPi^*\} \; \text{holds w.h.p. if} \; \;\begin{cases}
&\kappa\varrho(\bB^{*}) < \alpha_1 \log n \; \,\text{and} \; \,\log(\snr) \geq \frac{(8+\kappa)\log n}{\kappa \varrho(\bB^{*})} + c_0,  \\
&\kappa\varrho(\bB^{*}) \geq \alpha_1 \log n \;\, \text{and} \; \,\snr \geq c_1.
\end{cases}
\end{align*}
%\begin{itemize}
%\item
%\textbf{Case I} -- $\kappa\varrho(\bB^{*}) < \alpha_1 \log n$:  $\{ \wh{\bPi} = \bPi^*\}$ w.h.p. %if
%\begin{align}
%\label{eqn:oracle_succ_optim_assump1}
%$\log(\snr) \geq $,
%\end{align}

%\item
%\textbf{Case II} -- $\kappa\varrho(\bB^{*}) \geq \alpha_1 \log n$:  $\{ \wh{\bPi} = \bPi^*\}$ %w.nh.p. if
%$\snr \geq c_1$.
%\end{itemize}

% then the ML estimator in~\eqref{eq:sys_ml_estimator} gives
% correct recovery with high probability, namely, \[
% \Prob\bracket{\hat{\bPi} \neq \bPi^{*}}
% \leq \dfrac{2\alpha_0^{2 \kappa \varrho(\bB^{*})}}{n^2},
% \]
% where $0 < \alpha_0 < 1$, $\kappa > 0$ are universal constants,
% $\alpha_1 = 2/\log(\alpha_0^{-1})$, $c_0 = \log(64\alpha_1\alpha_0^{-4}\log \alpha_0^{-1}) $,
% and $c_1 =\log(64\kappa\alpha_0^{-8}\log \alpha_0^{-1})$.
\end{remark*}

% \begin{proof}
% The proof of
% Theorem~\ref{thm:oraclesuccoptim} can be obtained effortless with
% Lemma~\ref{lemma:oracle_succ_optim}.
% In Case I, we assume $\kappa\varrho(\bB^{*}) < \alpha_1 \log n$.
% The condition~\eqref{eqn:oracle_succ_optim_assump} then reduces to
% \[
% \log(\snr) \geq  \dfrac{8\log n}{\kappa \varrho(\bB^{*})} +
% \log\bracket{\frac{\log n}{m}} + \log \alpha_1 + \alpha_2.
% \]
% With the relation $\log(\log n/m) \leq \log n/m \leq \log n/\varrho(\bB^{*})$,
% we can upper-bound
% $8\log n/(\kappa \varrho(\bB^{*})) + \log(\log n/m)$ by
% $(8+\kappa)\log n/(\kappa \varrho(\bB^{*}))$ .
% Meanwhile in Case II, we let $\kappa\varrho(\bB^{*}) \geq \alpha_1 \log n$
% and upper-bound~\eqref{eqn:oracle_succ_optim_assump} as
% \[
%  \dfrac{8\log n}{\kappa \varrho(\bB^{*})} + \log\bracket{\kappa\varrho(\bB^{*})\vcup \alpha_1\log n } \leq \
% 8\alpha_1^{-1} + \log \kappa + \log m.
% \]
% Hence our task transforms to showing Lemma~\ref{lemma:oracle_succ_optim},
% whose proof can be found in
% Section~\ref{lemma_proof:oracle_succ_optim}.
% \end{proof}

In summary, Theorem~\ref{thm:oracle_succ_optim} yields a more relaxed
requirement on the $\snr$ needed to recover $\bPi^{*}$
as the stable rank $\varrho(\bB^{*})$ exceeds a certain threshold.
More specifically, the requirement becomes $\snr \geq c_1$ for some positive constant
$c_1$ (in particular, the right hand side does not grow with $n$) once $\kappa \rho(\bB^{*}) \gsim \log n$.

\subsection{Realistic case: unknown $\bB^{*}$}
For this case with $\bB^{*}$ unknown, we first present a basic result in Theorem~\ref{thm:succ_recover} that will be improved
upon later under additional assumptions.

\begin{theorem}
\label{thm:succ_recover}
Let $\epsilon > 0$ be arbitrary, and suppose that $n > N_1(\epsilon)$, where $N_1(\epsilon) > 0$ is a positive
constant depending
only on $\epsilon$.
Provided the following conditions hold:
$(i)$ $\snr \cdot n^{-\frac{2n}{n-p}} \geq 1$; and $(ii)$
\begin{align}
\label{eq:succ_snr_condition}
\log (m\cdot \snr) & \geq
\bracket{c_0 + c_1\epsilon}\log n,  %\\
% \log(m\cdot \snr) & \geq
% \bracket{\bracket{c_0 + c_1\epsilon}\log n} \vcup \bracket{c_2\log m};
% 380\bracket{\frac{3}{2} + \epsilon + \frac{n}{190(n-p)}}\log n,
\end{align}
then the ML estimator in~\eqref{eq:sys_ml_estimator}
gives the ground-truth permutation matrix $\bPi^{*}$ with probability
exceeding $1-c_3 n^{-\epsilon}\Bracket{(n^{\epsilon} - 1)^{-1} \vcup 1}$,
where $c_0, c_1, c_2, c_3> 0$ are fixed positive constants. 	
\end{theorem}

\begin{proofoutline} The proof extends the proof strategy employed in
Pananjady et al.~\cite{pananjady2016linear} for $m = 1$ to arbitrary $m$, which amounts
to showing that
%Given a fixed permutation matrix $\bPi$, the ML estimator
%estimates $\bB$ as $\bracket{\bPi \bX}^{\dagger}\bY$. Back-subsitution then gives us $\wh{\bPi} = %\argmax \fnorm{P^{
%\perp}_{\bPi\bX}\bY}$. The proof idea is by exhaustively searching
%$\bPi\neq \bPi^{*}$,  to show  that
$\fnorm{P^{\perp}_{\bPi^{*}\bX}\bY} <
\fnorm{P^{\perp}_{\bPi\bX}\bY}$ holds true for all $\bPi \neq \bPi^{*}$, where
for $\bPi \in \calP_n$, $P_{\bPi \bX}^{\perp}$ denotes the projection on the orthogonal
complement of the range of $\bPi \bX$. Several critical steps in \cite{pananjady2016linear} do no longer apply for the case of multiple $m$ considered herein, and prompt new technical challenges
to be overcome.
%holds with high probability.
Details of the proof are available in
Appendix~\ref{thm_proof:succ_recover}, including specific values
of the constants $c_0$ and $c_1$.
\end{proofoutline}

Theorem~\ref{thm:succ_recover} states that
exact recovery of $\bPi^*$ can be achieved with
high probability if $\log(m \cdot \snr) \gsim \log n$.
For the rank-one case,
we can see that this result is tight up to multiplicative constants in light
of Theorem~\ref{thm:exact_minimax},
which implies failure of recovery with high probability provided that
$\log(1 + m\snr)\lsim \log n$.
However, Theorem~\ref{thm:succ_recover}
suggests that multiple measurements behave like a single
measurement with the same energy level, which can be far from the actual behavior beyond
the rank-one case.
Unlike Theorem ~\ref{thm:oracle_succ_optim} concerning the oracle case, Theorem~\ref{thm:succ_recover} thus fails to capture potential improvement brought by higher
measurement diversity as quantified by the
stable rank $\varrho(\bB^{*})$. To address this limitation, we present a refined result that comes at the expense of
additional assumptions on $\dh(\bI; \bPi^{*})$ and $\varrho(\bB^{*})$.
\begin{theorem}
\label{thm:succ_recover_refine}
Suppose that $\dh(\bI; \bPi^{*}) \leq h_{\max}$ with $h_{\max}$ satisfying the relation
$h_{\mathsf{max}}r(\bB^{*}) \leq n/8$. Let further $\epsilon > 0$ be arbitrary, and suppose that
$n > N_2(\epsilon)$, where $N_2(\epsilon) > 0$ is a positive constant depending
only on $\epsilon$.
In addition, suppose that the following conditions hold:
%$(i)$ $$, $(ii)$
%$$,
%$(iii)$ $h_{\mathsf{max}}r(\bB^{*}) \leq n/8$,
%and $(iii)$
\begin{equation}
\label{eq:snr_refine_cond}
(i)\;\snr > c_0, \quad (ii)\;\varrho(\bB^{*})\geq  c_1(1+\epsilon)\log n, \quad (iii)\;\log\bracket{\snr}\geq \frac{c_3(1+\epsilon)\log n}{\varrho(\bB^{*})} + c_4.
\end{equation}
Then the ML estimator \eqref{eq:sys_ml_estimator} subject to the constraint $\dh(\bI; \bPi)\leq h_{\mathsf{max}}$
equals $\bPi^{*}$ with probability at least
$1 - 10n^{-\epsilon}\Bracket{(n^{\epsilon} - 1)^{-1}\vcup 1}$,
where $c_0, \ldots, c_4 > 0$ are some positive constants.
\end{theorem}

\begin{remark}
The $\snr$ requirement in \eqref{eq:snr_refine_cond} matches
the minimax bound in Theorem~\ref{thm:exact_minimax} up
to a logarithmic factor, when setting $\calH$ as $\{\bPi \in \calP_n:\dh(\bI; \bPi) \leq h_{\mathsf{max}}\}$ and
$h_{\mathsf{max}}\asymp \frac{n}{\log n}$.
\end{remark}

In contrast to Theorem~\ref{thm:succ_recover}, the above theorem uncovers the benefits
brought by larger stable rank $\varrho(\bB^{*})$. The outline of the
proof strategy is given as follows with the technical details being placed in
Section~\ref{thm_proof:succ_recover_refine}.

\begin{proofoutline}
The structure of the proof is similar to that of Theorem~\ref{thm:succ_recover}. The
main innovation is an improved upper bound on the probability
$\Prob\bracket{\fnorm{P_{\bPi^{*} \bX}^{\perp}\bX\bB^{*}} \leq c\fnorm{\bB^{*}}^2}$,
where $P_{\bPi^{*} \bX}^{\perp}$ denotes the projection on the orthogonal complement
of the range of $\bPi^{*} \bX$. The above probability is bounded based on an
$\varepsilon$-covering of the set of sparse unit vectors whose relevance is a consequence of the constraint $\dh(\bI; \bPi^{*}) \leq h_{\max}$.
\end{proofoutline}

Let us comment on the additional constraint (i)
$\dh(\bI; \bPi^{*}) \leq h_{\max}$ in Theorem \ref{thm:succ_recover_refine}. To ensure that signal diversity as quantified by
$\varrho(\bB^{*})$ improves the
recovery performance, we require $\varrho(\bB^{*})= \Omega\bracket{\log n}$. In this case,  we obtain the condition $\snr \geq C$ for
some constant $C > 0$, which then also matches the assertion in  Theorem~\ref{thm:oracle_succ_optim}.
At the same time, $h_{\mathsf{max}}$ is required to be of the order
$h_{\mathsf{max}} \lsim \frac{n}{\log n}$,
which is only slightly sub-optimal compared to $h_{\max}$ being a linear
fraction of $n$.
We hypothesize that the constraint on $\dh\bracket{\bI; \bPi^{*}}$ can be eliminated, either with the help of more advanced proof techniques or by imposing more stringent constraints involving the ratio $n / p$.
% \tcr{
% If the hypothesis is true,
% the $\snr$ requirement in \eqref{eq:snr_refine_cond} matches
% the minimax bound in Theorem~\ref{thm:exact_minimax} up to some multiplicative constants, when setting $\calH$ as $\{\bPi \in \calP_n: \; \dh(\bI; \bPi) \leq h_{\mathsf{max}}\}$. Otherwise,
% Theorem~\ref{thm:succ_recover_refine} matches
% the minimax bound up to a logarithmic factor.
% }

Since the order for the required $\snr$ to achieve correct recovery remains the same as in Theorem~\ref{thm:oracle_succ_optim},
we can draw the conclusion that the major difficulty in
recovering $(\bPi^{*}, \bB^{*})$ is due to the
sensing noise $\bW$
while the fact that $\bB^{*}$ is not given a priori
% the ignorance of $\bB^{*}$
does not change the level of difficulty significantly.

\section{Computational Approach}\label{sec:comp_mtd}
In this section, we focus on the computational aspects of the problem.
Recall that for the oracle case, the ML estimator in
\eqref{eq:sys_ml_estimator} reduces to the linear assignment problem
and can be solved efficiently by the Hungarian
algorithm~\cite{kuhn1955hungarian}
or the auction algorithm~\cite{bertsekas1992forward}.
Hence the emphasis in this sections is on the realistic case
with $\bB^{*}$ unknown.
As proved in~\cite{pananjady2016linear}, the computation of the ML estimator is NP-hard
except for the special case $m= p=1$.
In this paper,
we reformulate~\eqref{eq:sys_ml_estimator} as a bi-convex problem
by introducing an auxiliary variable,
and
solve the resulting problem with the ADMM algorithm~\cite{boyd2011distributed}; details
are given in Algorithm~\ref{alg:compt_mtd_admm}.
Numerical experiments based on that algorithm confirm the behavior predicted by
Theorem~\ref{thm:succ_recover_refine}. %and the detailed descriptions
%can be found in Algorithm~\ref{alg:compt_mtd_admm}.

\begin{algorithm}[h!]
\caption{ADMM algorithm for the recovery of $\bPi$.}
\label{alg:compt_mtd_admm}
\begin{algorithmic}[1]
\Statex \textbullet~
\textbf{Input}: observations $\bY$,
sensing matrix $\bX$, penalty parameter
$\rho > 0$, and maximum iteration number
$T_{\mathsf{max}}$.

\Statex \textbullet~\
\textbf{Initialization:}
Initialize $\bPi_1^{(0)},\bPi_2^{(0)}$ as $\bPi^{(0)} = \argmax \bracket{\la \br{\bY}, \bPi \br{\bX}\ra}^2$,
where  $\br{\bY} = m^{-1}\bracket{\sum_{i=1}^m \bY_{:, i}}$
and $\br{\bX} = p^{-1}\bracket{\sum_{i=1}^p \bX_{:, i}}$.

\Statex \textbullet~
\textbf{For $t$ from $0$ to $(T_{\mathsf{max}}-1)$:}
Update $\bPi_1^{(t+1)}, \bPi_2^{(t+1)}$ as
\begin{align}
\label{eq:num_simul_admm}
\bPi_1^{(t+1)} =&\argmin_{\bPi_1}
\la \bPi_1,- \bY\bY^{\rmt}\bPi^{(t)}_2 \Proj_{\bX}^{\rmt} +
\bmu^{(t)} - \rho \bPi^{(t)}_2\ra \notag \\
\bPi_2^{(t+1)} =&
\argmin_{\bPi_2}\la \bPi_2,~\bY\bY^{\rmt} \bPi^{(t+1)}_1 \Proj_{\bX}
- \bmu^{(t)} - \rho\bPi^{(t+1)}_1\ra  \notag \\
\bmu^{(t+1)} =&~ \bmu^{(t)} +
\rho\left(\bPi^{(t+1)}_1 - \bPi_2^{(t+1)}\right),\notag
\end{align}$\quad$ where $\Proj_{\bX} \defequal \bX\bracket{\bX^{\rmt}\bX}^{-1}\bX^{\rmt}$ is the projection onto the
column space of $\bX$.
\Statex \textbullet~
\textbf{Termination}: Stop the ADMM algorithm once $\bPi_1^{(t+1)}$
is identical to $\bPi_2^{(t+1)}$; or $t = T_{\mathsf{max}}$.

\end{algorithmic}
\end{algorithm}

\subsection{Initialization}
\label{subsec:compt_mtd_sort}
%===============================
This subsection provides a heuristic justification of the
initialization method.
To begin with, we consider the special case when $p = 1$; note that the matrices $\bX$ and $\bB^*$ then reduce to a column vector and to a row vector, respectively.
The basic idea is to generalize the sorting method
in~\cite{pananjady2016linear} applicable for $p=m=1$.
For the case $m\geq 2$, we estimate
$\bB_{:, i}^*$ as ${\la \bPi \bX, \bY_{:, i} \ra}/{\ltwonorm{\bX}^2}$
for any fixed $\bPi$.
Back-substitution yields
\[
\fnorm{\bY - \bPi \bX\bB}^2 =
\sum_{i=1}^m \bracket{\ltwonorm{\bY_{:, i}}^2 - \dfrac{\left(\la \bY_{:, i}, \bPi \bX \ra \right)^2}{\ltwonorm{\bX}^2}},
\]
which means that we can recover $\wh{\bPi}$ by
\begin{equation}
\label{eq:compt_mtd_max_coherence}
\wh{\bPi}\in \argmax_{\bPi}~
\sum_{i=1}^m {\la \bY_{:, i},~\bPi \bX\ra}^2.
\end{equation}
Given that the all entries of $\bB^{*}$ are non-negative, we assume
\begin{equation*}
\la \bY_{:, i}, \bPi\bX \ra \approx \Expc \la \bY_{:, i}, \bPi\bX \ra =
\Expc\la \bPi^{*}\bX, \bPi \bX\ra  \bB^{*}_{:, i} \geq 0, \quad 1 \leq i \leq m,
\end{equation*}
and relax~\eqref{eq:compt_mtd_max_coherence} to
\begin{equation}
\label{eq:compt_mtd_simple_sort}
\sum_{i=1}^m {\left|\la \bY_{:, i},~\bPi \bX\ra\right|}^2 \leq
m^2{\la \br{\bY},~\bPi \bX \ra}^2,
\end{equation}
where $\br{\bY} = m^{-1}\bracket{\sum_{i=1}^m \bY_{:, i}}$.
We then compute an estimator $\wh{\bPi}$ by maximizing the
above upper bound in~\eqref{eq:compt_mtd_simple_sort},
which can be formulated as a linear assignment problem
as in~\eqref{eq:LAP}.
The computational complexity of this estimator is
$\Omega(nm + n\log n)$, since only averaging and
sorting are needed~\cite{cormen2009introduction}.

As for the case when $p\geq 2$, we replace
$\bX$ in~\eqref{eq:compt_mtd_simple_sort} by its average, i.e., $\br{\bX} = p^{-1}\bracket{\sum_{i=1}^p \bX_{:, i}}$.
Hence we obtain the starting point $\wh{\bPi}^{(0)}$ for
the ADMM scheme in Algorithm~\ref{alg:compt_mtd_admm}.

\subsection{Update equation}
%We relax the ML estimation problem~\eqref{eq:sys_ml_estimator} to a bi-convex problem and solve it via an ADMM algorithm
% (c.f.~Alg.~\ref{alg:compt_mtd_admm}).
This subsection presents the  detailed derivation of the
update equation of the ADMM estimator in Algorithm~\ref{alg:compt_mtd_admm}.
First, we have the following relation in light of~\eqref{eq:sys_b_estimator}:
\begin{equation}
\label{eq:num_simul_p2_ml}
\min_{\bPi, \bB}~\fnorm{\bY - \bPi \bX \bB}^2 =
\min_{\bPi}\fnorm{\Proj_{\bPi \bX}^{\perp}\bY}^2,
\end{equation}
where $\Proj_{\bPi \bX}^{\perp}$ is defined as
$\bI - \bPi \bX \left(\bX^{\rmt}\bX\right)^{-1}\bX^{\rmt}\bPi^{\rmt}$. Note that we can decompose $\bY$ as
$\Proj_{\bPi \bX}^{\perp} \bY + \Proj_{\bPi \bX} \bY$. Since $\fnorm{\bY}^2 = \fnorm{\Proj_{\bPi \bX}^{\perp} \bY}^2 +
\fnorm{\Proj_{\bPi \bX} \bY}^2$ can be treated as a constant,
minimizing $\fnorm{\Proj_{\bPi \bX}^{\perp} \bY}^2$ is equivalent
to maximizing $\fnorm{\Proj_{\bPi \bX}\bY}^2$.
Introducing two redundant variables $\bPi_1$ and $\bPi_2$,
we formulate~\eqref{eq:num_simul_p2_ml} as
\begin{equation}\label{eq:num_simul_biconvex}
\min_{\bPi_1,~\bPi_2} -
\trace\left( \bPi_1 \Proj_{\bX}
\bPi_2^{\rmt} \bY\bY^{\rmt}\right),~
\St \bPi_1 = \bPi_2,
\end{equation}
where $\Proj_{\bX} \defequal \bX\bracket{\bX^{\rmt}\bX}^{-1}\bX^{\rmt}$.
We propose to solve~\eqref{eq:num_simul_biconvex} with
the ADMM algorithm~\cite{boyd2011distributed}, whose scheme then yields Algorithm~\ref{alg:compt_mtd_admm}.

\section{Numerical Results}\label{sec:num_simul}
In this section, we present simulation results
and investigate the relation between the
correct recovery rate $\Prob(\wh{\bPi} = \bPi^{*})$
and the signal energy.
The experiments are
divided into two parts:
1) the oracle case (known $\bB^{*}$) and 2) the realistic case with $\bB^{*}$ being unknown.

\subsection{Oracle case}

In this subsection, we study the relation between the
correct probability $\Prob(\wh{\bPi}= \bPi^{*})$
and the $\snr$ in the oracle case with known $\bB^{*}$.
As mentioned previously, the ML estimator is obtained as the
solution of the linear assignment problem \eqref{eq:LAP}. The latter is here solved by the auction algorithm~\cite{bertsekas1992forward}.
% which is shown in
% Fig.~\ref{fig:oracle_recover_rank1} and Fig.~\ref{fig:oracle_recover_fullrank}.
The simulation results confirm
our theoretical results in
Theorem~\ref{thm:exact_minimax} and
Proposition~\ref{thm:fail_recover_optim}.
In virtue of Theorem~\ref{thm:exact_minimax}, we plot
\[
\frac{\log \det(\bI + \bB^{*\rmt}\bB^{*}/\sigma^2)}{\log n} =
\frac{\sum_i \log\left(1 + {\lambda_i^2}/{\sigma^2}\right)}{\log n}
\]
on the horizontal axis, and the empirical probability of
permutation recovery
on the vertical axis. We also use
$\snr$ on the horizontal axis to illustrate the
energy savings
brought by multiple measurement vectors.

\paragraph{Rank-one case}
We use $\bB^{*}$ such that all $\{ \bB_{:, i}^{*} \}_{i=1}^m$ are identical. The
simulation results are displayed in Fig.~\ref{fig:oracle_recover_rank1}.
The left panels use $\snr$ for the horizontal axis while
the right panels show the corresponding values of
$\frac{\logdet\bracket{\bI + \bB^{*\rmt}\bB^{*}/\sigma^2}}{\log n}$,
which can be rewritten as
$\frac{\log(1 + m\snr)}{\log n}$ in this case. Observing that the
curves coincide in the right panels, we conclude that
increasing values of $m$ are irrelevant in this case
given a fixed ratio of the total signal energy to the noise
variance $\fnorm{\bB^{*}}^2/\sigma^2$.

\begin{figure}[t]
\begin{center}
\mbox{
\includegraphics[width=2.1in]{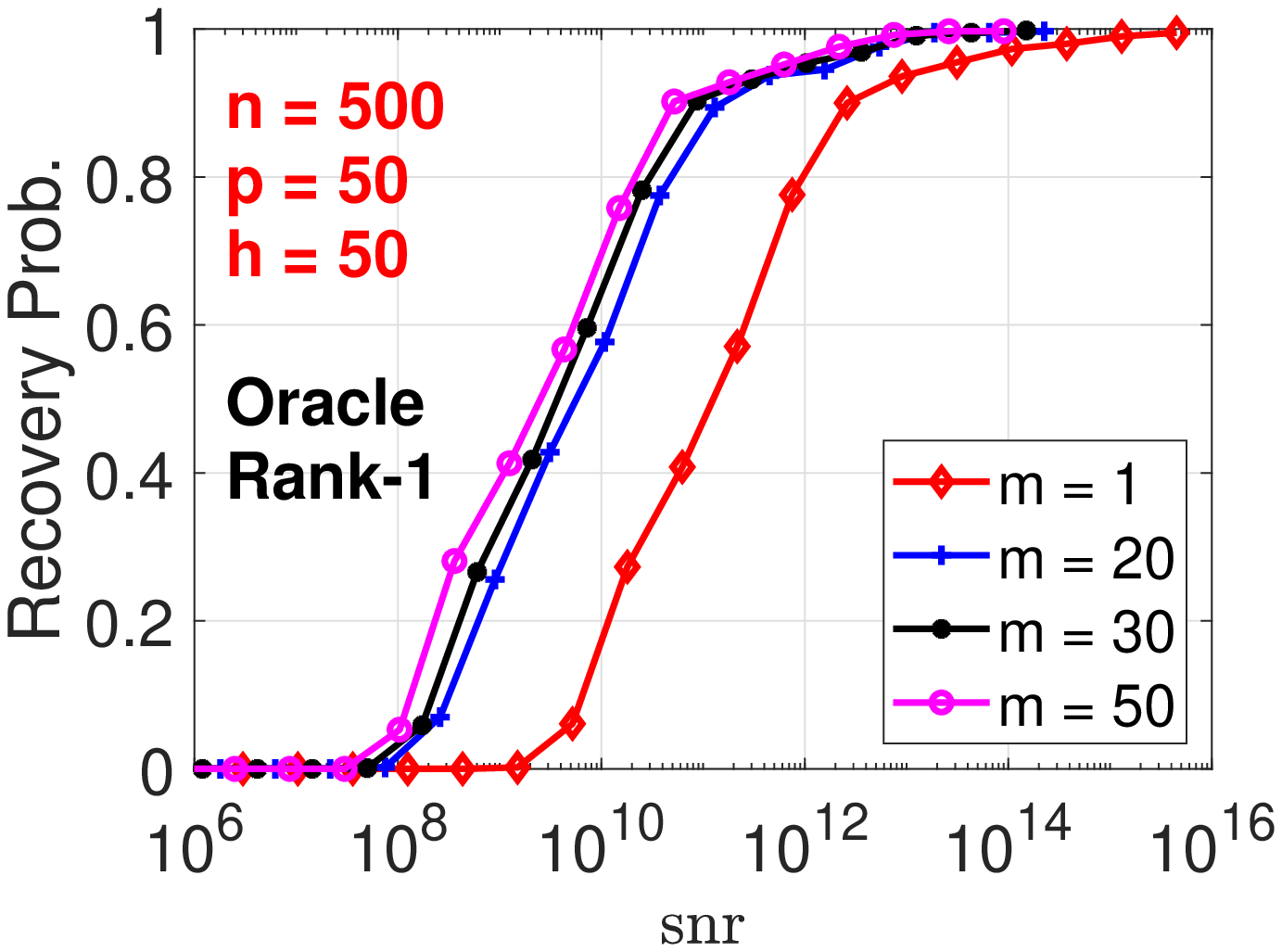}
\includegraphics[width=2.1in]{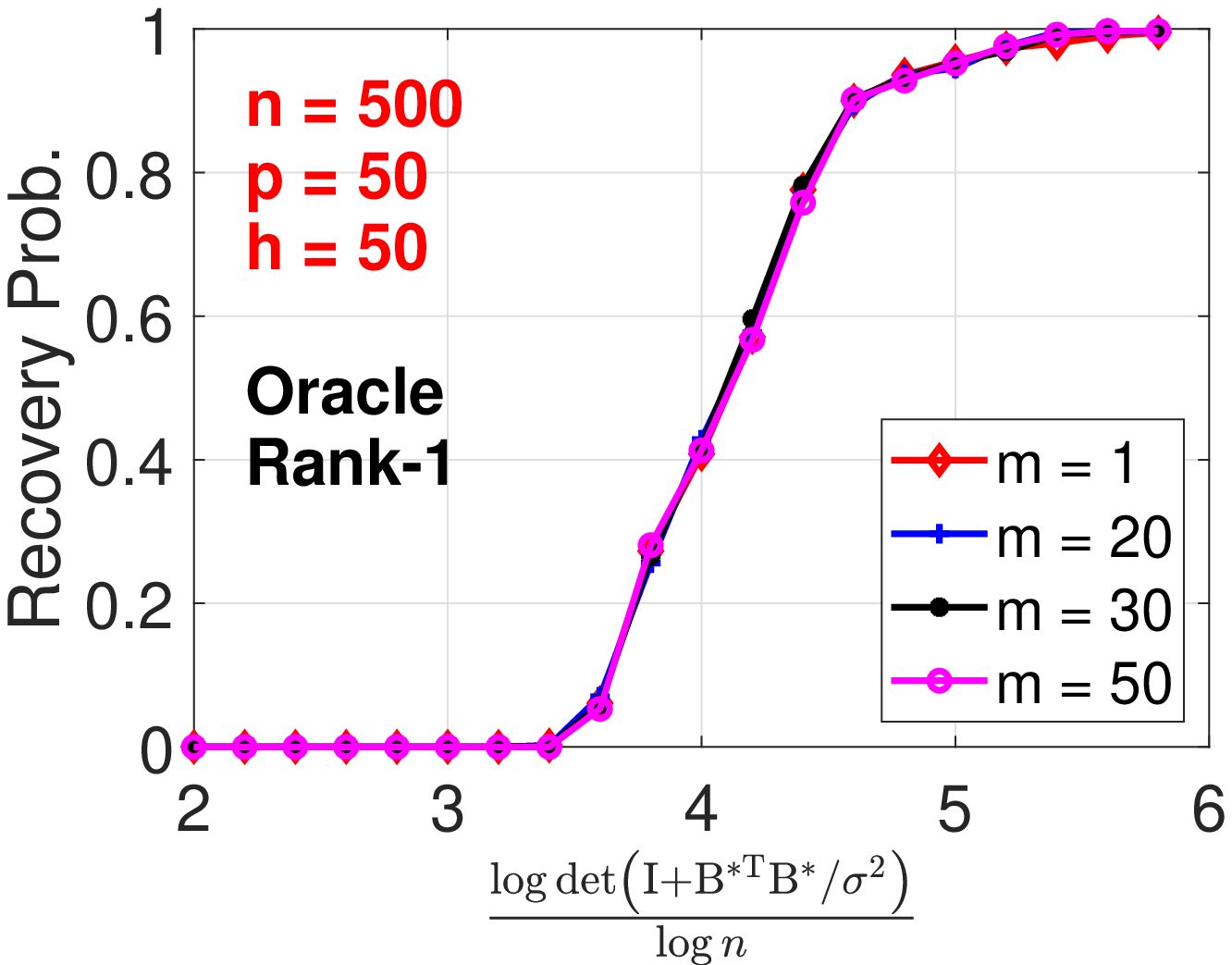}
}
\mbox{
\includegraphics[width=2.1in]{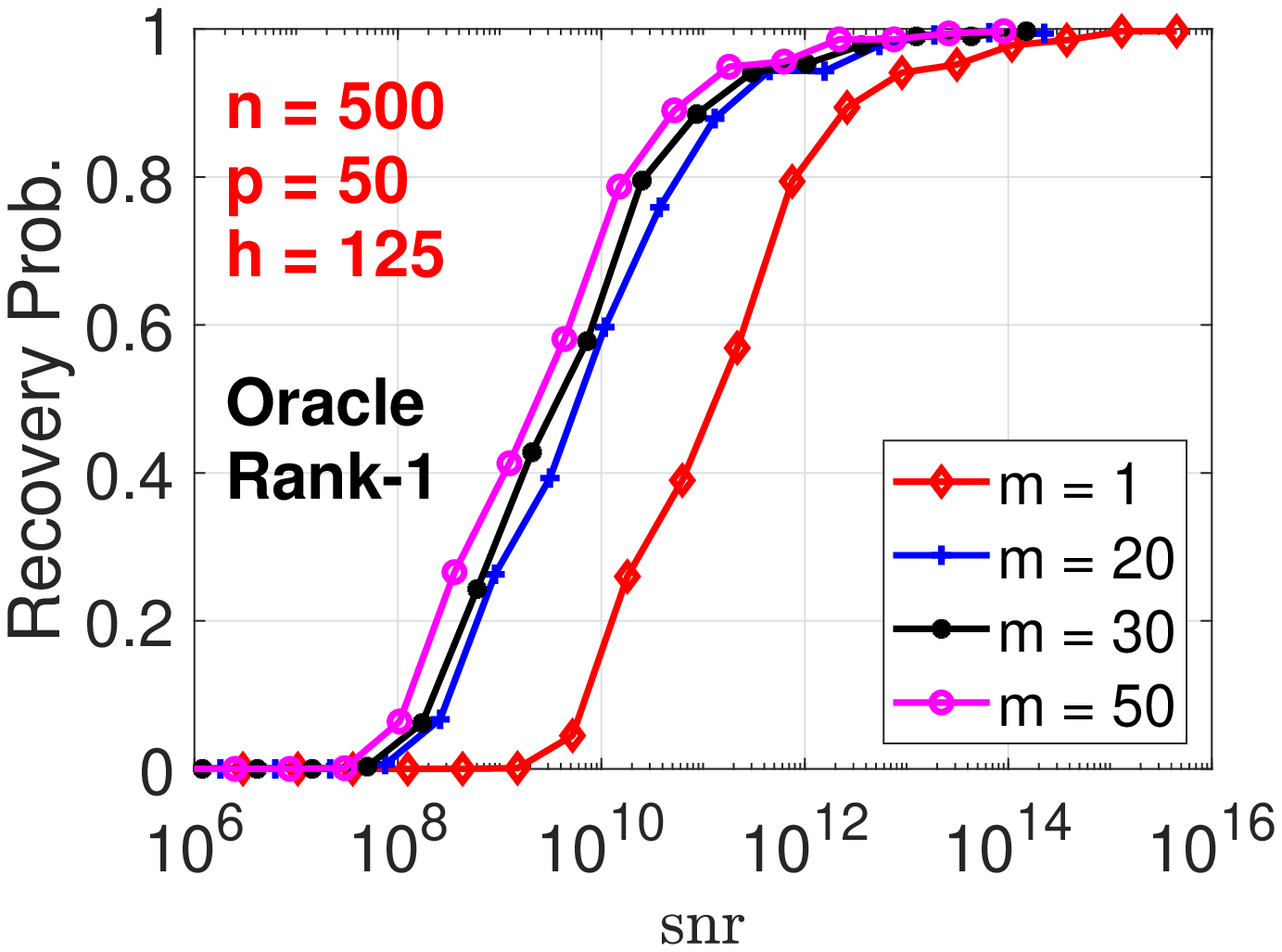}
\includegraphics[width=2.1in]{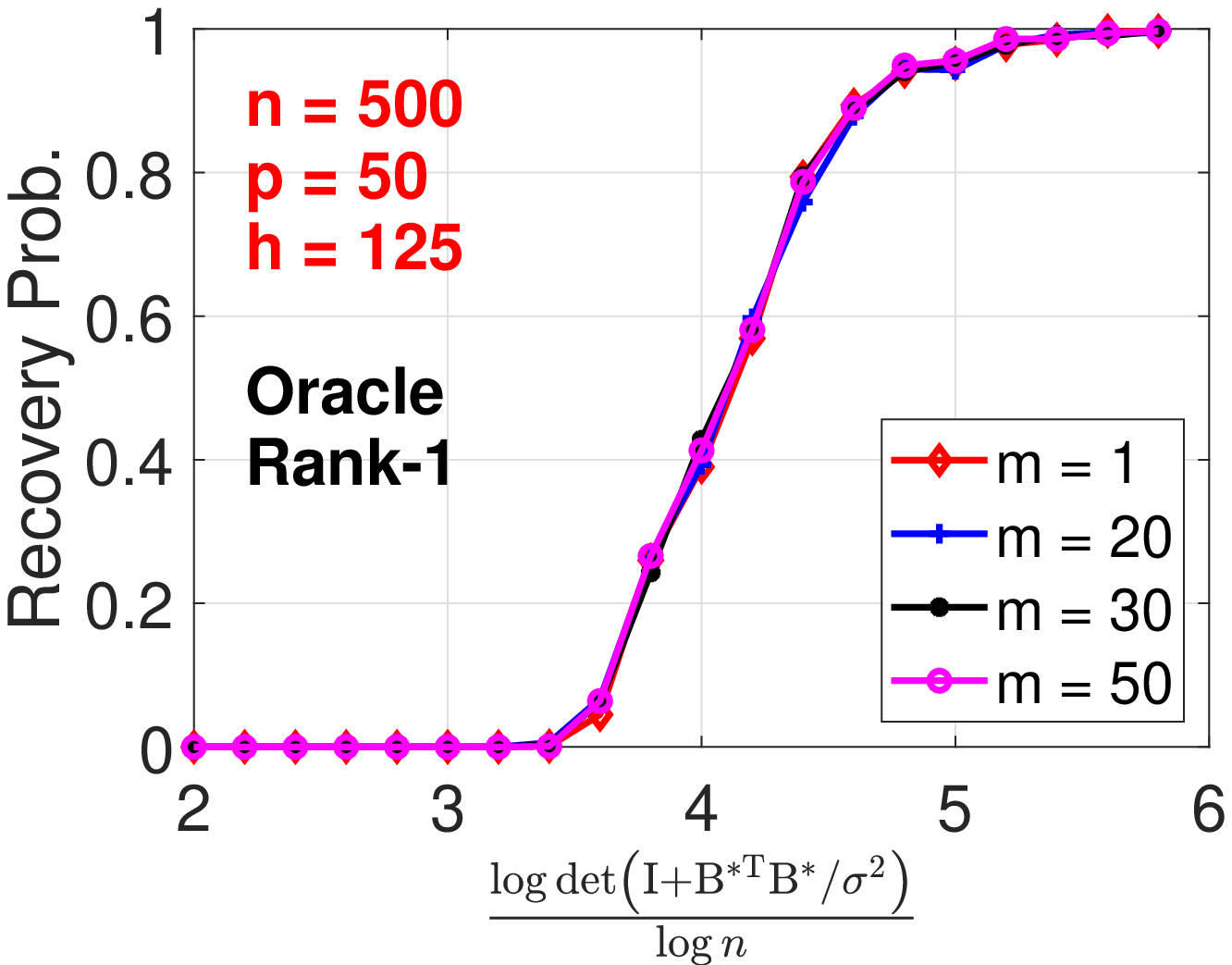}
}

\mbox{
\includegraphics[width=2.1in]{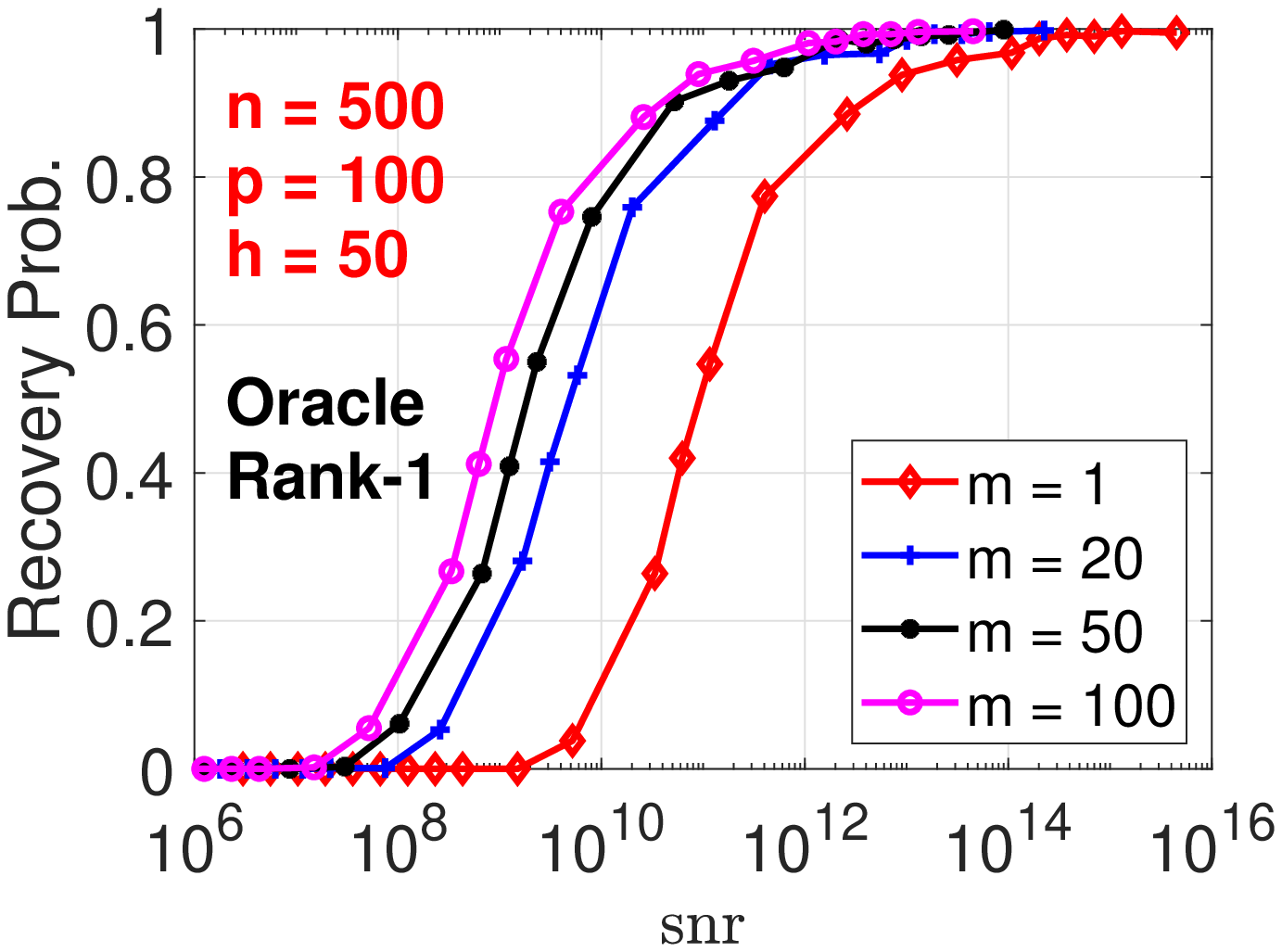}
\includegraphics[width=2.1in]{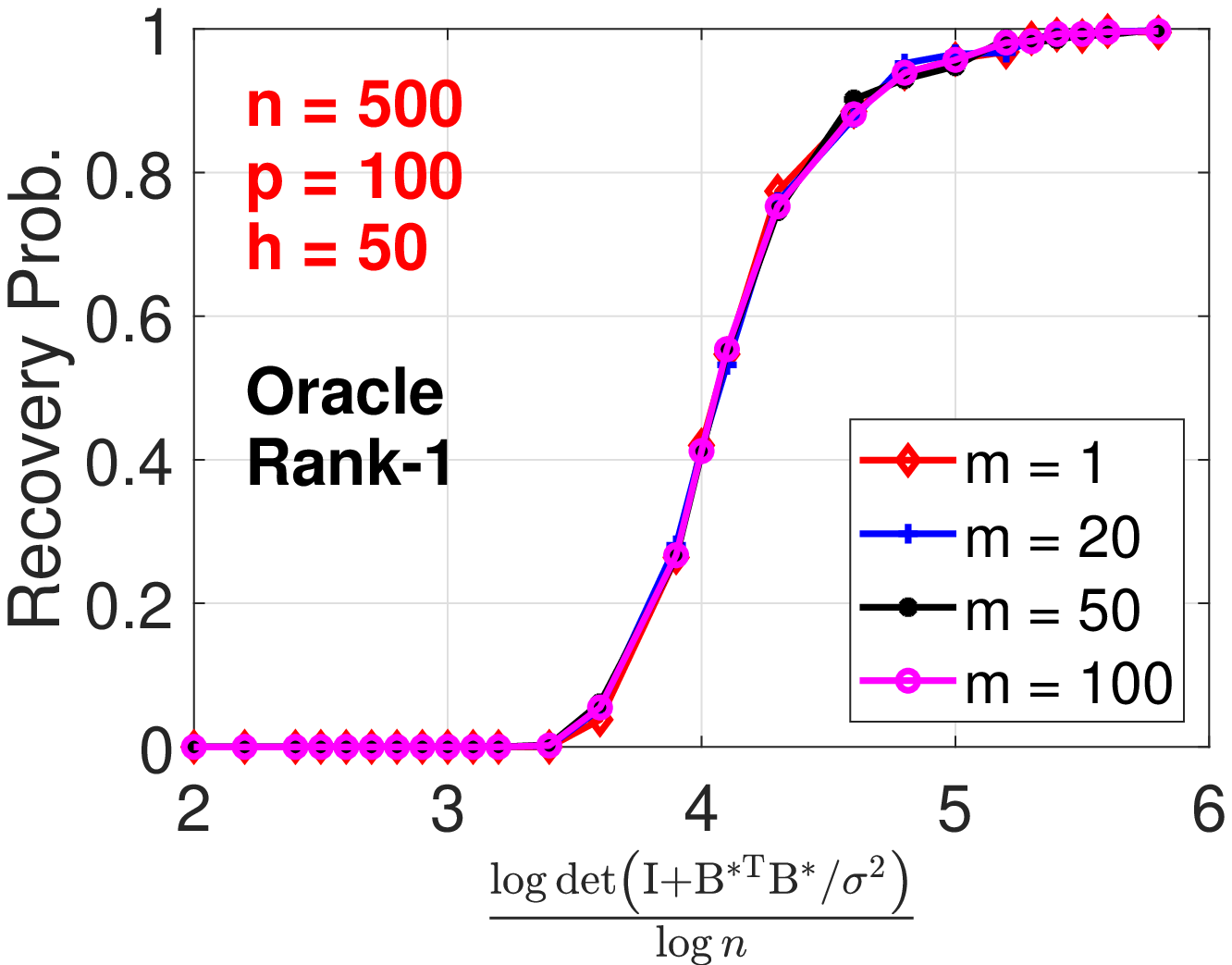}
}
\mbox{
\includegraphics[width=2.1in]{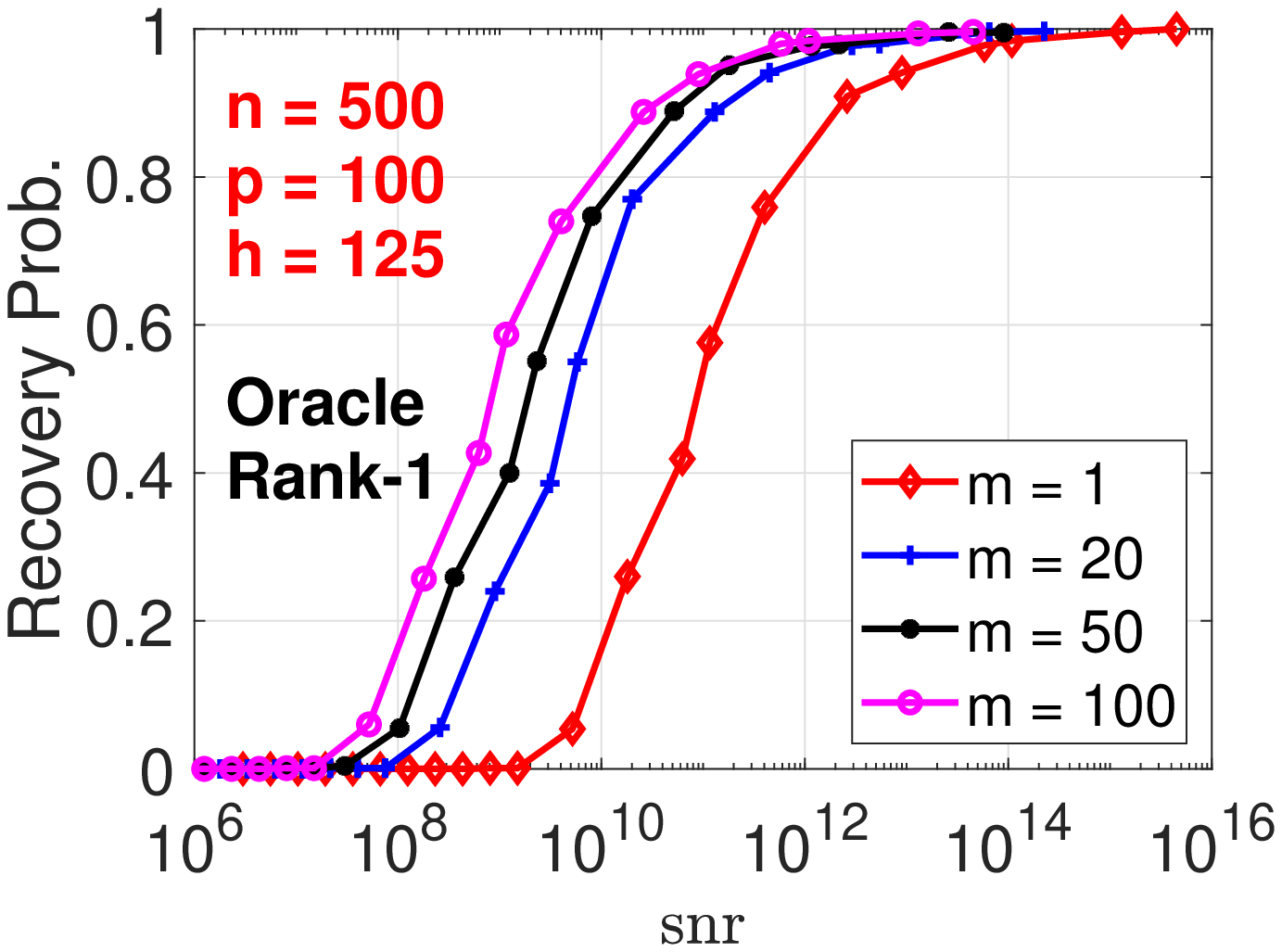}
\includegraphics[width=2.1in]{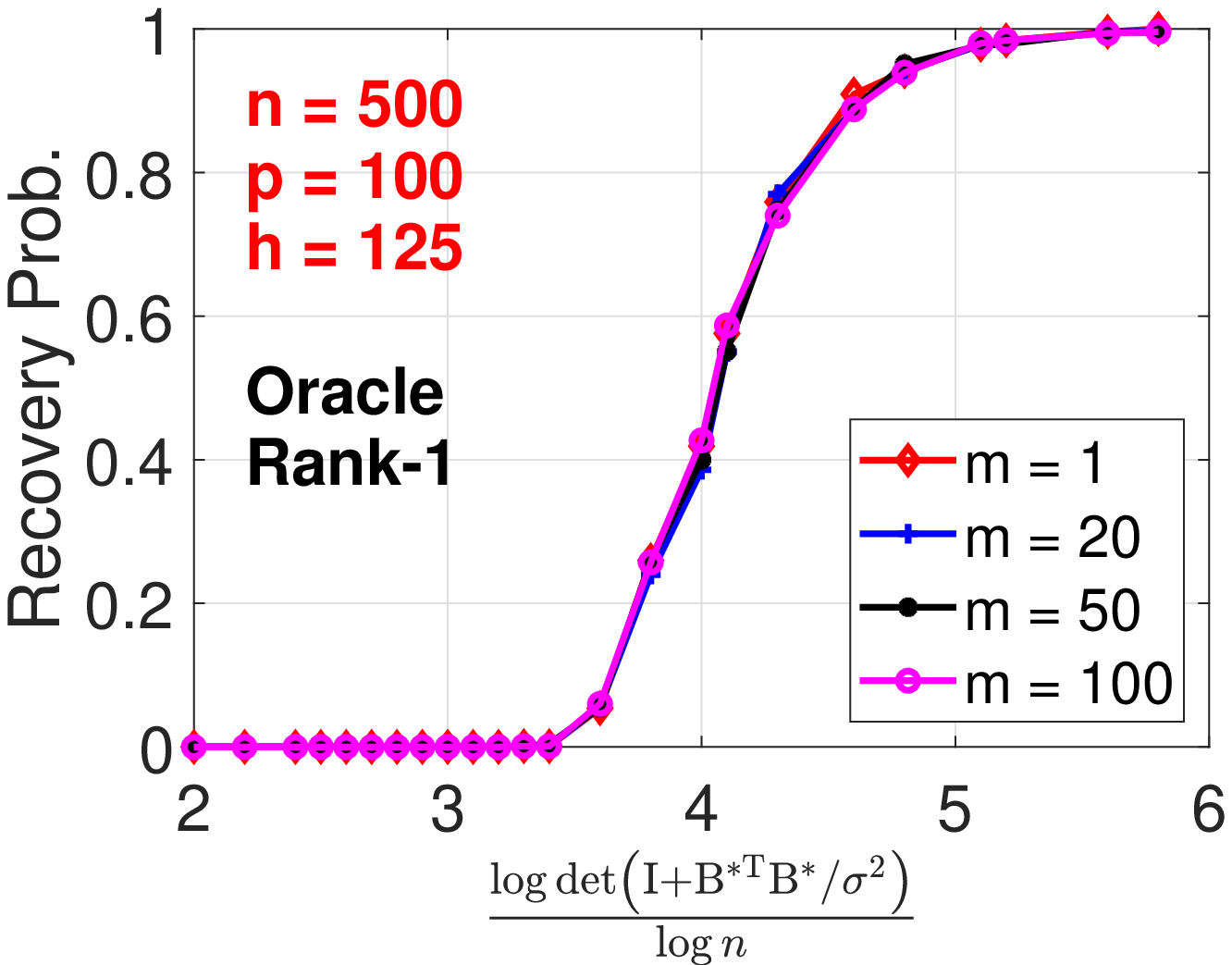}
}
\end{center}
\vspace{-0.2in}
\caption{Oracle case (Rank-1):  Correct recovery probability
$\Prob(\hat{\bPi} = \bPi^{*})$ (vertical axis)
versus $\snr$ (\textbf{left panels}) or
 $\frac{\logdet\bracket{\bI + \bB^{*\rmt}\bB^{*}/\sigma^2}}{\log n} = \frac{\log
 \left(1+m\cdot\snr\right)}{\log n}$ (\textbf{right panels}).
 }
\label{fig:oracle_recover_rank1}
\end{figure}

\begin{figure}[t]
\begin{center}
\mbox{
\includegraphics[width=2.1in]{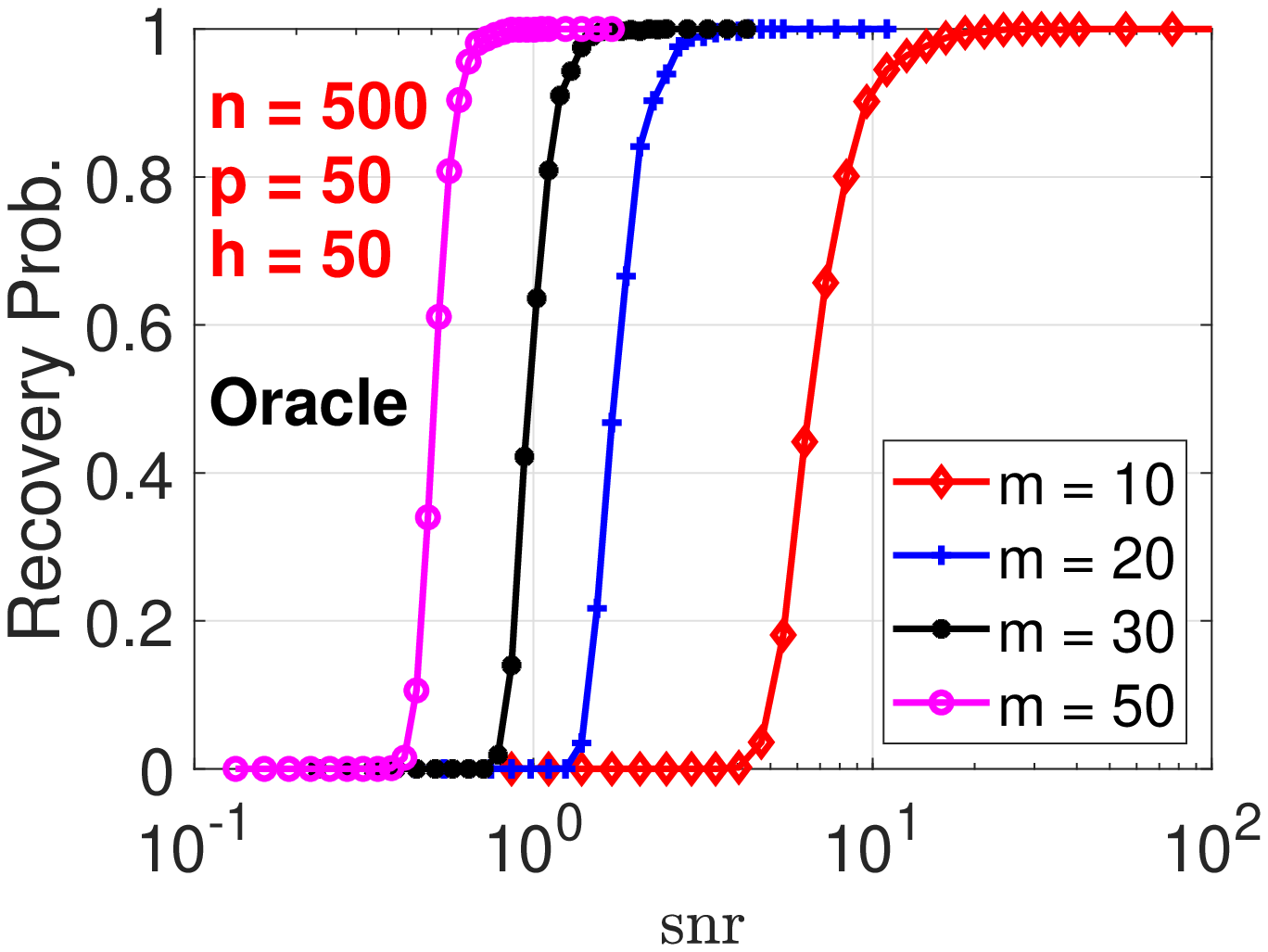}
\includegraphics[width=2.1in]{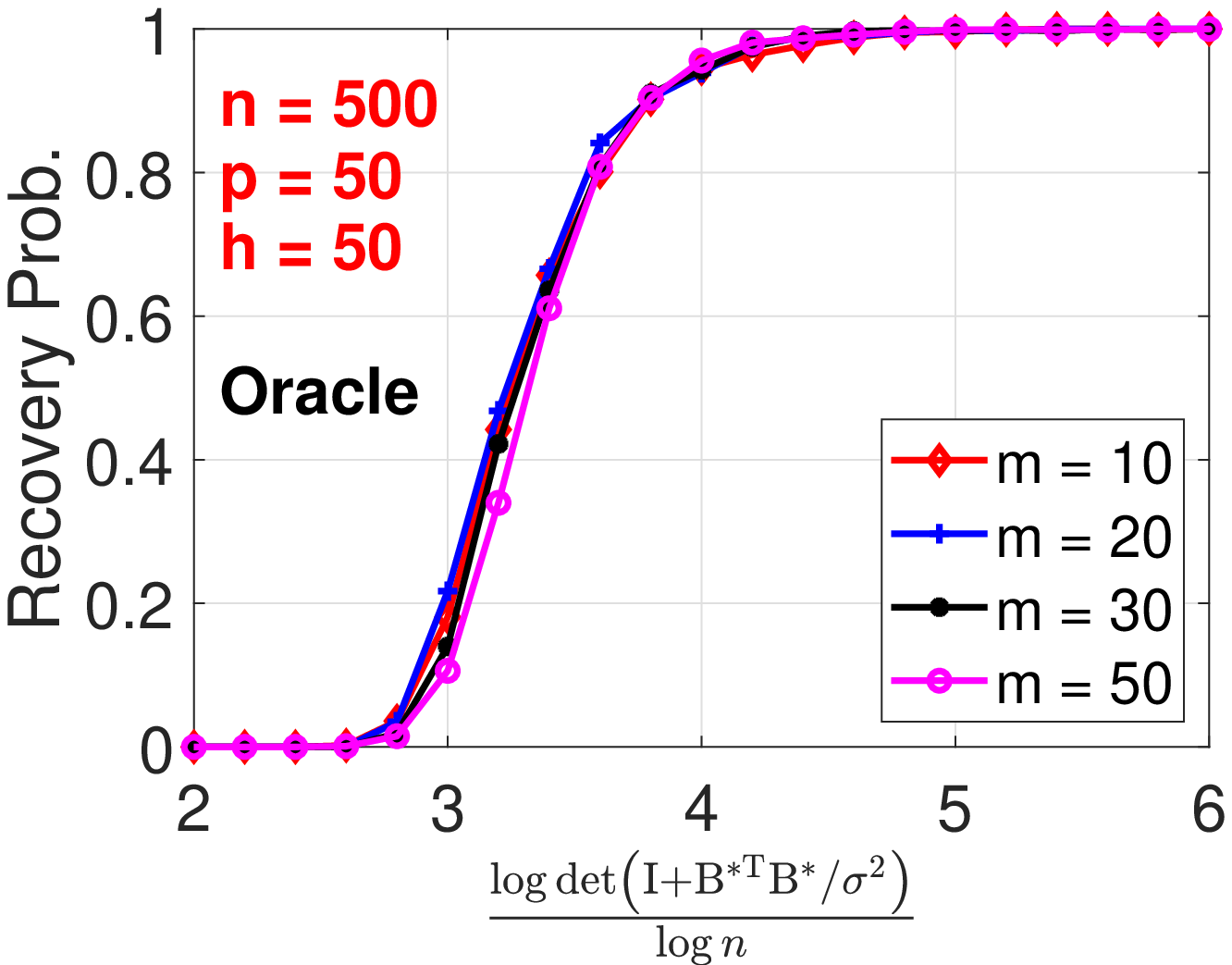}
}
\mbox{
\includegraphics[width=2.1in]{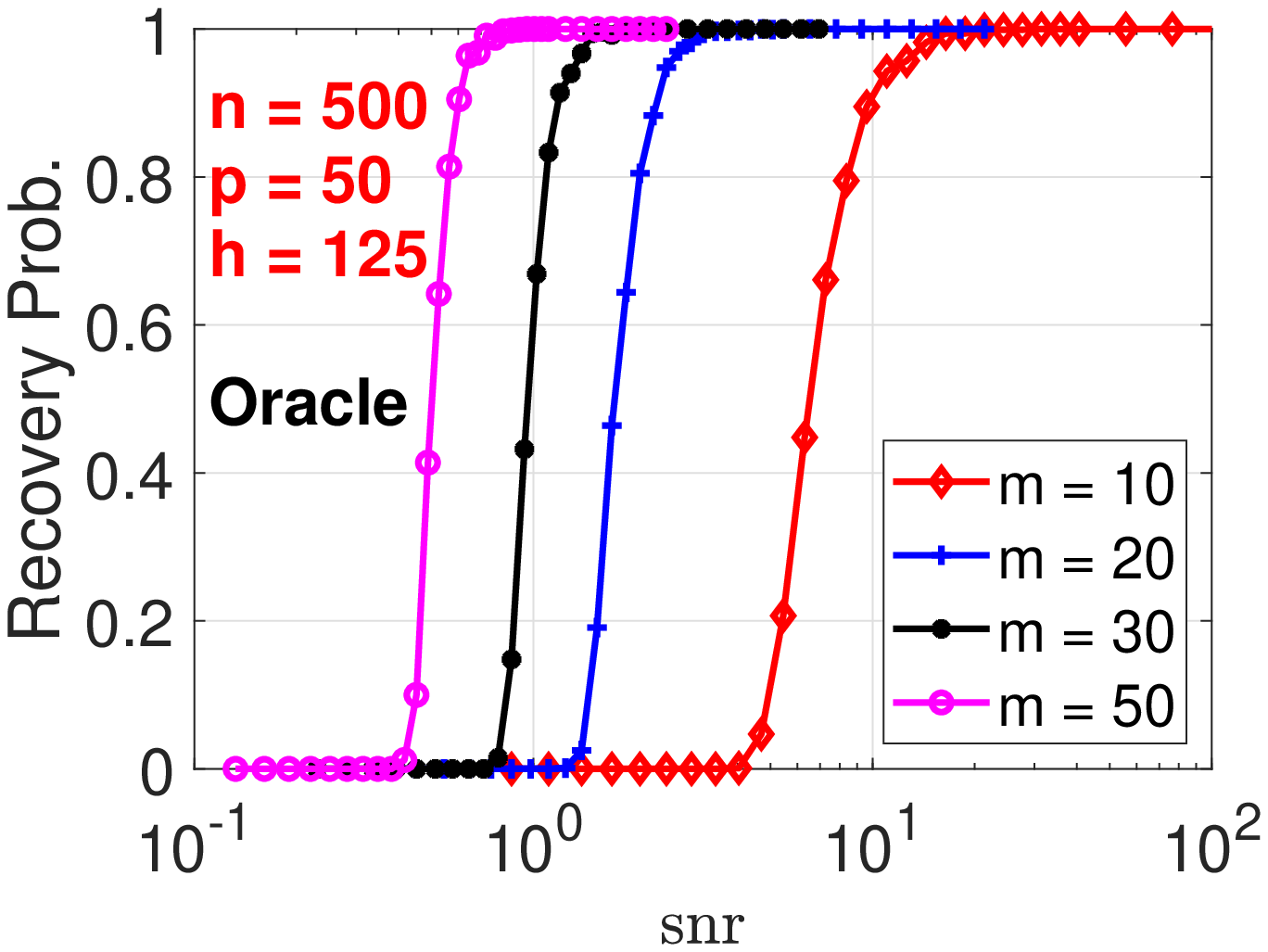}
\includegraphics[width=2.1in]{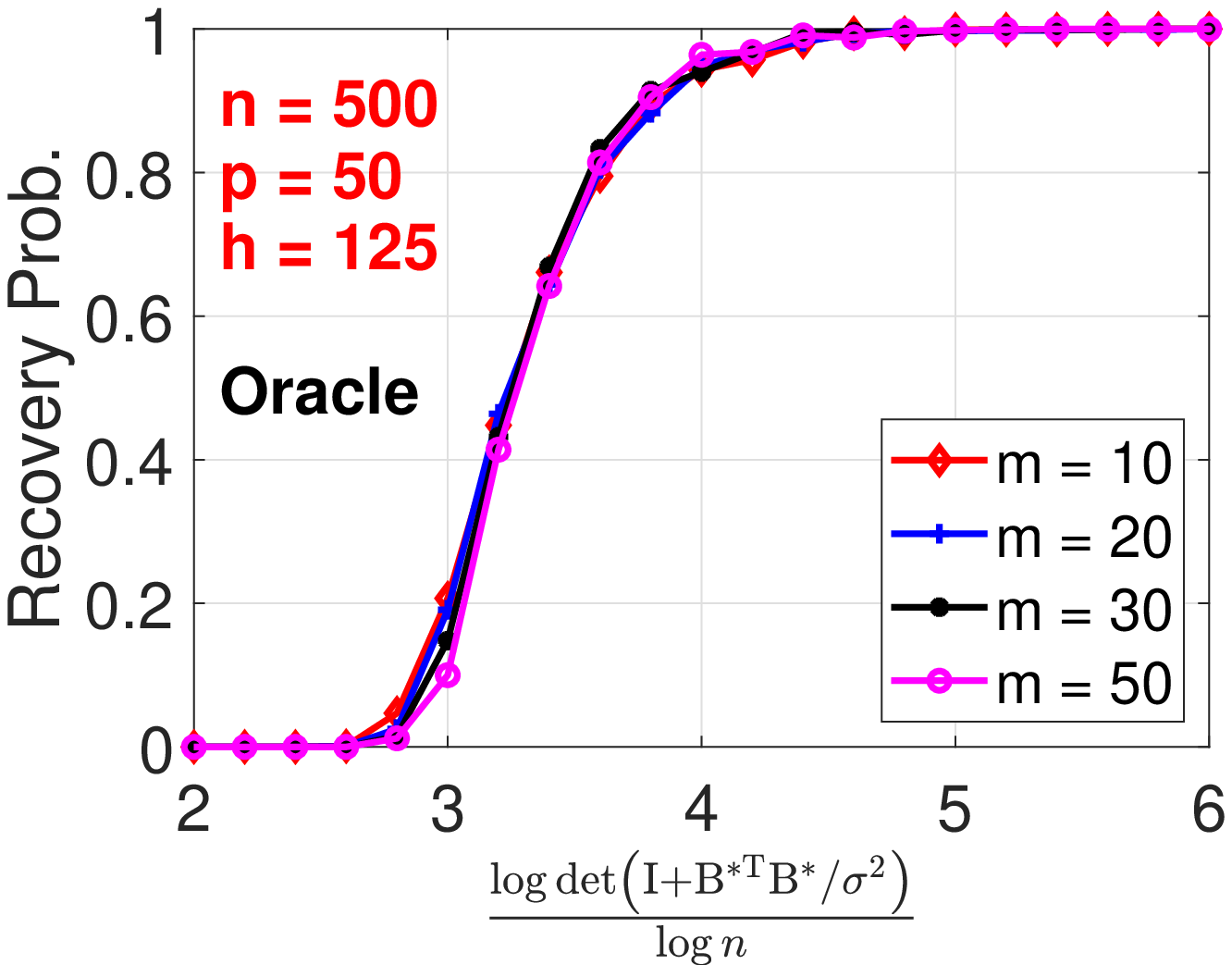}
}

\mbox{
\includegraphics[width=2.1in]{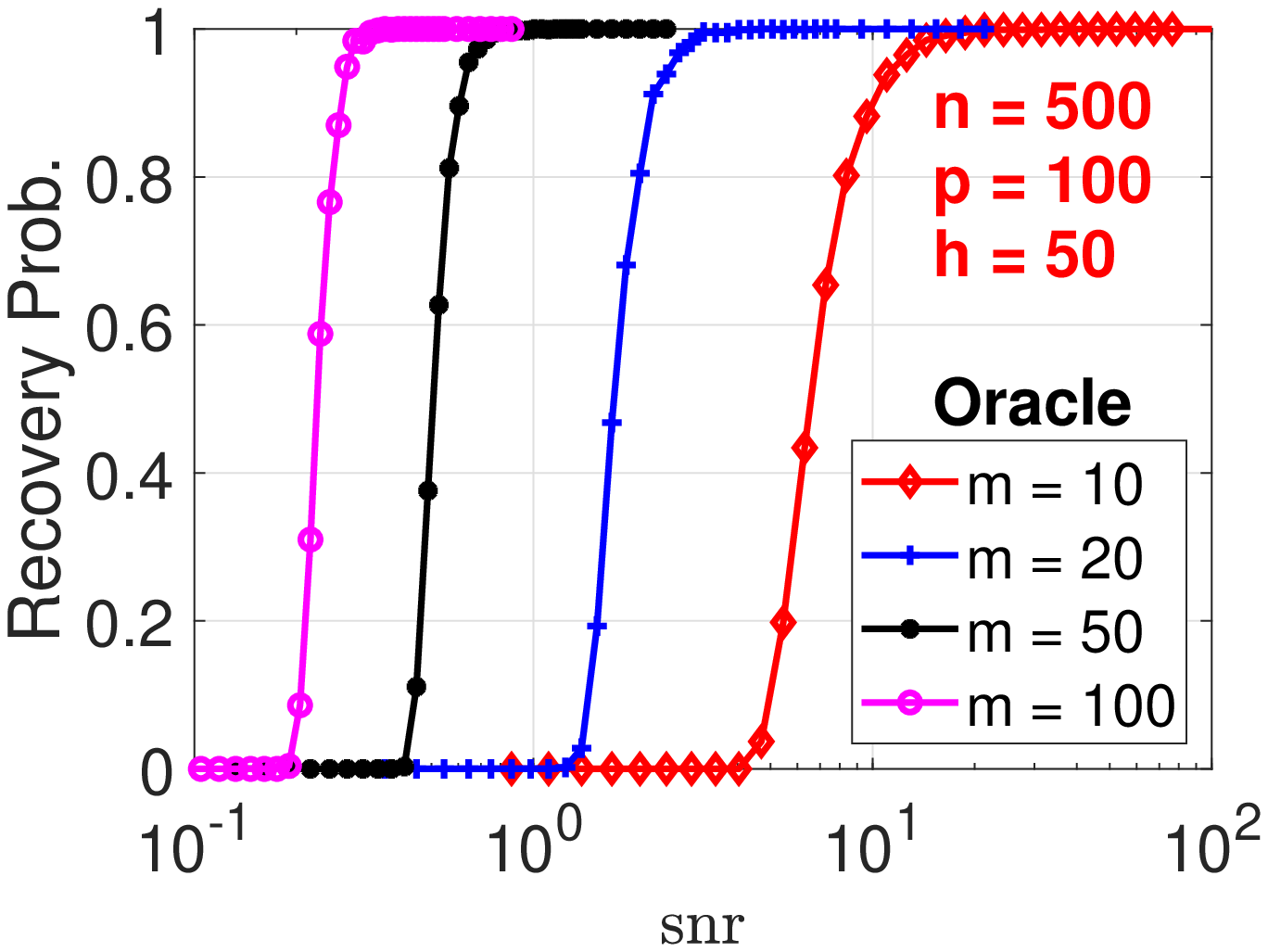}
\includegraphics[width=2.1in]{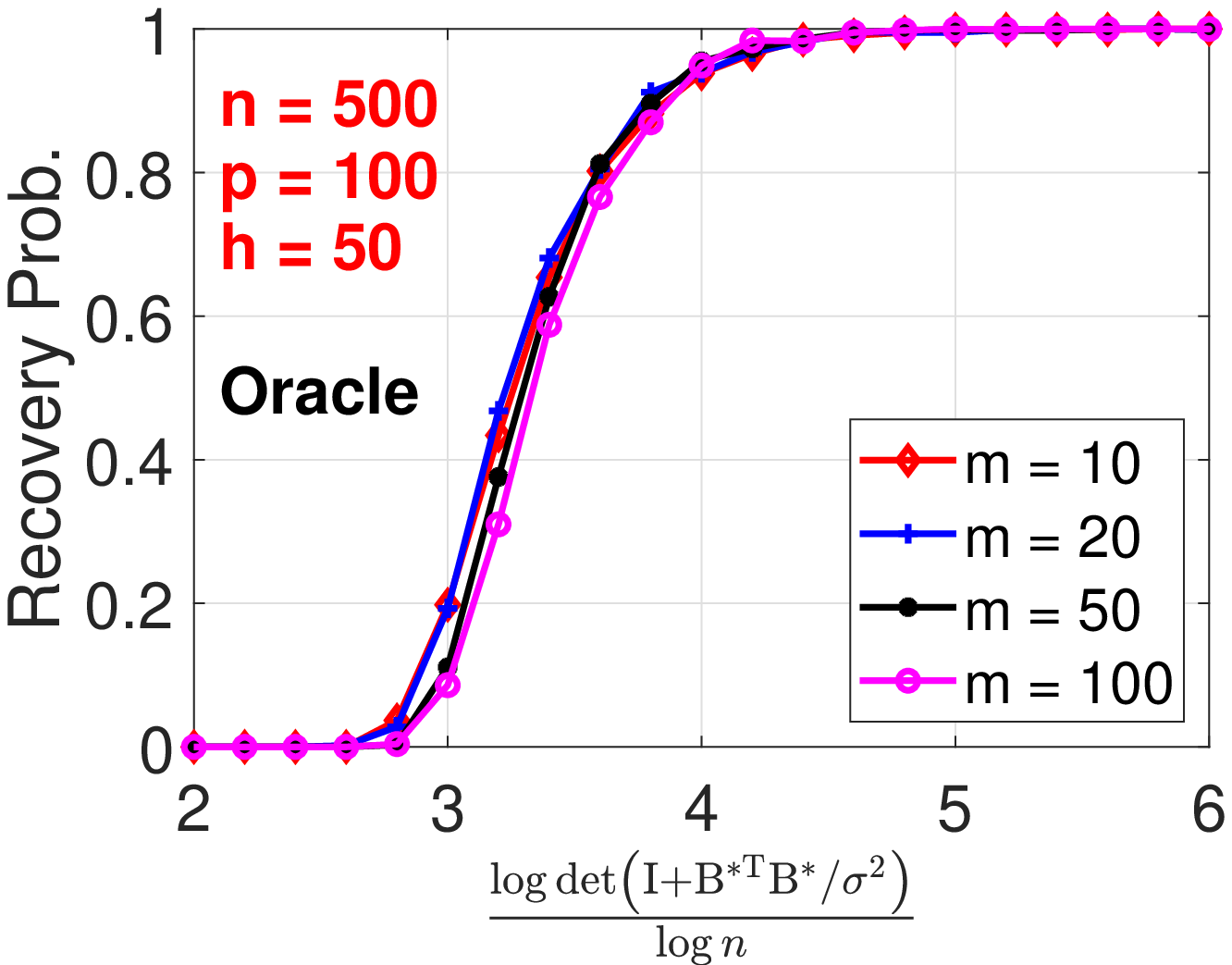}
}
\mbox{
\includegraphics[width=2.1in]{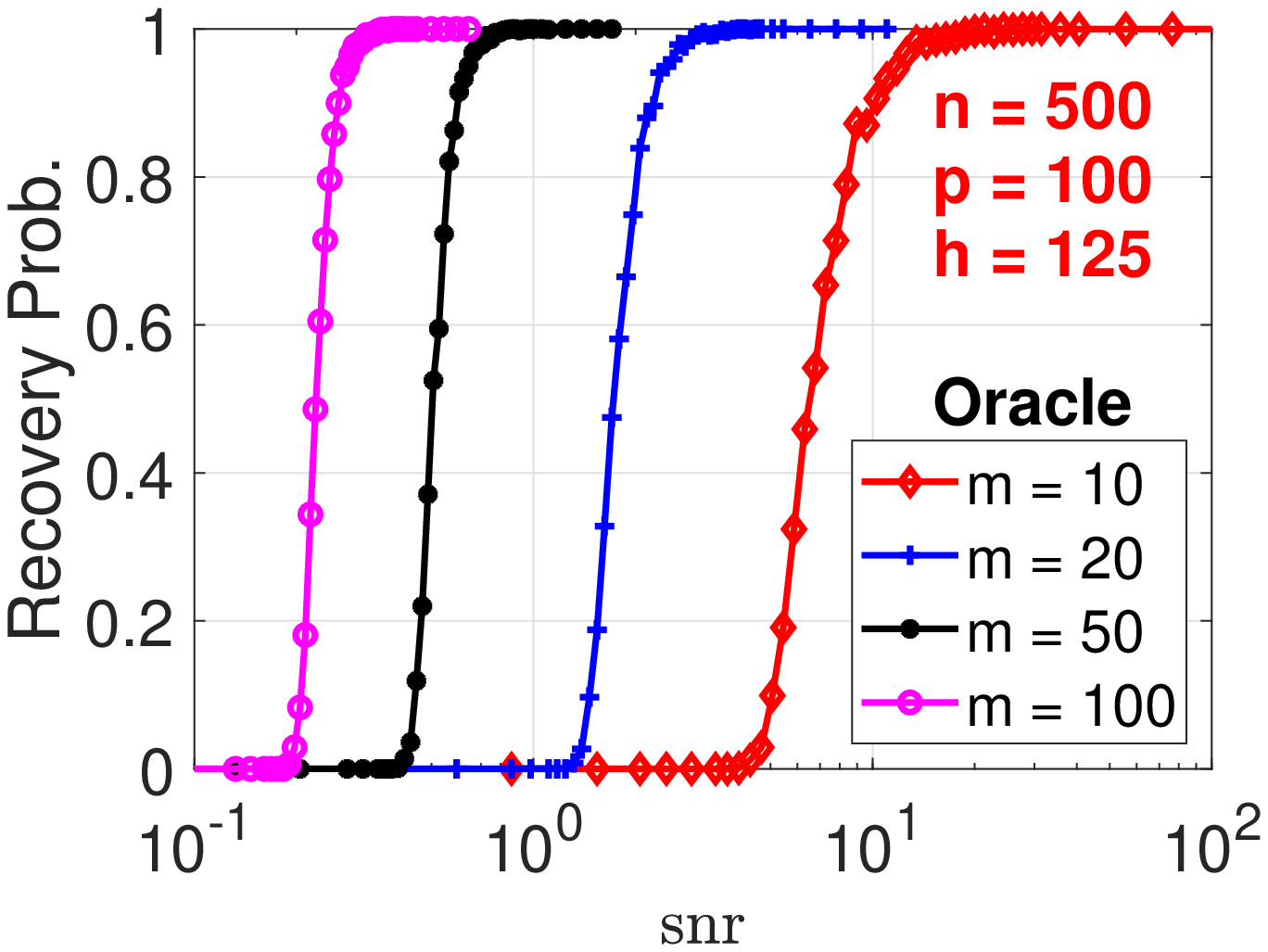}
\includegraphics[width=2.1in]{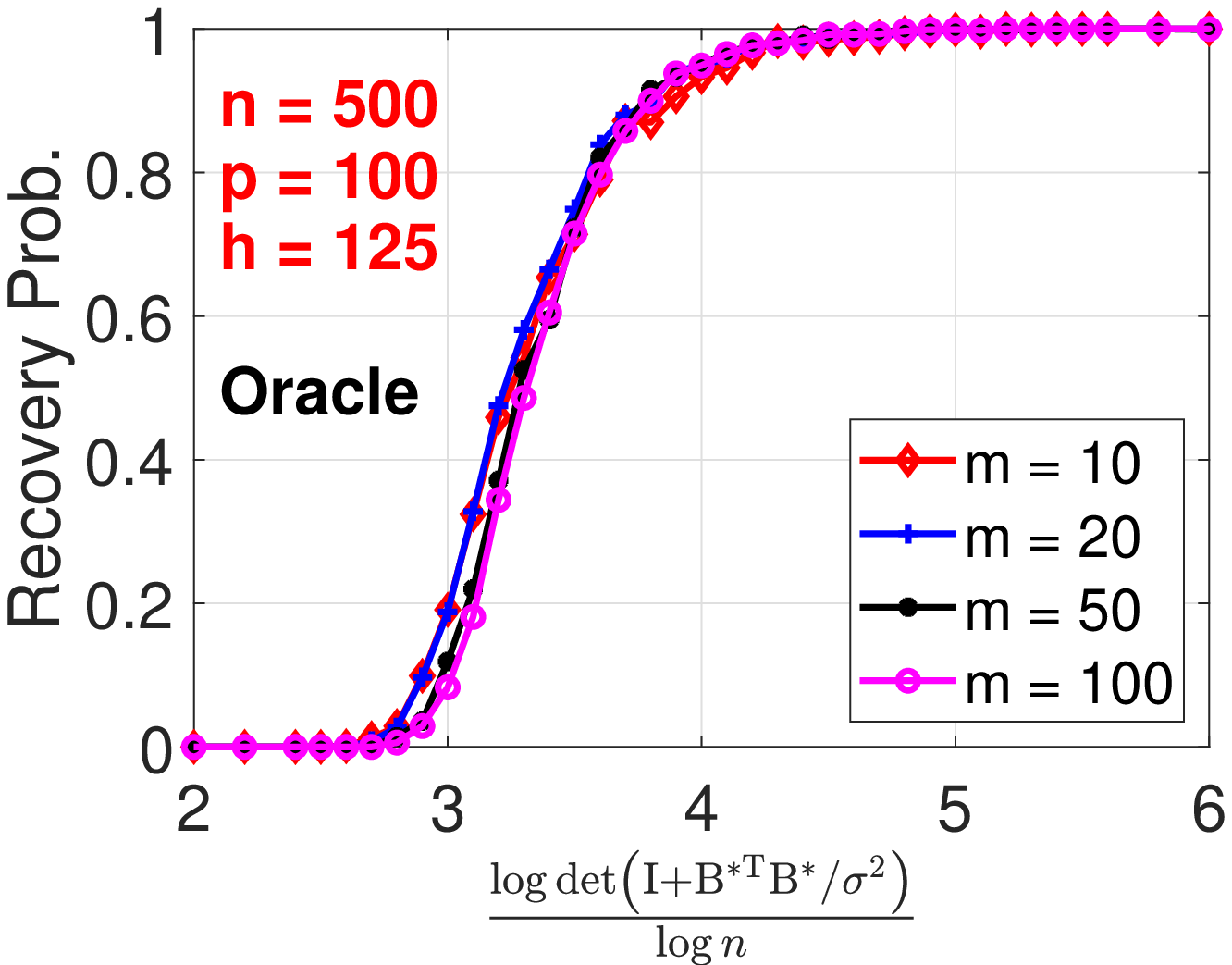}
}
\end{center}
\vspace{-0.2in}
\caption{Oracle case (Full-rank):  Correct recovery probability
$\Prob(\hat{\bPi} = \bPi^{*})$ (vertical axis) versus $\snr$ (\textbf{left panels}) or
 $\frac{\log
 \left(1+m\cdot\snr\right)}{\log n}$ (\textbf{right panels}).}
\label{fig:oracle_recover_fullrank}%\vspace{-0.15in}
\end{figure}

\paragraph{Full-rank case}
We consider the case in which the columns of $\bB^*$ are orthogonal to each other, i.e., $\bB_{:, i}^{*} \perp \bB_{:, j}^*$, $1\leq i \neq j \leq m$. For simplicity, we set $\bB_{:, i}^{*}  \parallel \be_i$, where $\{ \be_i \}$ denotes the canonical basis. The simulation results for this setting are shown in
Fig.~\ref{fig:oracle_recover_fullrank}.
As for the rank-one case, we observe that the curves
displaying the correct recovery rate $\Prob(\hat{\bPi} = \bPi^{*})$
for different values of $m$
almost coincide when using the quantity $\frac{\logdet\bracket{\bI + \bB^{*\rmt}\bB^{*}/\sigma^2}}{\log n}$ for the horizontal axis.
The latter is thus confirmed to be the central determining factor in predicting
whether $\bPi^{*}$ can be successfully
recovered or not.  However, different from the
rank-one case, we witness a significant decrease
regarding the required $\snr$ needed for high recovery rates. For example,
$\snr \approx 10^{14}$ is required in the
rank-one case, while in the full-rank case, the required value
of $\snr$ is less than $10$. As predicted by Theorem~\ref{thm:exact_minimax} and Theorem~\ref{thm:oracle_succ_optim}, this reduction is a consequence of an increased stable rank $\varrho(\bB^{*})$.

%\tcr{Additionally, we find that the performance is irrelevant with the
%setting of $h$, which is consistent with both Theorem~\ref{thm:exact_minimax}
%and Theorem~\ref{thm:oracle_succ_optim}. This phenomenon forms a sharp
%contrast with the realistic case where large Hamming distance $h$ will
%jeopardize the recovery performance.}

\subsection{Realistic case}

\begin{figure}[h!]
\begin{center}
\mbox{
\includegraphics[width=2.1in]{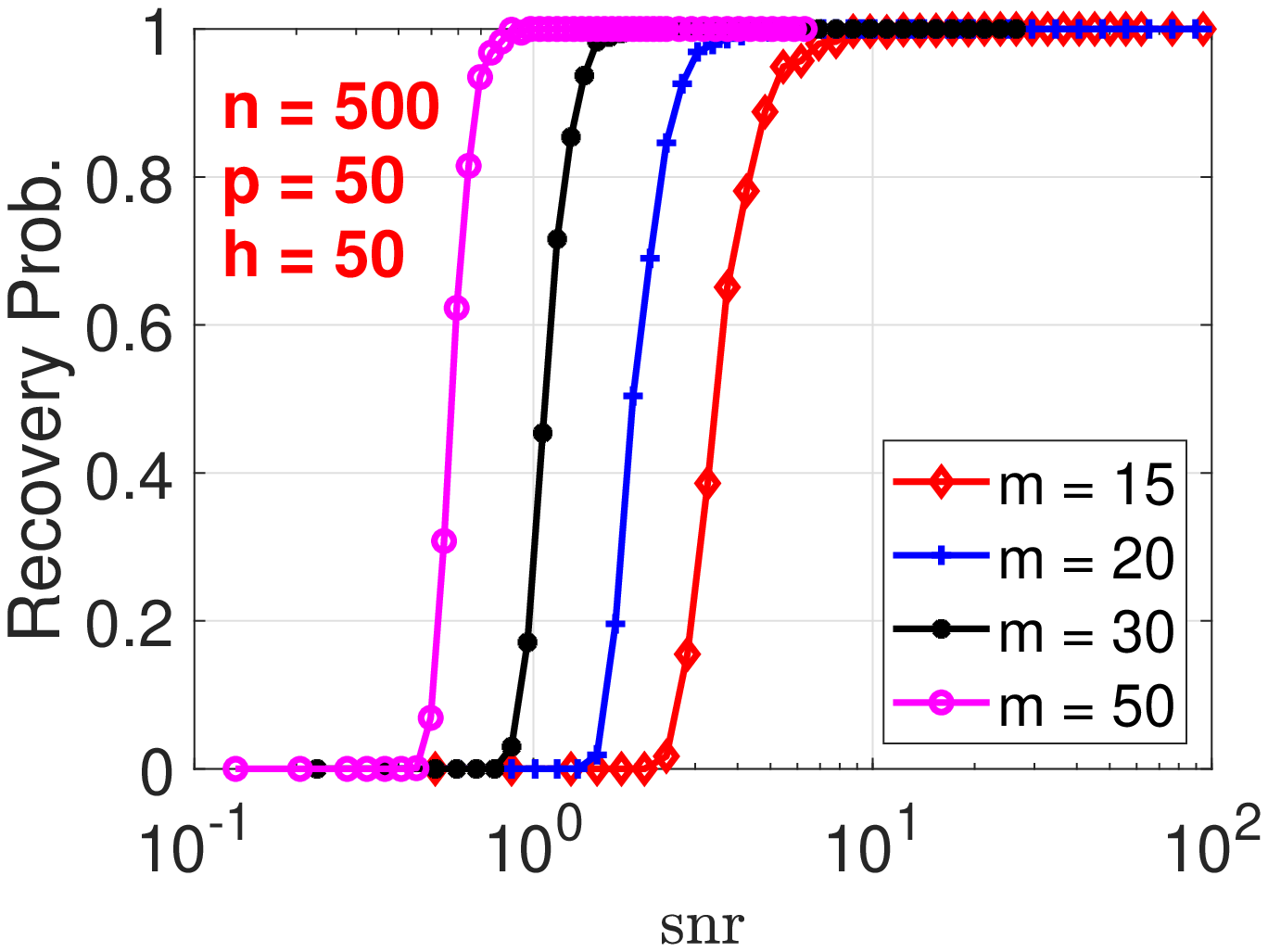}	
\includegraphics[width=2.1in]{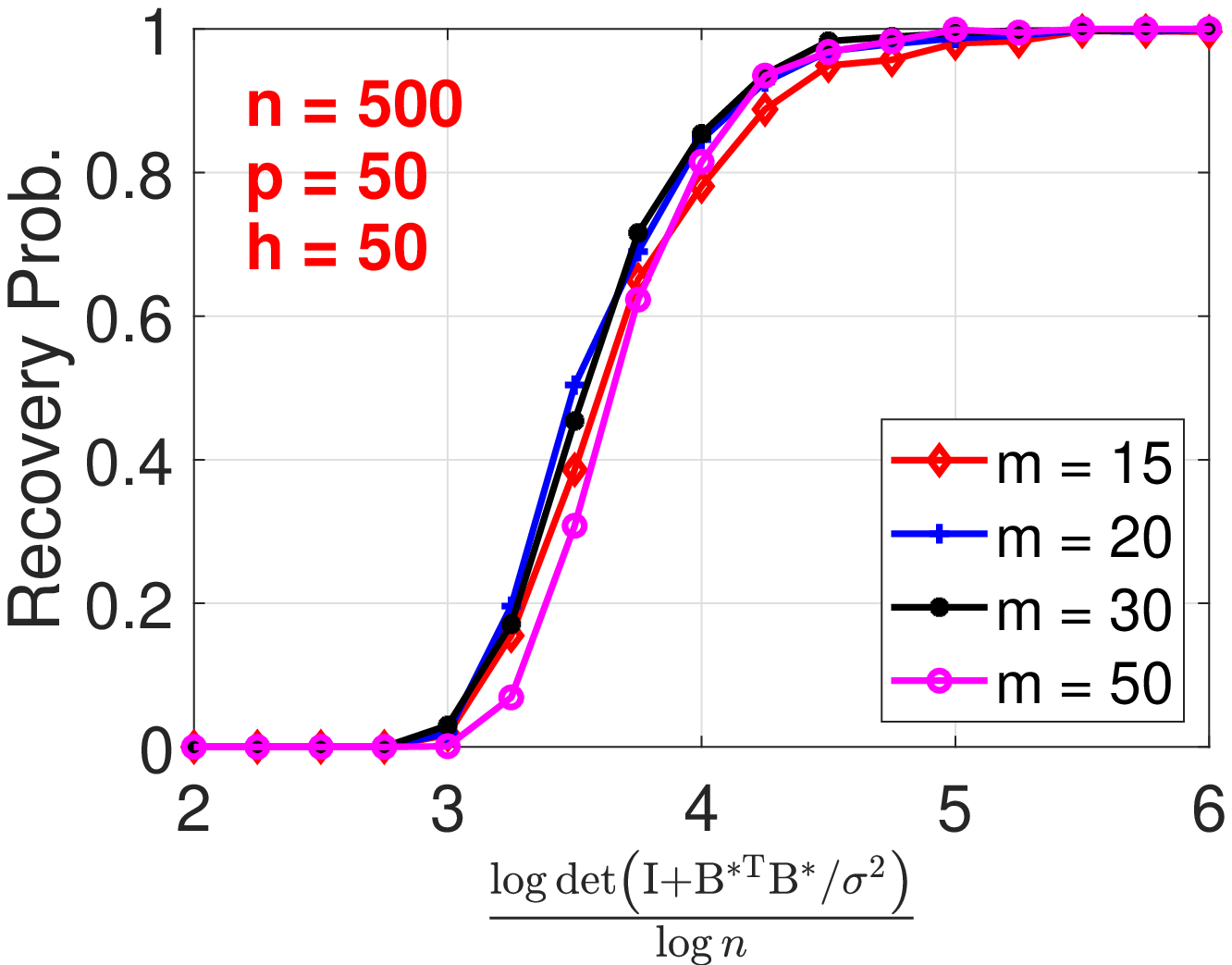}	
\includegraphics[width=2.1in]{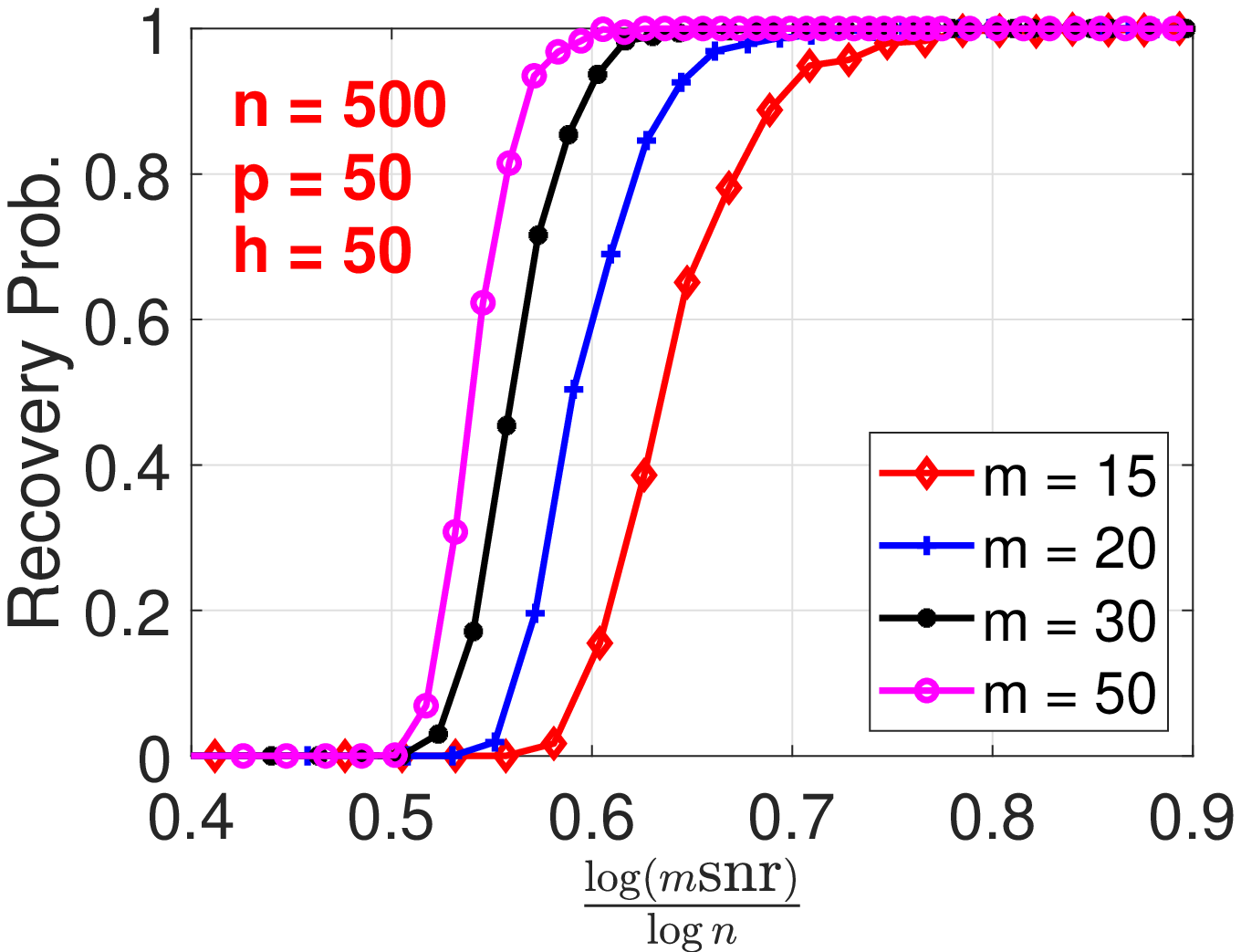}
}

\mbox{
\includegraphics[width=2.1in]{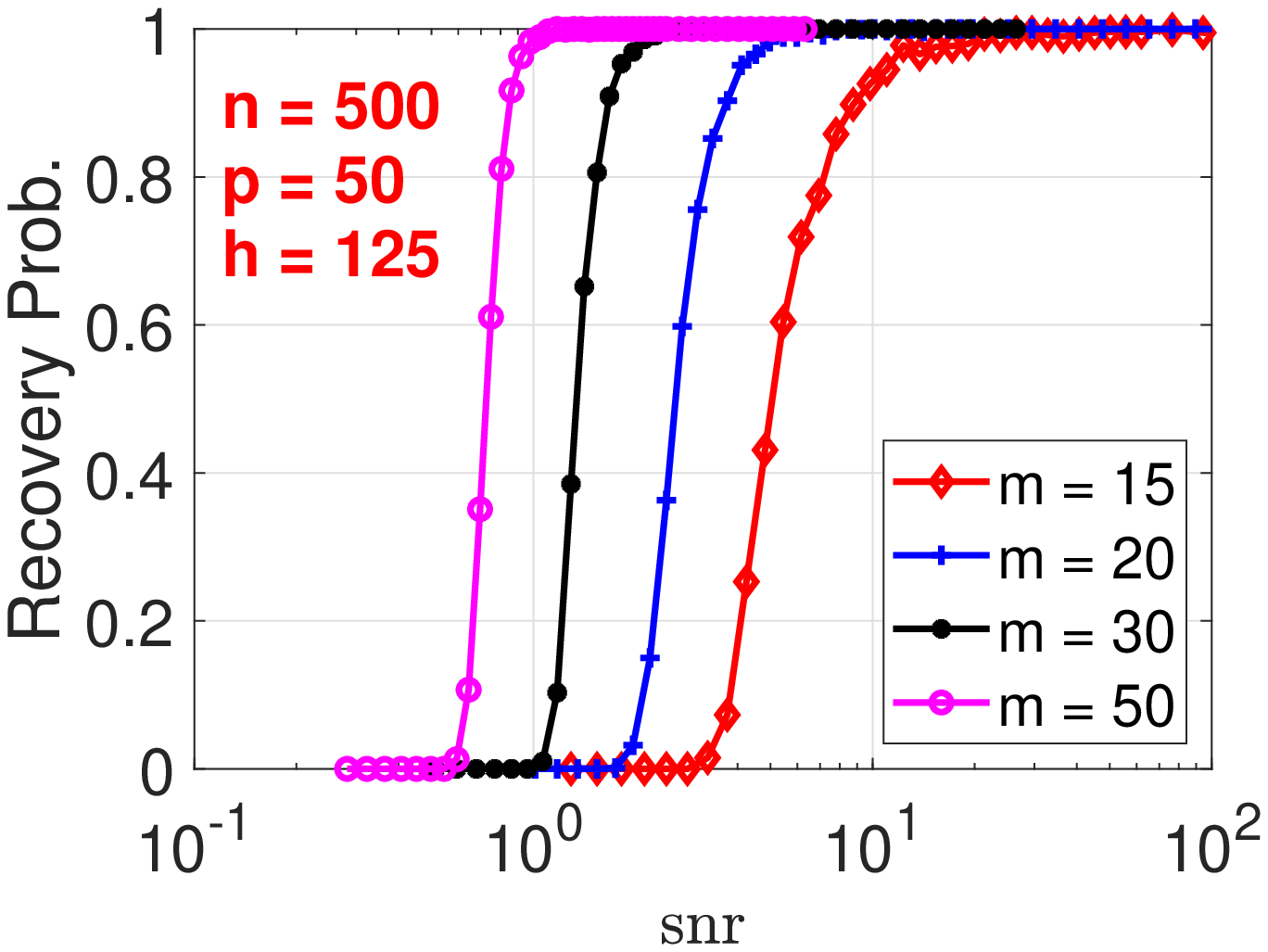}	
\includegraphics[width=2.1in]{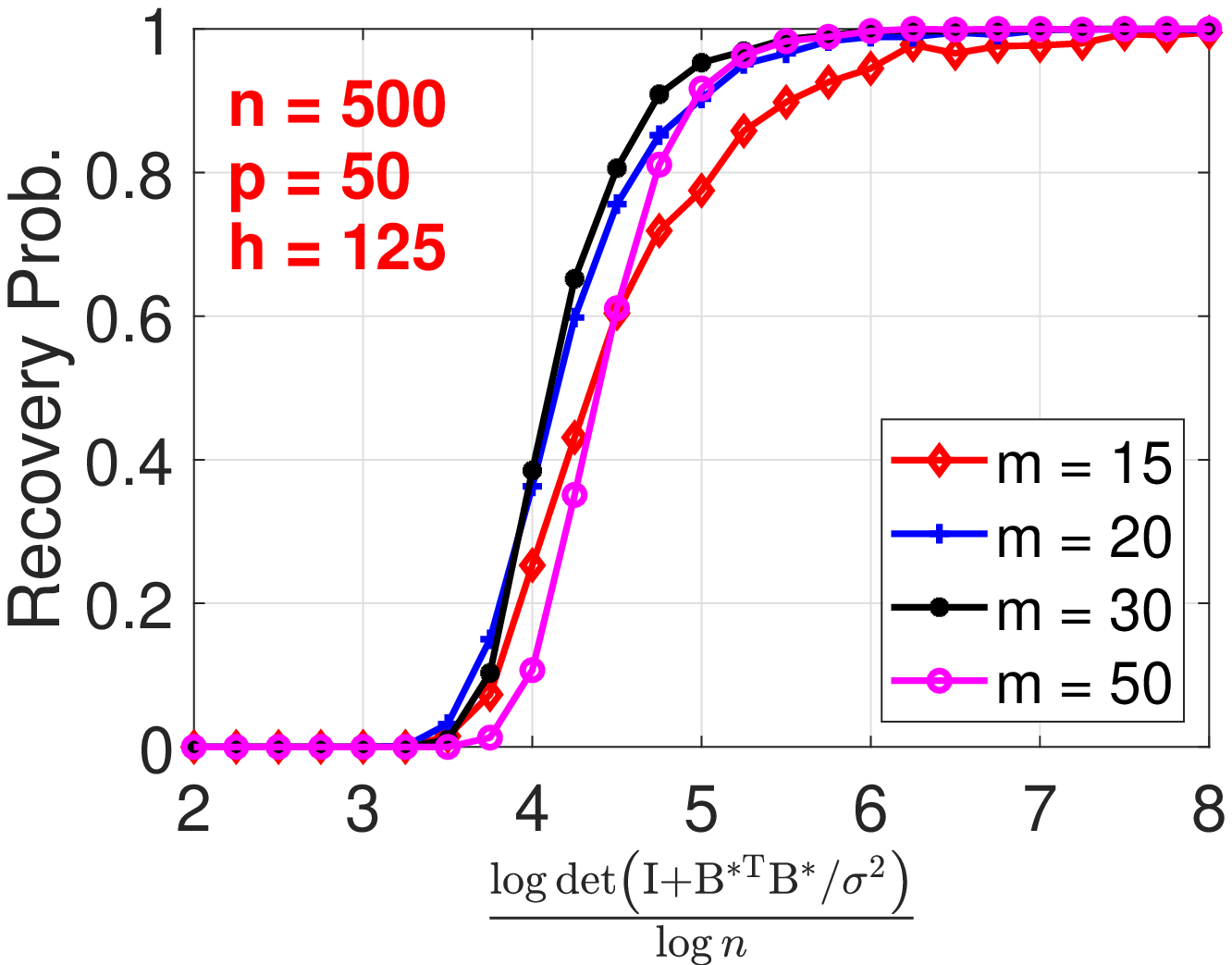}	
\includegraphics[width=2.1in]{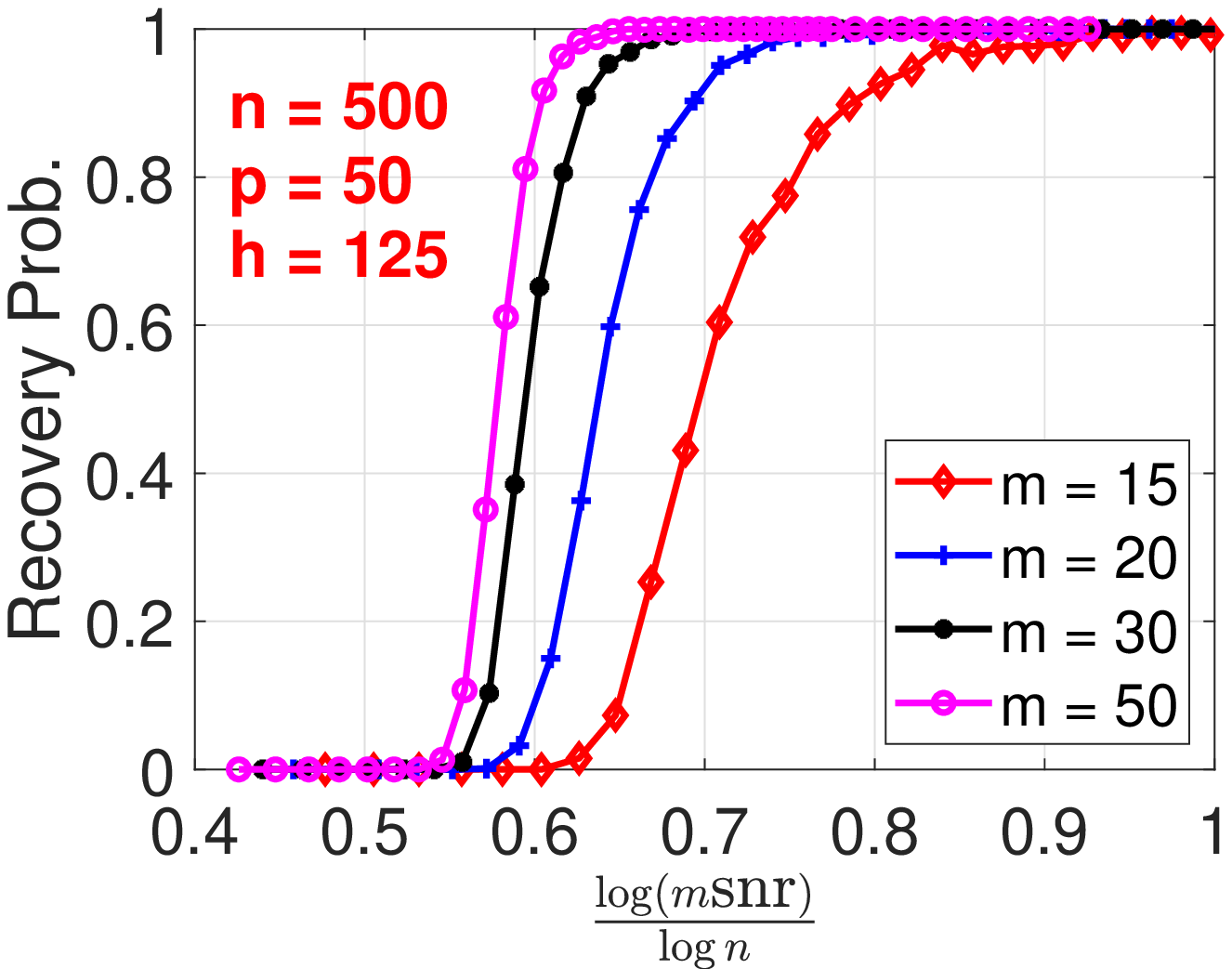}
}

\mbox{
\includegraphics[width=2.1in]{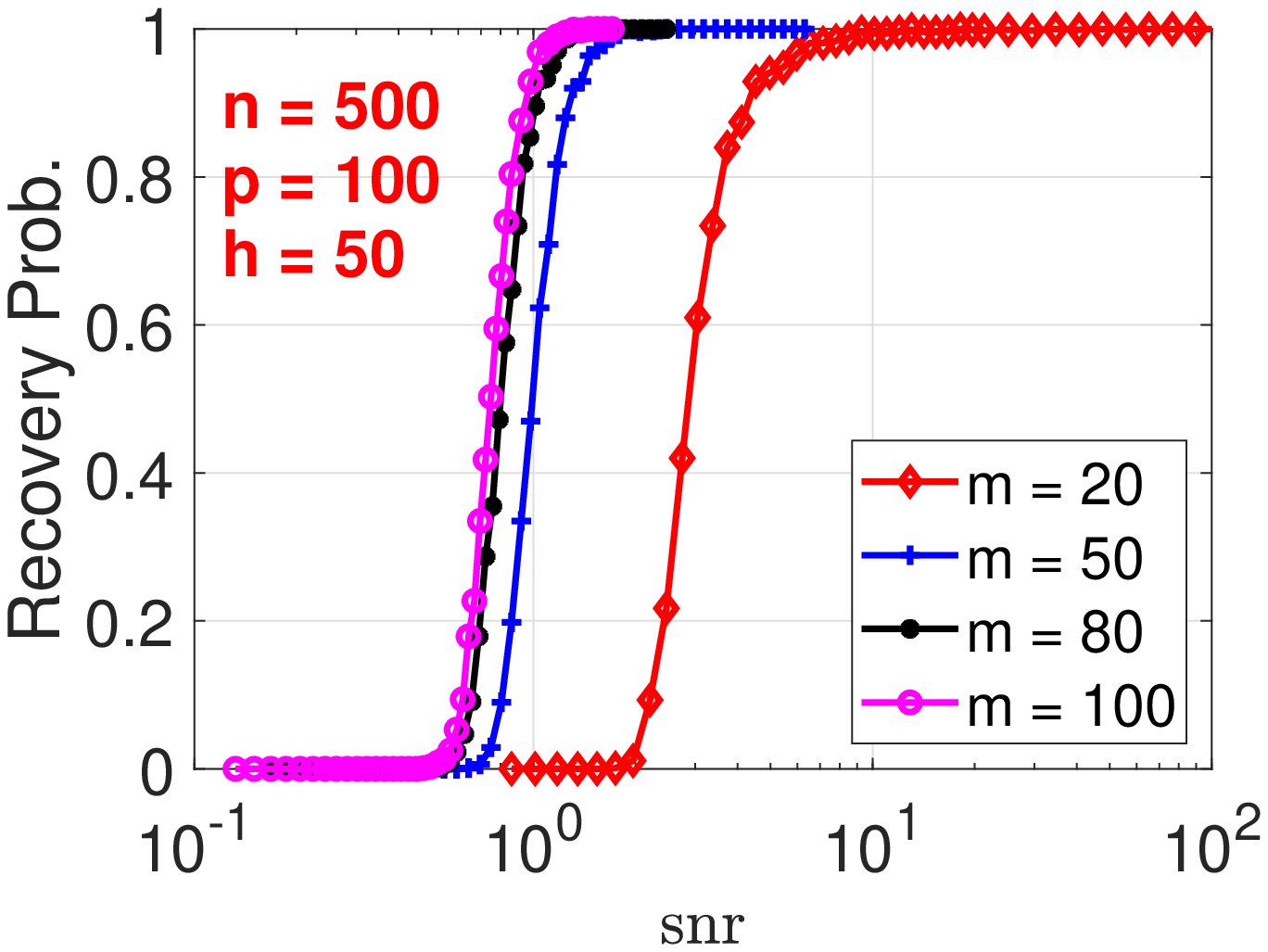}	
\includegraphics[width=2.1in]{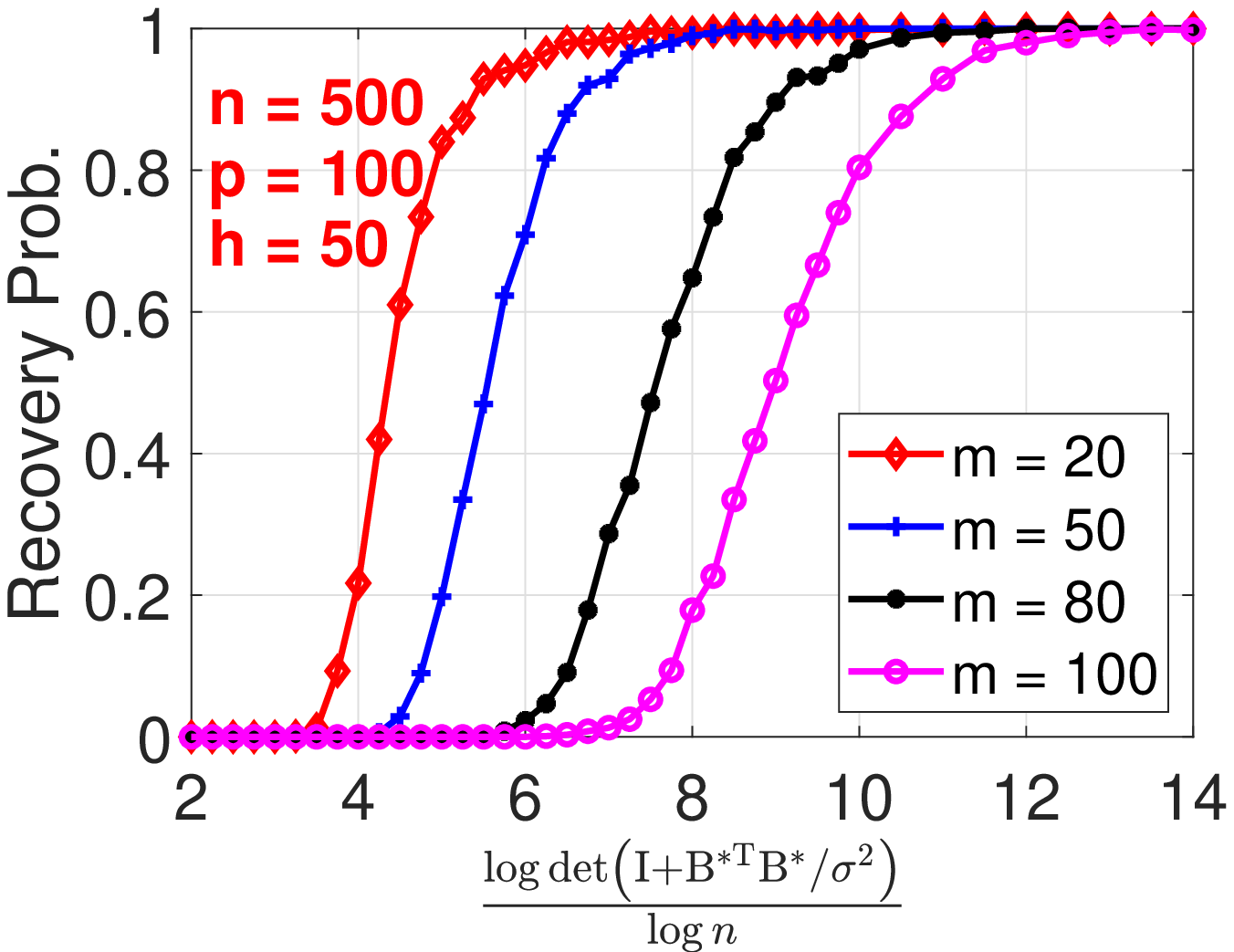}	
\includegraphics[width=2.1in]{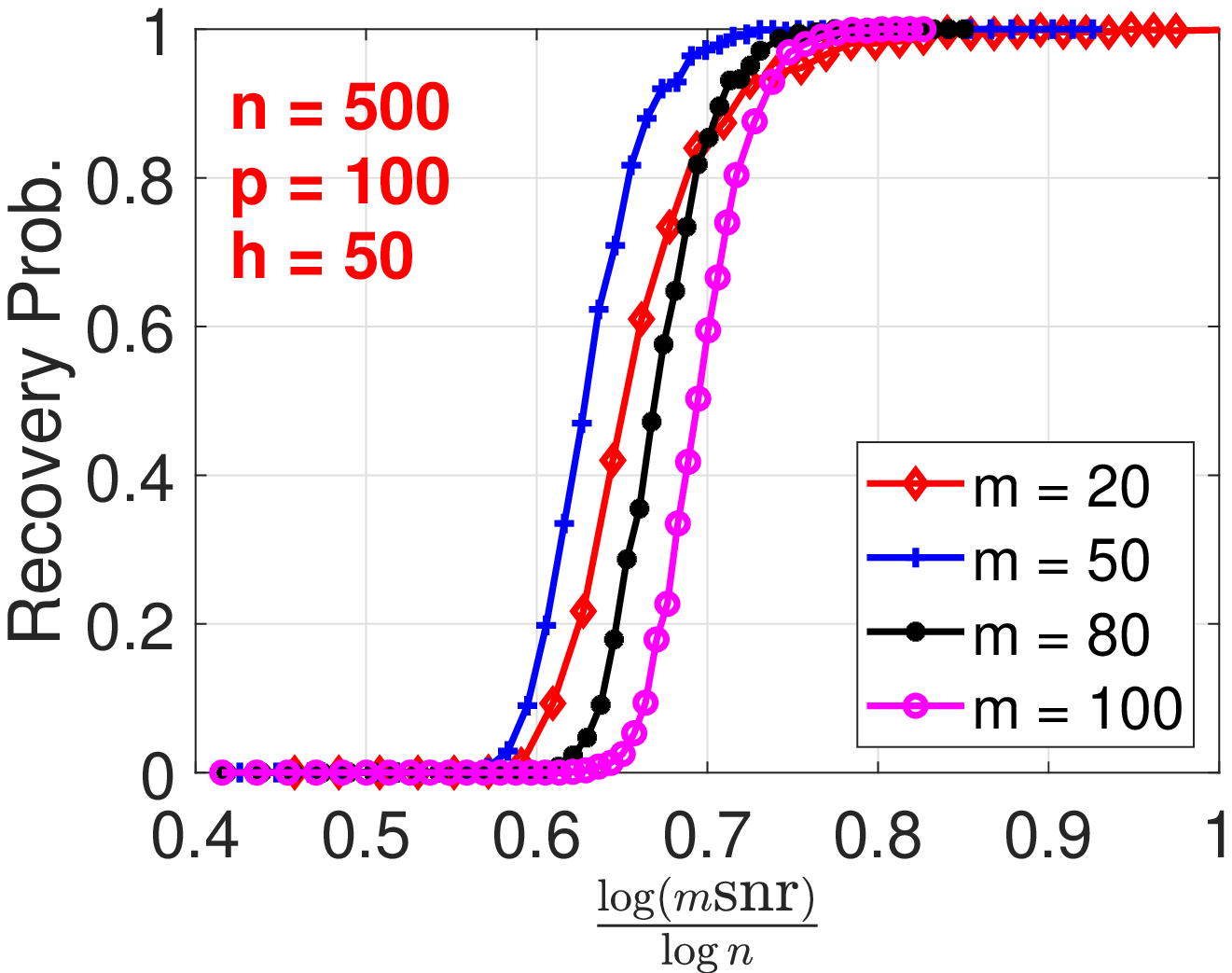}
}

\mbox{
\includegraphics[width=2.1in]{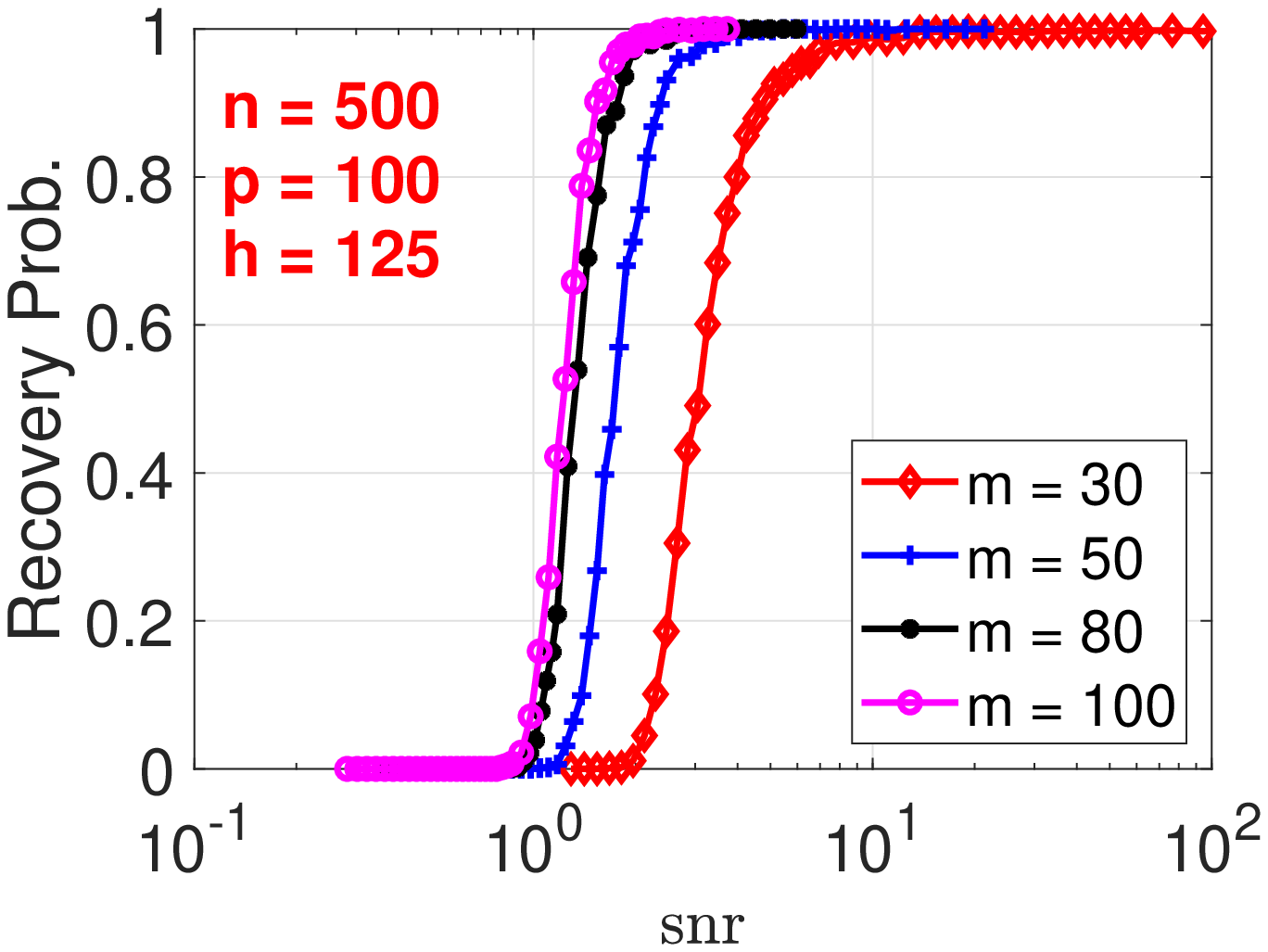}	
\includegraphics[width=2.1in]{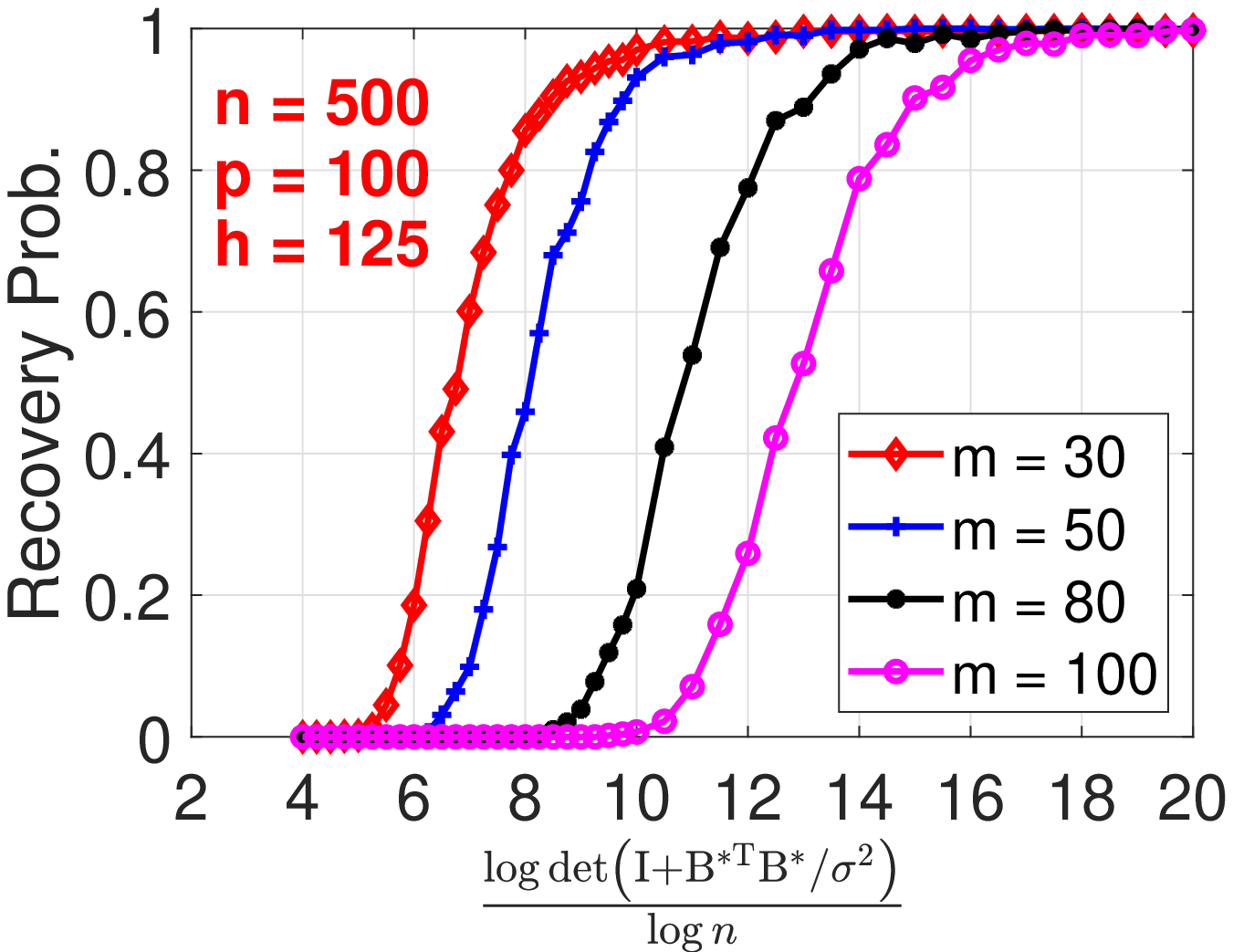}	
\includegraphics[width=2.1in]{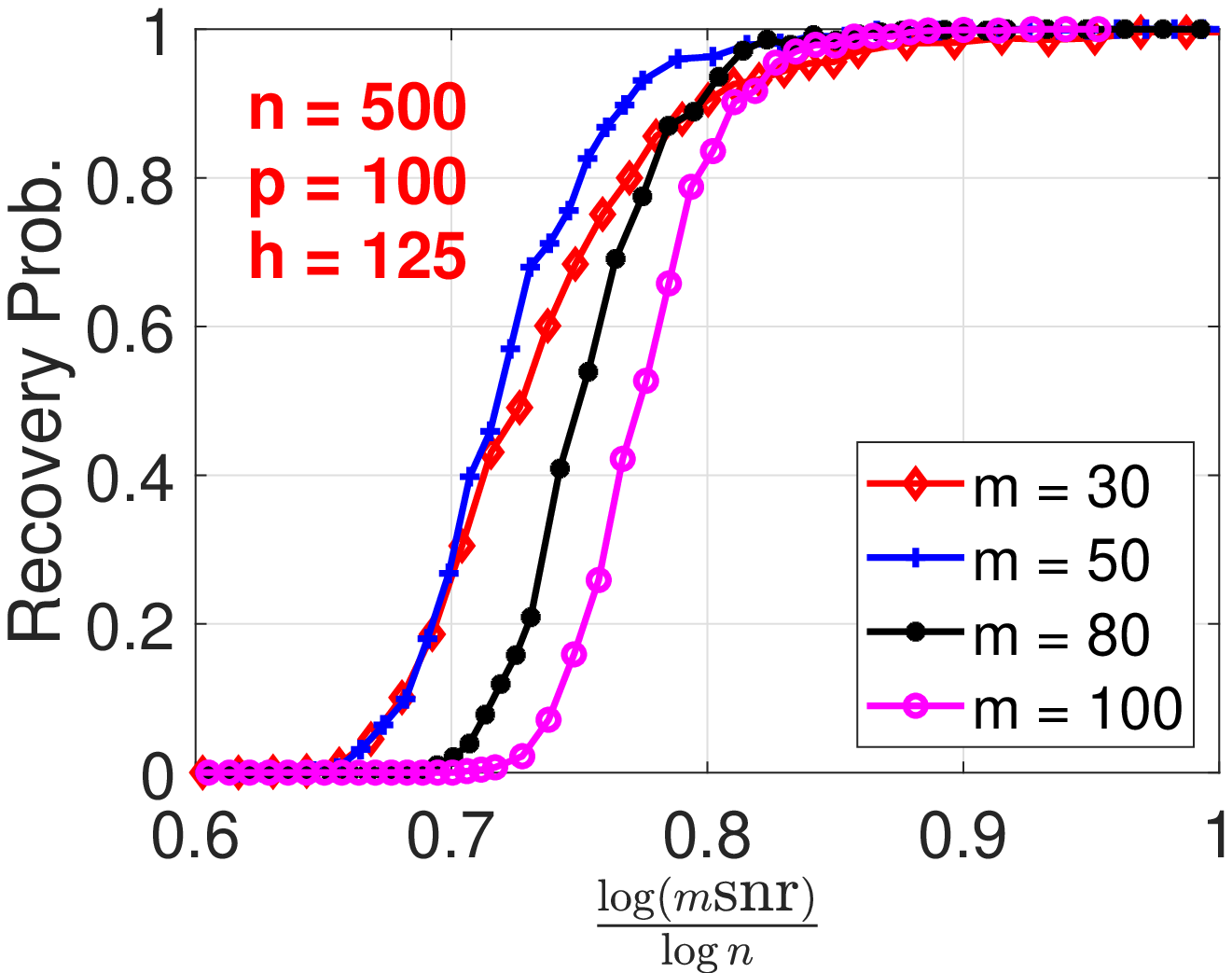}
}

\end{center}
\vspace{-0.2in}
\caption{Realistic case: Correct recovery probability
$\Prob(\hat{\bPi} = \bPi^{*})$ (vertical axis) versus
$\snr$ (\textbf{left panels})
or $\frac{\log \det(\bI + \bB^{*\rmt}\bB^{*}/\sigma^2)}{\log n}$
(\textbf{middle panels}), or
$\frac{\log(m\snr)}{\log n}$ (\textbf{right panels}).
}
\label{fig:prac_recover}
\end{figure}

This subsection is concerned with the realistic case in which
$\bB^{*}$ is not known. We fix $n = 500$ and consider $p = \set{50, 100}$
as well as $h = \set{50, 125}$, where $h = \dh\bracket{\bI; \bPi^{*}}$. The estimator of $\bPi^*$ is obtained by applying Algorithm~\ref{alg:compt_mtd_admm}. The results are
shown in Fig.~\ref{fig:prac_recover}. Given the excessive requirements regarding $\snr$ in the rank-one case even in the oracle case, we here focus on the full-rank case. Apart from
using $\snr$ and $\frac{\logdet\bracket{\bI + \bB^{*\rmt}\bB^{*}/\sigma^2}}{\log n}$ for the horizontal axis, we additionally consider
$\frac{\log(m\snr)}{\log n}$ in virtue of Theorem~\ref{thm:succ_recover}.

Inspection of the left panels of Fig.~\ref{fig:prac_recover} indicates a similar phenomenon as observed in the oracle case, namely,
a significant reduction of the required $\snr$ with large stable rank $\rho(\bB^{*})$ ($=m$).
When $m = 10$, the $\snr$'s requirement
is within the range $[10,~100]$. When $m$ increases to $100$, the
required $\snr$ drops below $1$ in alignment with
the implications of Theorem~\ref{thm:exact_minimax} and
Theorem~\ref{thm:succ_recover_refine}.

However, different from the
oracle case where
the ratio $\frac{\logdet\bracket{\bI + \bB^{*\rmt}\bB^{*}/\sigma^2}}{\log n}$
required for permutation recovery is almost independent of the triple $(n, p, h)$, we now observe variation across different settings. When $n = 500, p = 50, h = 50$, we need
$\frac{\logdet\bracket{\bI + \bB^{*\rmt}\bB^{*}/\sigma^2}}{\log n}$ to be above
$5$ for the correct permutation recovery, which is almost the same
as for the oracle case shown in Fig.~\ref{fig:oracle_recover_fullrank}). However, this
ratio inflates to $[6,~7]$ when $(n, p, h) = (500, 50, 125)$; $12$ when $(n, p, h) = (500, 100, 50)$; and
$16$ when $(n, p, h) = (500, 100, 125)$. Generally speaking,
high $n/p$ ratio and small values of $h$ reduce the required value of the ratio $\frac{\logdet\bracket{\bI + \bB^{*\rmt}\bB^{*}/\sigma^2}}{\log n}$.

\section{Conclusion}\label{sec:conclusion}

In this paper, we have studied the unlabeled sensing problem given
multiple measurement vectors.
First, we establish the statistical limits in terms
of conditions on the $\snr$ implying failure of recovery with high probability,
namely, $\varrho(\bB^{*})\log\bracket{\snr}\lsim \log n$.
The tightness of these conditions is consolidated by the
corresponding condition for correct recovery with $\bB^{*}$
being known. Without knowledge of $\bB^{*}$, we need
$\log\bracket{m\cdot\snr} \gsim \log n$
for correct recovery, which matches the lower bound
for the oracle case with $\varrho(\bB^{*}) = 1$.
By imposing the additional assumption $\dh(\bI; \bPi^{*}) \leq h_{\mathsf{max}}$,
it can be proved that $\varrho(\bB^{*})\log\bracket{\snr}\gsim \log n$ is
sufficient for correct recovery, which matches the minimax bound differ in
a logarithmic factor.
Moreover, we propose an optimization scheme based on the Alternating Direction Methods of Multipliers (ADMM) to tackle the computational difficulties associated with the ML estimator.
The simulation results can largely corroborate our theoretical findings.

%with the exception
%of gaps that are likely attributable to the fact that
%ADMM may deliver spurious local optima with potentially
%different statistical properties compared to the
%global optimum that is the object of our theoretical analysis.
% In future work,
% we aim to bridge this gap by designing improved computational schemes,
% e.g., based on more reliable initialization procedures to avoid
% spurious local optima.

%\newpage
\bibliographystyle{IEEEtran}
%\bibliographystyle{authordate3}
%\bibliography{ref,permute_var}
\bibliography{ref}

%\newpage
\begin{appendices}

\section{Notations}\label{sec:appendix_notation}
We begin the appendix with a restatement of the notations we use.
For an arbitrary matrix $\bA \in \RR^{m\times n}$, we denote by $\bA_{:, i} \in \RR^n$
the $i^{\mathsf{th}}$ column of $\bA$ while
$\bA_{i, :} \in \RR^m$ denotes the $i^{\mathsf{th}}$ row, treated as column vector.
Moreover, $A_{ij}$ denotes the $(i,j)^{\mathsf{th}}$ element of the matrix
$\bA$.
The pseudo-inverse $\bA^{\dagger}$ of the matrix $\bA$ is defined
as $\bracket{\bA^{\rmt}\bA}^{-1}\bA^{\rmt}$.
We define $P_{\bA} = \bA \bA^{\dagger}$
as the projection onto the column space of $\bA$, while
$P_{\bA}^{\perp}= \bI-P_{\bA}$ denotes the projection onto its orthogonal complement.
The \emph{singular value decomposition} (SVD)
of the matrix $\bA$ ~\cite{golub2012matrix} (Section 2.4, P76)
is represented by $\textup{SVD}(\bA)$,
such that
$\textup{SVD}(\bA) = \bU \bSigma \bV^{\rmt}$, $\bU \in \RR^{m\times m}$,
$\bSigma \in \RR^{m\times n}$, and $\bV\in \RR^{n\times n}$,
where $\bU^{\rmt}\bU = \bU \bU^{\rmt} = \bI_{m\times m}$,
$\bV^{\rmt}\bV = \bV\bV^{\rmt} = \bI_{n\times n}$.
The operator
$\textup{vec}(\bA)$ denotes
the vectorization of $\bA$ that is obtained by concatenating the columns of $\bA$
into a vector.
We write $\fnorm{\cdot}$ for the Frobenius norm
while $\opnorm{\cdot}$ is used for the operator norm,
whose definitions can be found
in~\cite{golub2012matrix} (Section 2.3, P71). The ratio $\varrho(\cdot) = \fnorm{\cdot}^2/\opnorm{\cdot}^2$ represents the stable rank
while $r(\cdot)$ represents the usual rank of a matrix.

We write $\bpi(\cdot)$ for a permutation of $\set{1, 2,\cdots, n}$
that moves index $i$ to $\bpi(i)$, $1 \leq i \leq n$.  The corresponding
permutation matrix is denoted by $\bPi$.
We use $\dh(\cdot; \cdot)$ to denote
the Hamming distance between two permutation matrices, $i.e.$,
$\dh(\bPi_1; \bPi_2) = \sum_{i=1}^n \Ind\bracket{\bpi_1(i)\neq \bpi_2(i)}$.
Viewing $\bPi$ as a RV distributed among set $\calH$, we denote
its entropy as $\entH(\bPi)$. The differential entropy is denoted as
$\enth(\cdot)$ and the mutual information is
denoted as  $\entI(\cdot; \cdot)$.

For an event $\+E$, we denote its complement by $\br{\+E}$, and
use $\Psi(\calE)$ to denote $\Expc \Ind(\calE)$.
In addition, we use $a\vcup b$ to denote the maximum of $a$ and $b$
while $a\vcap b$ to denote the minimum of $a$ and $b$.

\section{Proof of Theorem~\ref{thm:exact_minimax}}
\label{thm_proof:oracle_fail_general}

\begin{proof}
The proof of Theorem~\ref{thm:exact_minimax} heavily relies on
Lemma~\ref{lemma:exact_recover_error_bound}.
We put a uniform prior on $\bPi^{*}$ over the support
$\calH$, which maximizes the entropy $H(\bPi^{*}) = \log\abs{\calH}$,
and exploit the inequality
\begin{align}
\label{eq:minimax_error_prob_ub}
\sup_{\bPi} \Expc \Ind(\wh{\bPi}\neq \bPi)
\geq \Prob\bracket{\wh{\bPi}\neq \bPi^{*}}.
\end{align}
Since Lemma~\ref{lemma:exact_recover_error_bound} holds for
arbitrary estimator $\wh{\bPi}$, we can safely add
$\inf_{\wh{\bPi}}$ to the left-hand side in~\eqref{eq:minimax_error_prob_ub} and
complete the proof. 	
\end{proof}

%==========================
\begin{lemma}
\label{lemma:exact_recover_error_bound}
Viewing $\bPi^{*}$ as a RV distributed among the
set $\calH$, we have
\[
\Prob(\wh{\bPi}\neq \bPi^{*}) \geq
\dfrac{\entH(\bPi^{*}) - 1 -  (n/2)
\logdet\bracket{\bI + \bB^{*\rmt}\bB^{*}/\sigma^2}
}{\log\bracket{|\calH|}},
\]
for an arbitrary estimator $\wh{\bPi}$, where
$\entH(\cdot)$ is the entropy of $\bPi^{*}$.
\end{lemma}

\begin{proof}
Without loss of generality, we assume that $\bB^{*}$ is known.
Note that if we cannot recover $\bPi^{*}$
even when $\bB^{*}$ is known, it is hopeless to recover $\bPi^{*}$
with unknown $\bB^{*}$.
%============================
We can reformulate the sensing relation~\eqref{eq:sys_sense_relation}, i.e.,
$\bY = \bPi^{*} \bX \bB^{*} + \bW$,
as the following transmission process
\begin{equation}\label{eq:fail_transimt}
\bPi^{*} \stackrel{\circled{1}}{\rightarrow}
\bPi^{*} \bX\bB^{*}
\stackrel{\circled{2}}{\rightarrow}
\bPi^{*} \bX \bB^{*} + \bW,
\end{equation}
where in $\circled{1}$ the signal $\bPi^{*}$ is encoded to the
codeword $\bPi^{*} \bX \bB_{:, i}^{*}$, and in $\circled{2}$ the
$n$ codewords $\bPi^{*} \bX \bB_{:, i}^{*}$ are transmitted through
$n$ i.i.d. Gaussian channels.
With this reformulation, we can treat the recovery of
$\bPi^{*}$ as a decoding problem. Denote the
recovered permutation matrix as $\wh{\bPi}$.
Following a similar approach as in~\cite{cover2012elements}
(cf. Section~$7.9$, P$206$), we have
\[
\begin{aligned}
\entH(\bPi^{*}) \stackrel{\circled{3}}{=} \entH(\bPi^{*}~|~\bX) %\\
\stackrel{\circled{4}}{=}&~ \entH(\bPi^{*}~|~\wh{\bPi}, \bX) + \entI(\bPi^{*};~\wh{\bPi}~|~\bX) \\
\stackrel{\circled{5}}{\leq} &~ \entH(\bPi^{*}~|~\wh{\bPi}) + \entI(\bPi^{*};~\wh{\bPi}~|~\bX) \\
\stackrel{\circled{6}}{\leq}&~
1 + \log\bracket{|\calH|}\Prob(\wh{\bPi}\neq \bPi^{*}) + \entI(\bPi^{*};~\wh{\bPi}~|~\bX) \\
\stackrel{\circled{7}}{\leq}&~
1 + \log\bracket{|\calH|}\Prob(\wh{\bPi}\neq \bPi^{*}) + \entI(\bPi^{*};~\bY~|~\bX)\\
\stackrel{\circled{8}}{\leq}&~1 + \log\bracket{|\calH|}\Prob(\wh{\bPi}\neq \bPi^{*}) +
\frac{n}{2}
\logdet\bracket{\bI + \frac{\bB^{*\rmt}\bB^{*}}{\sigma^2}},
\end{aligned}
\]
where $\lambda_i$
denote the  $i^{\mathsf{th}}$ singular value of $\bB^{*}$,
in $\circled{3}$ we use the fact that $\bX$ and $\bPi^{*}$ are independent,
in $\circled{4}$ we use the definition
of the conditional mutual information $\entI(\bPi^{*}; \wh{\bPi}~|~\bX)$,
in $\circled{5}$ we use $\entH(\bPi^{*}~|~\wh{\bPi},\bX) \leq \entH(\bPi^{*}~|~\wh{\bPi})$,
in $\circled{6}$ we use Fano's inequality in Theorem~$2.10.1$ in
\cite{cover2012elements}, in $\circled{7}$
we use the data-processing inequality, noting that $\bPi^{*}\rightarrow \bY \rightarrow \wh{\bPi}$ forms a Markov chain
\cite{cover2012elements},
and in $\circled{8}$ we use Lemma~\ref{lemma:fail_channel_capacity}
to upper bound the conditional mutual information
$\entI(\bPi^{*};~\bY~|~\bX)$.

We thus obtain the following lower
bound on $\Prob(\wh{\bPi}\neq \bPi^{*})$
\[
\Prob(\wh{\bPi}\neq \bPi^{*}) \geq
\dfrac{\entH(\bPi^{*}) - 1 -  (n/2)\sum_i \log\left(1 +{\lambda_i^2}/{\sigma^2} \right)}{\log\bracket{|\calH|}},
\]
which is bounded below by $1/2$ provided
\[
\entH(\bPi^{*}) > 1 + \dfrac{n}{2}
\logdet\bracket{\bI + \frac{\bB^{*\rmt}\bB^{*}}{\sigma^2}} +
\dfrac{\log\bracket{|\calH|}}{2},
\]
and complete the proof.
\end{proof}

\begin{restatable}{lemma}{failsnrchannelcapacity}
\label{lemma:fail_channel_capacity}
For the channel described in~\eqref{eq:fail_transimt},
we have
\[
\entI(\bPi^{*};~\bY_{:, 1}, \bY_{:, 2}, \dots, \bY_{:, m}|~\bX) \leq
\dfrac{n}{2}
\logdet\bracket{\bI + \frac{\bB^{*\rmt}\bB^{*}}{\sigma^2}},
\]	
where $\lambda_i$ denotes the $i^{\mathsf{th}}$ singular value of $\bB^{*}$.
\end{restatable}
%====================================

\begin{proof}
Let $\textup{vec}\bracket{\bY}$ ($\textup{vec}\bracket{\bW}$) be the vector  by
concatenating $\bY_{:, 1}, \bY_{:, 2}, \cdots, \bY_{:, m}$
($\bW_{:, 1}, \bW_{:, 2}, \cdots, \bW_{:, m}$),
according to the definition in Appendix~\ref{sec:appendix_notation}.
For simplicity of notation, we use
$\entI\bracket{\bPi^{*};~\textup{vec}\bracket{\bY}|\bX}$ as
a shortcut for $\entI\bracket{\bPi^{*};~\bY_{:, 1}, \bY_{:, 2}, \dots, \bY_{:, m}|~\bX}$.
We then calculate the conditional mutual information
$\entI\bracket{\bPi^{*};~\textup{vec}\bracket{\bY}|\bX}$
as
\begin{align}\label{eq:mutual_intropy_yy_bound}
\entI\bracket{\bPi^{*};\textup{vec}\bracket{\bY}|\bX}
\stackrel{\circled{1}}{=}~&\enth(\textup{vec}(\bY)|\bX) - \enth(\textup{vec}(\bY)|\bX, \bPi^{*}) \notag \\
\stackrel{\circled{2}}{=}~& \Expc_{\bPi^{*}, \bX, \bW}
\enth\bracket{\textup{vec}(\bY)|\bX = \bx}
%\psi_{\textup{vec}(\bY)}(\bX)
- \enth(\textup{vec}(\bW))\notag \\
\stackrel{\circled{3}}{=}~& \Expc_{\bPi^{*}, \bX, \bW}\enth\bracket{\textup{vec}(\bY)|\bX = \bx}
%\psi_{\textup{vec}(\bY)}(\bX)
-\dfrac{mn}{2}\log \sigma^2 \notag \\
\stackrel{\circled{4}}{\leq}~&\Expc_{\bX}\dfrac{1}{2}\log \det \bracket{\Expc_{\bPi^{*}, \bW|\bX =\bx} \textup{vec}\bracket{\bY}\textup{vec}\bracket{\bY}^{\rmt}}
-\dfrac{mn}{2}\log \sigma^2,\notag \\
\leq~&\dfrac{1}{2}\log \det  \Expc_{\bPi^{*}, \bX, \bW} \textup{vec}\bracket{\bY}\textup{vec}\bracket{\bY}^{\rmt} -
\dfrac{mn}{2}\log \sigma^2,
\end{align}
where in $\circled{1}$ we use the definition of the conditional mutual information
$\entI(\bPi^{*};~\textup{vec}\bracket{\bY}|~\bX)$,
in $\circled{2}$ we have used that
\begin{align*}
 \enth(\textup{vec}\bracket{\bY}|~\bX, \bPi^{*})
= \enth(\textup{vec}\bracket{\bPi^{*}\bX \bB^{*} + \bW}|\bX, \bPi^{*}) %\\
=\enth(\textup{vec}\bracket{\bW}|\bX, \bPi^{*}),
\end{align*}
in $\circled{3}$ we use that the  $mn$  entries of $\textup{vec}(\bW)$ are i.i.d Gaussian distributed
with entropy is $\frac{1}{2}\log(\sigma^2)$ each,
in $\circled{4}$ we use a result in~\cite{cover2012elements} (Thm 8.6.5, P254)
which yields
\[
\enth(\bZ) \leq \dfrac{1}{2}\log \det \Cov(\bZ) \leq
\dfrac{1}{2}\log\det \Expc [\bZ\bZ^{\rmt}],
\]
where $\bZ$ is an arbitrary RV with finite covariance matrix $\Cov(\bZ)$,
and we use the concavity:
$\Expc \log \det (\cdot) \leq \log \det \Expc(\cdot)$.

In the sequel, we compute the entries of
the matrix $\Expc_{\bPi^{*}, \bX, \bW}\textup{vec}\bracket{\bY}\textup{vec}\bracket{\bY}^{\rmt}$.
For simplicity of notation, the latter matrix will henceforth be denoted by $\bSigma$.
First note that $\textup{vec}\bracket{\bY}$ equals the concatenation of $\bY_{:, 1}, \bY_{:, 2}, \cdots, \bY_{:, m}$.
We decompose the matrix $\bSigma$ into sub-matrices
$\bSigma_{i_1,i_2} = \Expc_{\bPi^{*}, \bX, \bW} \bY_{:, i_1} \bY_{:, i_2}^{\rmt}$,
$1\leq i_1, i_2 \leq m$,
which corresponds to the covariance matrix between $\bY_{:, i_1}$ and $\bY_{:, i_2}$.
The $(j_1, j_2)^{\mathsf{th}}$ element of sub-matrix
$\Sigma_{i_1,i_2}$ is defined as $\Sigma_{i_1,i_2,j_1,j_2}$.
The latter can be expressed as
\begin{align*}
&\Sigma_{i_1, i_2, j_1, j_2} =
\Expc_{\bPi^{*}, \bX, \bW} \bracket{Y_{j_1, i_1}Y_{j_2, i_2}} \\
=~&\Expc_{\bPi^{*}, \bX, \bW} \Bracket{ \bracket{\bX \bB^{*}_{ :, i_1}}_{\bpi^{*}(j_1)} + W_{j_1, i_1}} \times \Bracket{ \bracket{\bX \bB^{*}_{:, i_2}}_{\bpi^{*}(j_2)} + W_{j_2, i_2}} \\
=~&\Expc_{\bPi^{*}, \bX}\bracket{\la X_{\bpi^{*}(j_1), :}, \bB^{*}_{:, i_1} \ra \la \bX_{\bpi^{*}(j_2), :}, \bB^{*}_{:, i_2} \ra } + \Expc_{\bW} W_{j_1, i_1}W_{j_2, i_2},
\end{align*}
where $\bpi^{*}$ is the permutation corresponding to the permutation matrix $\bPi^{*}$ as defined in Appendix~\ref{sec:appendix_notation}.

We then split the calculation into three sub-cases:
\begin{equation*}
\begin{cases}
\textbf{Case $i_1 = i_2, j_1 = j_2$:}&\Sigma_{i_1,~i_1,~j_1,~j_1}=
\|\bB^{*}_{:, i_1} \|_2^2 + \sigma^2.\\
\textbf{Case $i_1 \neq i_2, j_1 = j_2$:}&\Sigma_{i_1, i_2, j_1, j_1} = \la  \bB^{*}_{:, i_1},  \bB^{*}_{:, i_2}\ra. \\
\textbf{Case $j_1 \neq j_2$:}&\Sigma_{i_1, i_2, j_1, j_2}= 0.
\end{cases}
\end{equation*}
%===========================================
In conclusion, the matrix $\bSigma$ can be expressed as
\[
\begin{aligned}
\bSigma = \underbrace{\begin{bmatrix}
\ltwonorm{\bB^{*}_{:, 1}}^2 + \sigma^2 & \la \bB^{*}_{:, 1}, \bB^{*}_{:, 2}\ra & \cdots & \la \bB^{*}_{:, 1}, \bB^{*}_{:, m}\ra \\
 \la \bB^{*}_{:, 2}, \bB^{*}_{:, 1}\ra   & \ltwonorm{\bB^{*}_{:, 2}}^2 + \sigma^2 & \cdots &  \la\bB^{*}_{:, 2}, \bB^{*}_{:, m}\ra\\
 \vdots &  & \ddots & \vdots \\
\la \bB^{*}_{:, m}, \bB^{*}_{:, 1}\ra   & \la \bB^{*}_{:, m}, \bB^{*}_{:, 2}\ra & \cdots &   \ltwonorm{\bB^{*}_{:, m}}^2 + \sigma^2	
\end{bmatrix}}_{\defequal~\bSigma_{1} }
\otimes \bI_{n\times n}
\end{aligned}
\]
where $\otimes$ denotes the Kronecker product
\cite{golub2012matrix} (Section 1.3.6, P27).
According to~\cite{golub2012matrix} (Section 12.3.1, P709), we have
\begin{align}\label{eq:mutual_entropy_sigma_matrix}
\det\left(\bSigma\right) = \left(\det\left(\bSigma_{1}\right)\right)^n\
\left(\det\left(\bI_{n\times n}\right)\right)^m \stackrel{\circled{5}}{=} \sigma^{2nm}\left(\det\left(\bI + \dfrac{\bB^{*\rmt}\bB^{*}}{\sigma^2}\right)\right)^n,
\end{align}
where in $\circled{5}$ we have calculated $\det(\bSigma_{1})$ as
\[
\det\left(\bSigma_{1}\right) = \det\left(\sigma^2 \bI +\bB^{*\rmt}\bB^{*}\right) =
\sigma^{2m}\det\left(\bI + \dfrac{\bB^{*\rmt}\bB^{*}}{\sigma^2}\right).
\]
By combining~\eqref{eq:mutual_intropy_yy_bound} and
\eqref{eq:mutual_entropy_sigma_matrix}, we have obtained the upper bound
\begin{align*}
\entI\left(\bPi^{*};~\textup{vec}\bracket{\bY}|\bX\right) \leq~&
\dfrac{n}{2}\log \det\left(\bI + \dfrac{\bB^{*\rmt}\bB^{*}}{\sigma^{2}}\right) %\\
\stackrel{\circled{6}}{=} \sum_i \log\left(1 + \dfrac{\lambda_i^2}{\sigma^2} \right),
\end{align*}
where $\circled{6}$ can be  verified via the singular value decomposition
$\textup{SVD}\bracket{\bB^{*}} = \bU \bSigma \bV^{\rmt}$
as introduced in Appendix~\ref{sec:appendix_notation}) and
by using basic properties of the matrix determinant~\cite{horn1990matrix}
(Section~0.3, P8).
\end{proof}

\section{Proof of Proposition~\ref{thm:fail_recover_optim}}
\label{thm_proof:fail_recover_optim}

\subsection{Roadmap}
Observe that the sensing relation
$\bY = \bPi^{*}\bX \bB^{*} + \bW$ is equivalent to
$\bPi^{*\rmt}\bY = \bX \bB^{*} + \bPi^{*\rmt}\bW$. As a consequence of the
rotational invariance of the Gaussian distribution,
$\bPi^{*\rmt}\bW$ follows the same distribution as $\bW$.
Since our proof applies to any instance of the permutation matrix
$\bPi^{*}$, we may assume $\bPi^{*} = \bI$ w.l.o.g.
%We begin the proof by first presenting the roadmap.
We proceed the proof with the following three stages.

\noindent\textbf{Stage I:}
Define $\widetilde{W}_{i,j}$ as
\[
\begin{aligned}
\widetilde{W}_{i, j} = \la \bW_{j, :} - \bW_{i, :},~\dfrac{\bB^{*\rmt}(\bX_{i, :} - \bX_{j, :})}{\ltwonorm{\bB^{*\rmt}(\bX_{i, :} - \bX_{j, :})}}\ra,
\end{aligned}
\]
for $1\leq i < j \leq n$, we would like to prove that
\[
\set{\exists~(i,j),\textup{s.t.}~\widetilde{W}_{i, j} \geq \ltwonorm{\bB^{*\rmt}\bracket{\bX_{i, :} - \bX_{j, :}} } } \subseteq \
\set{\wh{\bPi}\neq \bI}.
\]
\noindent
We then lower bound the probability $\Prob\bracket{\wh{\bPi}\neq \bI}$
as
\[
\Prob(\wh{\bPi}\neq \bI)\geq
\Prob\Bracket{\exists~(i,j),\textup{s.t.}~\widetilde{W}_{i, j} \geq \ltwonorm{\bB^{*\rmt}\bracket{\bX_{i, :} - \bX_{j, :}} } }.
\]

\noindent\textbf{Stage II:}
We lower bound the probability
$\Prob\bracket{\exists~(i,j),\textup{s.t.}~\widetilde{W}_{i, j} \geq \ltwonorm{\bB^{*\rmt}\bracket{\bX_{i, :} - \bX_{j, :}} } }$
by two separate probabilities, namely
\begin{align*}
&\Prob\bracket{\exists~(i,j),\textup{s.t.}~\widetilde{W}_{i, j} \geq \ltwonorm{\bB^{*\rmt}\bracket{\bX_{i, :} - \bX_{j, :}} } }  %\\
\geq ~  \Prob\bracket{\widetilde{W}_{1, j_0} \geq \rho_0 } \Prob\bracket{\ltwonorm{\bB^{*\rmt}\bracket{\bX_{1, :} - \bX_{j_0, :}}} \leq \rho_0},
\end{align*}
where
$j_0$ is picked as $\argmax_{j} \widetilde{W}_{1, j}$,
and $\rho_0$ is one positive parameter waiting to be set.

\noindent\textbf{Stage III}:
Provided
Condition~(\ref{eqn:fail_recover_rank_assump}) holds,
we are allowed to set $\rho_0 = 2\sqrt{2\sigma^2\log n}$
without violating the requirement of Lemma~\ref{lemma:fail_recover_smallball}.
We thereby conclude the proof by setting
$\rho_0 = 2\sqrt{2\sigma^2\log n}$ and
invoking
Lemma~\ref{lemma:fail_recover_noise_bound} and
Lemma~\ref{lemma:fail_recover_smallball}.

\subsection{Proof details}
\begin{proof}
Detailed calculation comes as follows.

\noindent \textbf{Stage I:}
We conclude the proof by showing if
$\set{\widetilde{W}_{i,j} \geq \ltwonorm{\bB^{*\rmt}\bracket{\bX_{i, :} - \bX_{j, :}}}}$ holds,
we would have
\[
\begin{aligned}
&\ltwonorm{\bY_{i, :}- \bB^{*\rmt}\bX_{j, :}}^2 +
\ltwonorm{\bY_{j, :} - \bB^{*\rmt}\bX_{i, :}}^2 %\\
\leq~\ltwonorm{\bY_{i, :} - \bB^{*\rmt}\bX_{i, :}}^2 +
\ltwonorm{\bY_{j, :} - \bB^{*\rmt}\bX_{j, :}}^2,
\end{aligned}
\]
which implies that $\min_{\bPi} \fnorm{\bY -  \bPi\bX\bB^{*}}^2 \leq \fnorm{\bY -  \bX\bB^{*}}^2$ since $\bPi$ can be chosen as the transposition that swaps
$\bY_{i, :}$ and $\bY_{j, :}$.
This implies failure of recovery, i.e., the event $\{ \wh{\bPi} \neq \bI \}$.

\noindent \textbf{Stage II:}
We lower bound the error probability
$\Prob\bracket{\wh{\bPi}\neq \bI}$ as
\[
\begin{aligned}
\Prob\bracket{\wh{\bPi}\neq \bI} %\\
\geq~& \
\Prob\bracket{\exists~(i,j),\textup{s.t.}~\widetilde{W}_{i, j} \geq \ltwonorm{\bB^{*\rmt}\bracket{\bX_{i, :} - \bX_{j, :}} } } \\
\stackrel{\circled{1}}{\geq}~&\Prob\bracket{\widetilde{W}_{1, j_0} \geq
\ltwonorm{\bB^{*\rmt}\bracket{\bX_{1, :} - \bX_{j_0, :}}}}\\
\geq~&\Prob\bracket{\widetilde{W}_{1, j_0} \geq \rho_0~\bigg|
\ltwonorm{\bB^{*\rmt}\bracket{\bX_{1, :} - \bX_{j_0, :}}} \leq \rho_0 }%\\&\times
\Prob\bracket{\ltwonorm{\bB^{*\rmt}\bracket{\bX_{1, :} - \bX_{j_0, :}} } \leq \rho_0 } \\
\stackrel{\circled{2}}{=}~& \Prob\bracket{\widetilde{W}_{1, j_0} \geq \rho_0 }
\Prob\bracket{\ltwonorm{\bB^{*\rmt}\bracket{\bX_{1, :} - \bX_{j_0, :}}} \leq \rho_0 },
\end{aligned}
\]
where in $\circled{1}$ we pick $j_0$ as
$\argmax_{j} \widetilde{W}_{1, j}$ and
in $\circled{2}$ we use the independence
between $\wt{W}_{i,j}$ and \newline $\ltwonorm{\bB^{*\rmt}\bracket{\bX_{i, :} - \bX_{j, :}}}$.
\end{proof}

%================================
\subsection{Supporting lemmas}
\begin{lemma}
\label{lemma:fail_recover_noise_bound}
When $n$ is large ($n \geq 10$), we have
\[
\Prob\bracket{\sup_j \wt{W}_{1,j} \geq 2\sqrt{2\sigma^2\log n}}
\geq 1 - n^{-1}.
\]
\end{lemma}

\begin{proof}
This result is quite standard and can be
easily proved by combining Section~2.5 (P31) and
Thm 5.6 (P126) in~\cite{Boucheron2012}.
We omit the details for the sake of brevity.
\end{proof}

\begin{lemma}
\label{lemma:fail_recover_smallball}
Given that $\rho_0 \geq 2\left(1 +2\sqrt{\frac{\log 2}{c_1\varrho(\bB^{*})}}\right)\fnorm{\bB^{*}}$,
we have
\[
\Prob\bracket{\ltwonorm{\bB^{*\rmt}\bracket{\bX_{i, :} - \bX_{j, :}}} \leq \rho_0}
\geq \dfrac{5}{9},
\]
where $c_1 > 0$ is some constant, and
$\varrho(\bB^{*})$ is the stable rank
of the matrix $\bB^{*}$.
\end{lemma}

%==========================
\begin{proof}
We begin the analysis by defining the following notations,
\begin{align*}
A^{(i,j)}_{\rho_0} \triangleq & \set{
\ltwonorm{\bB^{*\rmt}\bracket{\bX_{i, :} - \bX_{j, :}}} \leq \rho_0},
~1\leq i < j \leq n; \\
\BB\triangleq &
\left\{\bx|\|\bB^{*\rmt}\bx\|_2 \leq
\dfrac{\rho_0}{{2}}\right\}; \\
\zeta = &~\Prob\bracket{\|\bB^{*\rmt}\bx\|_2 \leq \dfrac{\rho_0}{2}},
\end{align*}
respectively, where $\bx \in \RR^p$ is a Gaussian RV satisfying
$\normdist(\bZero, \bI_{p\times p})$, and $\rho_0$ is some
positive parameter awaiting to be set.
First we prove the inequality
$\Prob({A}^{(i, j)}_{\rho_0}) \geq \zeta^2$.
Provided that $\bX_{i, :}\in \BB, \bX_{j,:}\in \BB$, we have
$\+{A}^{(i,j)}_{\rho_0}$ be true because
\[
\ltwonorm{\bB^{*\rmt}\bracket{\bX_{i, :} - \bX_{j, :}}} \leq
\ltwonorm{\bB^{*\rmt}\bX_{i, :}} +
\ltwonorm{\bB^{*\rmt}\bX_{j, :}} \leq \rho_0.
\]
Then we conclude
\begin{align}
\label{eq:geometric_ball}
\Prob\left({A}^{(i, j)}_{\rho_0}\right) \geq &~
\Expc\Ind\bracket{\bX_{i, :}\in \BB, \bX_{j,:}\in \BB} \
\stackrel{\circled{1}}{=} \zeta^2,
\end{align}
where $\circled{1}$ is because of the independence
between $\bX_{i, :}$ and $\bX_{j, :}$.
It thus remains to lower bound the $\zeta$, which
is accomplished by setting $\rho_0$ to be $\rho_0\geq 2(1 +t)\fnorm{\bB^{*}}$
and invoking Thm 2.1 in~\cite{rudelson2013hanson}
\begin{align*}
\Prob\bracket{\ltwonorm{\bB^{*\rmt}\bx} \geq (1 + t)\fnorm{\bB^{*}} } %\\
\leq~&\Prob\bracket{\big|\|\bB^{*\rmt}\bx\|_2- \fnorm{\bB^{*}}\big| \geq t\fnorm{\bB^{*}}}%\\
\leq 2e^{-c_1t^2 \varrho(\bB^{*})} \quad \forall t\geq 0.
\end{align*}
Setting $t = 1.4356/\sqrt{c_1\varrho(\bB^{*})}$, we have
$\zeta \geq \sqrt{5}/{3}$,
which implies
$\Prob\bracket{\ltwonorm{\bB^{*\rmt}\bracket{\bX_{i, :} - \bX_{j, :}}} \leq \rho_0} \geq \zeta^2 \geq 5/9$
in view of~\eqref{eq:geometric_ball}
and completes the proof.
\end{proof}

\section{Proof of Corollary~\ref{corollary:oracleapprox}}
\label{corollary_proof:oracleapprox}

\begin{proof}
First we define $\+E \defequal \Ind\left\{\dh(\wh{\bPi}; \bPi^{*}) \geq  \mathsf{D}\right\}$,
which corresponds to the failure in obtaining an approximation of $\bPi^{*}$
within a Hamming ball of radius $ \mathsf{D}$.
Moreover, we suppose that $\bPi^{*}$ follows a uniform distribution over the
set of all $n!$ possible permutation matrices. Using the same logic as
in Section~\ref{thm:exact_minimax}, we conclude
\[
\inf_{\wh{\bPi}}\sup_{\bPi^{*}}
\Expc \Ind(\dh(\wh{\bPi};~\bPi^{*}) \geq  \mathsf{D}) \geq
\Expc \Ind(\calE).
\]
Then we consider the conditional
entropy $\entH(\+E, \bPi^{*}|~\wh{\bPi}, \bY, \bX)$.
The latter can be decomposed as
\begin{align}
\entH(\+E, \bPi^{*}|~\wh{\bPi}, \bY, \bX) %\notag \\
=~& \entH(\bPi^{*}~|~\wh{\bPi}, \bY, \bX) +
\entH(\+E~|~\bPi^{*}, \wh{\bPi}, \bY, \bX)\notag \\
\stackrel{\circled{1}}{=}~& \entH(\bPi^{*}~|~\wh{\bPi}, \bY, \bX)
\stackrel{\circled{2}}{=} \entH(\bPi^{*}~|~\bY, \bX),
\label{eq:orac_approx_t1}
\end{align}
where in $\circled{1}$ we have used that
$\entH\bracket{\+E~|~\bPi^{*}, \wh{\bPi}, \bY, \bX} = 0$
since $\+E$ is deterministic
once $\bPi^{*}, \wh{\bPi}, \bY, \bX$ are given,
and in $\circled{2}$ we use the
fact $\entI(\wh{\bPi}; \bPi^{*}~|~\bY, \bX) = 0$
since $\wh{\bPi}$ and $\bPi^{*}$ are independent
given $\bX$ and $\bY$.  At the same time, we have
\begin{align}
 \entH\bracket{\+E, \bPi^{*}|~\wh{\bPi}, \bY, \bX} %\notag \\
=~& \entH\bracket{\+E|~\wh{\bPi}} + \entH\bracket{\bPi^{*}~|~\+E, \wh{\bPi}, \bY, \bX}\notag  \\
\stackrel{\circled{3}}{\leq} ~&\log 2 + \entH\bracket{\bPi^{*}~|~\+E, \wh{\bPi}, \bY, \bX} \notag \\
\leq ~& \log 2 + \Prob(\+E = 1) \entH(\bPi^{*}|\+E = 1, \wh{\bPi}, \bY, \bX) %\notag \\
+\Prob(\+E = 0) \entH(\bPi^{*}|\+E = 0, \wh{\bPi}, \bY, \bX) \notag \\
\leq~& \log 2 + \Prob\left(\+E = 1\right) \entH(\bPi^{*}|\+E = 1,~\wh{\bPi})%\notag \\
+ \Prob\left(\+E = 0\right)\entH(\bPi^{*}|\+E = 0,~\wh{\bPi}) \notag \\
\leq~& \log 2 + \left(1 - \Prob\left(\+E = 0\right)\right) \entH(\bPi^{*}) %\notag \\
+ \Prob\left(\+E = 0\right)\log \dfrac{n!}{(n- \mathsf{D}+1)!} \notag \\
\stackrel{\circled{4}}{=}~& \log 2 + \entH(\bPi^{*}) -
\Prob\left(\+E = 0\right)\log(n- \mathsf{D}+1)!
\label{eq:orac_approx_t2},
\end{align}
where in $\circled{3}$ we use the fact that $\+E$ is binary and hence
$\entH(\+E|\wh{\bPi}) \leq \log (2)$,
and in $\circled{4}$ we use the fact that $\entH(\bPi^{*}) = \log (n!)$.
Combing~\eqref{eq:orac_approx_t1} and~\eqref{eq:orac_approx_t2}
yields
\begin{align}
\Prob\left(\+E = 0\right) \leq&~ \dfrac{\entI(\bPi^{*};~\bY,~\bX) + \log 2}{\log (n- \mathsf{D}+1)!} \notag \\
\stackrel{\circled{5}}{=}&~ \dfrac{\entI(\bPi^{*}; \bX) +
\entI(\bPi^{*}; \bY|\bX) + \log 2}{\log (n- \mathsf{D}+1)!} \notag \\
\stackrel{\circled{6}}{=}&~\dfrac{\entI(\bPi^{*};\bY|\bX) + \log 2}{\log (n- \mathsf{D}+1)!}  \notag \\
\stackrel{\circled{7}}{\leq}&~ \dfrac{({n}/{2})\sum_i \log\left(1 + {\lambda_i^2}/{\sigma^2} \right)+
\log 2}{\log (n- \mathsf{D}+1)!},
\label{eq:orac_approx_sum}
\end{align}
which completes the proof of Corollary~\ref{corollary:oracleapprox},
where $\circled{5}$ is because of the chain rule of $I(\bPi^{*}; \bY, \bX)$,
$\circled{6}$ is because $\bPi^{*}$ and $\bX$ are independent and hence
$\entI(\bPi^{*}; \bX) = 0$,
and $\circled{7}$ is because of Lemma~\ref{lemma:fail_channel_capacity}.

In the end, we present the proof for~\eqref{eq:dh_minimax}, which proceeds
as
\[
\inf_{\wh{\bPi}}\sup_{\bPi^{*}}
\Expc\dh(\wh{\bPi}; \bPi^{*}) \geq \
(d+1)\Prob\Bracket{\dh\bracket{\wh{\bPi}; \bPi^{*}} \geq d +1},
\]
where $d$ is an arbitrary integer between $0$ and $n$. Replacing
$\mathsf{D}$ with $d$, we finish the proof with~\eqref{eq:orac_approx_sum}.
\end{proof}

\section{Proof of Theorem~\ref{thm:oracle_succ_optim}}
\label{lemma_proof:oracle_succ_optim}

\subsection{Notations}
\noindent
We first define the events $\calE_0, \calE_1(\delta), \calE_2(\delta)$ as
\begin{align*}
\calE_0 &\triangleq
\bigcap_{i=1}^n \left\{\ltwonorm{\bY_{i, :} -  \bB^{*\rmt}\bX_{i, :} }^2
< \min_{j\neq i} \ltwonorm{\bY_{i, :} - \bB^{*\rmt}\bX_{j, :}}^2 \right\}; \\
\calE_1(\delta) &\defequal \bigcup_{i=1}^n \bigcup_{j\neq i}
\set{2\la \bW_{i, :},~\dfrac{\bB^{*\rmt}\bracket{\bX_{j,:} - \bX_{i, :}}}{\ltwonorm{\bB^{*\rmt}\bracket{\bX_{j,:} - \bX_{i, :}}}}\ra \geq \delta} \\
\calE_2(\delta) &\defequal \bigcup_{i=1}^n \bigcup_{j\neq i}
\set{\ltwonorm{\bB^{*\rmt}\bracket{\bX_{j,:} - \bX_{i, :}}}\leq \delta},
\end{align*}
where $\delta > 0$ is an arbitrary positive number. In addition,
we define probabilities $\calP_1$ and $\calP_2$ as
\begin{align*}
\zeta_1 &\defequal \sum_{i=1}^n \sum_{j\neq i}
\Prob\left\{2\la \bW_{i, :}, \dfrac{\bB^{*\rmt}\bracket{\bX_{j,:}-\bX_{i,:}} }{\ltwonorm{\bB^{*\rmt}\bracket{\bX_{j,:}-\bX_{i,:}}}}\ra \geq \delta \right\}; \\
\zeta_2 &\defequal \sum_{i=1}^n \sum_{j\neq i}
\Prob\bracket{ \ltwonorm{\bB^{*\rmt}\bracket{\bX_{j, :} - \bX_{i, :}}}\leq \delta}.
\end{align*}

\subsection{Roadmap}
We start the proof
by first outlining the proof strategy.

%\begin{itemize}
%\item
\noindent\textbf{Stage I:}
We first show that $\set{\hat{\bPi} \neq \bI} \subseteq \br{\+E}_0$.

%\item
\noindent\textbf{Stage II:}
We would like to upper bound the probability of error
$\Prob(\wh{\bPi}\neq \bI)$ by $\Psi(\br{\+E}_0)$.
By re-arranging terms, we show
$\br{\+E}_0 \subseteq \calE_1(\delta) \bigcup \calE_2(\delta)$,
and separately upper bound $\Psi\bracket{\+E_1(\delta)}$ and
$\Psi\bracket{\+E_2(\delta)}$.

%\item
\noindent\textbf{Stage III:}
Treating the above upper bounds as functions
of $\delta$,
we complete the proof by choosing $\delta$ appropriately
and invoking the Condition~(\ref{eqn:oracle_succ_optim_assump}).
%\end{itemize}
%We now turn to the details of the proof.

\subsection{Proof details}
\begin{proof}
Following a similar argument as in Appendix~\ref{thm_proof:fail_recover_optim},
we assume that $\bPi^{*} = \bI$ w.l.o.g. and consider correct recovery
$\{\wh{\bPi} = \bI\}$.

\noindent\textbf{Stage I:}
We first establish that $\set{\wh{\bPi} \neq \bI} \subseteq \br{\+E}_0$ by
showing that $\+E_0 \subseteq \set{\wh{\bPi} = \bI}$.
Notice that $\+E_0$ can be rewritten as
\[
\+E_0 =
\bigcap_{i=1}^n \bigcap_{j\neq i}\left\{\ltwonorm{\bY_{i, :} - \bB^{*\rmt}\bX_{i, :}}^2
< \ltwonorm{\bY_{i, :} -\bB^{*\rmt}\bX_{j, :}}^2 \right\}.
\]
Based on the definition of the ML estimator~\eqref{eq:sys_ml_estimator},
we must have
%=========================
\begin{equation}
\label{eqn:oracle_succ_optim_relation}
\begin{aligned}
&\ltwonorm{\bY - \wh{\bPi}\bX\bB^{*}}^2\
\leq \ltwonorm{\bY - \bX\bB^{*}}^2,
\end{aligned}
\end{equation}
%=========================
Assuming that $\wh{\bPi}\neq \bI$, then
for each term $\ltwonorm{\bY_{i, :} - \bB^{*\rmt}\bX_{i, :}}$ we have
\begin{align*}
\ltwonorm{\bY_{i, :} - \bB^{*\rmt}\bX_{i, :} }^2 \leq&~
\min_{j\neq i}\ltwonorm{\bY_{i, :} - \bB^{*\rmt}\bX_{j, :}}^2
< \ltwonorm{\bY_{i, :} - \bB^{*\rmt}\bracket{\wh{\bPi}\bX}_{i, :} }^2,
\end{align*}
which leads to
$\ltwonorm{\bY - \bX\bB^{*}}^2 < \ltwonorm{\bY - \wh{\bPi}\bX\bB^{*}}^2$,
contradicting~\eqref{eqn:oracle_succ_optim_relation}.
Hence we have proved that $\+E_0 \subseteq \set{\wh{\bPi} = \bI}$.

%==================================
\noindent\textbf{Stage II:}
In this stage, we will prove that
$\br{\+E}_0 \subseteq \+E_1(\delta) \bigcup \+E_2(\delta)$. First, we expand
$\br{\+E}_0$ as
\[
\br{\+E}_0 =
\bigcup_{i=1}^n \bigcup_{j\neq i}\set{\ltwonorm{\bY_{i, :} - \bB^{*\rmt}\bX_{i, :}}^2
\geq \ltwonorm{\bY_{i, :} -  \bB^{*\rmt}\bX_{j, :} }^2}.
\]
Note that for each event in the union,
the left hand side can be rewritten as $\ltwonorm{\bW_{i, :}}^2$ and
the right hand side can be written as
\begin{align}\notag
&\ltwonorm{\bY_{i, :} - \bB^{*\rmt}\bX_{j, :}}^2
=\ltwonorm{\bB^{*\rmt}\bX_{i, :} + \bW_{i, :} - \bB^{*\rmt}\bX_{j, :}}^2 \\\notag
=~& \ltwonorm{\bW_{i, :}}^2 + \ltwonorm{\bB^{*\rmt}\bracket{\bX_{i, :} - \bX_{j, :}}}^2 %\\\notag
+2\la \bW_{i, :},~\bB^{*\rmt}\bracket{\bX_{i, :}- \bX_{j, :}} \ra.
\end{align}
Hence, the event $\br{\+E}_0$ is equivalent to
\noindent
\[
\begin{aligned}
\br{\+E}_0 =&~ \bigcup_{i=1}^n \bigcup_{j\neq i}\bigg\{2\la \bW_{i, :},\dfrac{\bB^{*\rmt}\bracket{\bX_{j, :} -\bX_{i, :}}}{\ltwonorm{\bB^{*\rmt}\bracket{\bX_{j, :} -\bX_{i, :}}}}\ra %\\
\geq \ltwonorm{\bB^{*\rmt}\bracket{\bX_{j, :} -\bX_{i, :}}}\bigg\}
\subseteq \+E_1(\delta) \bigcup \+E_2(\delta),
\end{aligned}
\]
\noindent
since otherwise we will have the inequality reversed.
Hence, we can upper bound $\Psi(\br{\+E}_0)$ as
\begin{align*}
\Psi\bracket{\br{\+E}_0} \leq~&\Psi(\+E_1(\delta)) + \Psi(\+E_2(\delta))
\leq \zeta_1 + \zeta_2,
\end{align*}
where $\Psi(\cdot)$ denotes $\Expc \Ind(\cdot)$, and the terms $\zeta_1$ and
$\zeta_2$ can be bounded by Lemma~\ref{lemma:oracle_succ_optim_noise} and
Lemma~\ref{lemma:oracle_succ_optim_small_ball} (given below), respectively.

\noindent\textbf{Stage III:}
Set $\delta^2$ as $16\sigma^2 \log \frac{n}{\epsilon_0}$,
where $\epsilon_0 = \alpha_0^{\kappa \varrho(\bB^{*})}/n$.
We can bound $\zeta_1$ as
\begin{align}
\label{eq:succ_oracle_p1_bound}
\zeta_1 \leq n^2 \
\exp\bracket{-\dfrac{16\sigma^2}{8\sigma^2}\log \frac{n}{\epsilon_0}} =  \epsilon_0^2.
\end{align}
At the same time, we can show that $\zeta_2$
is no greater than $\epsilon_0^2$.
To invoke Lemma~\ref{lemma:oracle_succ_optim_small_ball},
first we need to verify the condition
$\delta^2 < \alpha_0^2 \fnorm{\bB^{*}}^2/2$.
This is proved by
\begin{align}\notag
\dfrac{\fnorm{\bB^{*}}^2}{\sigma^2} \stackrel{\circled{1}}{\geq}&~
32\log\left(\dfrac{n}{\epsilon_0}\right)
\left(\dfrac{n}{\epsilon_0}\right)^{4/\kappa \varrho(\bB^{*})}% \\
\stackrel{\circled{2}}{=}~32 \log\left(\dfrac{n}{\epsilon_0}\right)
\left(\dfrac{n^2}{\alpha_0^{\kappa \varrho(\bB^{*})}}\right)^{4/\kappa  \varrho(\bB^{*})} %\\
\stackrel{\circled{3}}{\geq}~\dfrac{32}{\alpha_0^2} \log\left(\dfrac{n}{\epsilon_0}\right),
\end{align}
where in $\circled{1}$ we use condition~\eqref{eqn:oracle_succ_optim_assump},
in $\circled{2}$ we use the definition of $\epsilon_0 = \alpha_0^{\kappa \varrho(\bB^{*})}/n$,
and in $\circled{3}$ we use $\alpha_0 \in (0, 1)$ and $n \geq 1$.

We can then invoke Lemma~\ref{lemma:oracle_succ_optim_small_ball}
and bound $\zeta_2$ as
\begin{align}
\zeta_2 \leq&~ n^2\left(\dfrac{2\delta^2}{\fnorm{\bB^{*}}^2}\right)^{\kappa  \varrho(\bB^{*})/2}\notag \\
\stackrel{\circled{4}}{=}&~ n^2 \exp\left[-\frac{\kappa  \varrho(\bB^{*})}{2}
\left(\log\left(\frac{\fnorm{\bB^{*}}^2}{\sigma^2}\right) - \log\left(32\log\left(\frac{n}{\epsilon_0}\right)\right) \right)\right] \notag \\
\stackrel{\circled{5}}{\leq} &~
n^2 \exp\left[-\frac{\kappa  \varrho(\bB^{*})}{2}
\left(\frac{4}{\kappa  \varrho(\bB^{*})}\log\dfrac{n}{\epsilon_0}\right)\right] = \epsilon_0^2,
\label{eq:succ_oracle_p2_bound}
\end{align}
where in $\circled{4}$ we plug in the definition
$\delta^2 = 16\sigma^2 \log(n/\epsilon_0)$,
and in $\circled{5}$ we use condition~\eqref{eqn:oracle_succ_optim_assump}.
Combining the bounds for $\zeta_1$ in~\eqref{eq:succ_oracle_p1_bound} and
$\zeta_2$ in~\eqref{eq:succ_oracle_p2_bound} will
complete the proof.
\end{proof}

\subsection{Supporting Lemmas}
\begin{lemma}
\label{lemma:oracle_succ_optim_noise}
It holds that
\begin{align*}
&\sum_{i=1}^n \sum_{j\neq i}
\Prob\bracket{2\la \bW_{i, :},~\dfrac{\bB^{*\rmt}\bracket{\bX_{j,:} - \bX_{i, :}} }{\ltwonorm{\bB^{*\rmt}\bracket{\bX_{j,:} - \bX_{i, :}}}}\ra \geq \delta }
\leq~ n^2 e^{-\delta^2/8\sigma^2}.
\end{align*}	
\end{lemma}

\begin{proof}
First, we consider a single term, namely
$2\la \bW_{i, :},\frac{\bB^{*\rmt}\bracket{\bX_{j,:} - \bX_{i, :}} }{\ltwonorm{\bB^{*\rmt}\bracket{\bX_{j,:} - \bX_{i, :}}}}\ra$, ($1\leq i<j \leq n$).
With $\bX$ fixed, it is easy to check that
this term is a Gaussian random variable with zero mean
and variance $4\sigma^2$.

Then the probability
$\Prob\bracket{2\la \bW_{i, :},~\frac{\bB^{*\rmt}\bracket{\bX_{j,:} - \bX_{i, :}} }{\ltwonorm{\bB^{*\rmt}\bracket{\bX_{j,:} - \bX_{i, :}}}}\ra \geq \delta }$
can be bounded as
\[
\begin{aligned}
\Prob\bracket{2\la \bW_{i, :},~\dfrac{\bB^{*\rmt}\bracket{\bX_{j,:} - \bX_{i, :}} }{\ltwonorm{\bB^{*\rmt}\bracket{\bX_{j,:} - \bX_{i, :}}}}\ra \geq \delta } %\\
=~&\Expc_{\bX}\Prob\bracket{2\la \bW_{i, :},
\dfrac{\bB^{*\rmt}\bracket{\bX_{j,:} - \bX_{i, :}}}{\ltwonorm{\bB^{*\rmt}\bracket{\bX_{j,:} - \bX_{i, :}}}}\ra \geq \delta~|~\bX} \\
\stackrel{\circled{1}}{\leq} ~&\Expc_{\bX}~e^{-{\delta^2}/{8\sigma^2}} =
e^{-{\delta^2}/{8\sigma^2}},
\end{aligned}
\]
where in $\circled{1}$ we use the
tail bound for the Gaussian RV $\bW_{i, :}$.
Combining the above together, we show that
$\calP_1 \leq n^2 e^{-{\delta^2}/{8\sigma^2}}$ and complete the proof.
\end{proof}

\begin{lemma}\label{lemma:oracle_succ_optim_small_ball}

Given that $\fnorm{\bB^{*}}^2 > 2\delta^2/\alpha_0^2$, we have
\[
\sum_{i=1}^n \sum_{j\neq i}
\Prob\bracket{ \norm{\bB^{*\rmt}(\bX_{j, :} -
\bX_{i,:})}{2} \leq \delta} \leq
n^2\left(\dfrac{2\delta^2}{\fnorm{\bB^{*}}^2}\right)^{{\kappa \varrho(\bB^{*})}/{2}},
\]
where $\alpha_0 \in (0, 1)$ is a universal constant.
\end{lemma}

\begin{proof}
We consider an arbitrary term $\ltwonorm{\bB^{*\rmt}(\bX_{i, :} - \bX_{j, :})}$, $(i < j)$, and
define $\bZ = \frac{\bX_{i, :} - \bX_{j, :}}{\sqrt{2}}$. It is easy to verify that
$\bZ$ is a $p$-dimensional random vector with i.i.d. $\normdist(0, 1)$-entries.
We then have
\[
\Prob\bracket{\ltwonorm{\bB^{*\rmt}(\bX_{i, :} - \bX_{j, :})} \leq \delta} =
\Prob\bracket{\ltwonorm{\bB^{*\rmt}\bZ}^2\leq 2\delta^2 }.
\]
%============================
According to Lemma 2.6 in~\cite{latala2007banach} (which is re-stated
in Appendix~\ref{sec:appendix_prob} herein),
this probability can
be bounded as
\[
\begin{aligned}
&\Prob\bracket{\ltwonorm{\bB^{*\rmt}\bZ}^2 \leq 2\delta^2 } =
\Prob\bracket{\ltwonorm{\bB^{*\rmt}\bZ} \leq \sqrt{2}\delta } %\\
\leq \exp\left(-\kappa \varrho(\bB^{*}) \log\left(\dfrac{\fnorm{\bB^{*}}}{\sqrt{2}\delta}\right)
\right) = \left(\dfrac{2\delta^2}{\fnorm{\bB^{*}}^2}\right)^{\kappa \varrho(\bB^{*})/2},
\end{aligned}
\]
provided  $\delta^2 < {\alpha_0^2\fnorm{\bB^{*}}^2}/{2}$,
where $\alpha_0 \in (0, 1)$ is a universal constant.
With the union bound, we complete the proof.
\end{proof}

\section{Proof of Theorem~\ref{thm:succ_recover}}
\label{thm_proof:succ_recover}

\subsection{Notations}
We define the events $\calE_3(h), \calE_5(t, h), \calE_6(t, h)$ as
\begin{align}
\label{eq:event3456}
\calE_3(h) & \defequal
\set{\fnorm{\Proj_{\bPi \bX}^{\perp} \bY}^2 \leq
\fnorm{\Proj_{\bPi^{*}\bX}^{\perp} \bY}^2,~\dh(\bPi; \bPi^{*}) = h}; \notag  \\	
\calE_4(t, h) & \defequal
\set{T_{\bPi} \leq \frac{t\fnorm{\bB^{*}}^2}{m}, \dh(\bPi; \bPi^{*}) = h}; \notag \\
\calE_5(t, h) &\defequal \left\{\
\fnorm{ P_{\bPi \bX}^{\perp} \bY }^2 -
\fnorm{ P_{\bPi \bX}^{\perp} \bW }^2
\leq  \frac{2T_{\bPi}}{3},\dh(\bPi; \bPi^{*}) = h \right\};\notag \\
\calE_6(t,h) &\defequal \left\{\
\bigg|\fnorm{ P_{\bPi^{*}\bX}^{\perp} \bW }^2 -
\fnorm{ P_{\bPi \bX}^{\perp} \bW }^2\bigg|
\geq \frac{T_{\bPi}}{3},\dh(\bPi; \bPi^{*}) = h\right\},
\end{align}
where $T_{\bPi}\defequal \fnorm{ \Proj_{\bPi \bX}^{\perp} \bPi^{*}\bX\bB^{*}}^2$.
Additionally we define $T_i(t, h)$ $(1\leq i \leq 3)$ as
\begin{align}
\label{eq:t123}
T_1(t, h) &\defequal \exp\bracket{-t\times \snr/72}; \notag \\
T_2(t,h) &\defequal 2\exp\bracket{-\bracket{ \frac{t^2\times \snr^2}{mh}\vcap \bracket{t\times \snr}}/288};  \notag \\
T_3(t, h)&\defequal 6r \left[\
\dfrac{tn^{\frac{2n}{n-p}}}{mh} \exp \left(1 - \dfrac{tn^{\frac{2n}{n-p}}}{mh} \right)\right]^{\frac{h}{10}},
\end{align}
respectively, where $t < mh$ and $r$ is the rank of $\bB^{*}$.

\subsection{Roadmap}
We first restate Theorem~\ref{thm:succ_recover} with the
specific values of $c_0, c_1$ and $c_2$.
\begin{theorem*}
Fix $\epsilon > 0$ and let $n > C(\epsilon)$, where $C(\epsilon) > 0$ is a positive
constant depending
only on $\epsilon$.
Provided the following conditions hold:
$(i)$ $\snr \cdot n^{-\frac{2n}{n-p}} \geq 1$; and $(ii)$
\begin{equation}
\label{eq:succ_snr_condition_appendix}
\frac{\log (m\cdot \snr)}{380}\geq
\bracket{1 + \epsilon + \frac{n}{190(n-p)}}\log n
+ \frac{1}{2}\log r(\bB^{*}),
\end{equation}
then the ML estimator in~\eqref{eq:sys_ml_estimator}
gives the ground-truth permutation matrix $\bPi^{*}$ with probability
exceeding $1-c_2 n^{-\epsilon}\Bracket{(n^{\epsilon} - 1)^{-1} \vcup 1}$. 	
\end{theorem*}

With the requirement $n\geq 2p$ and $r(\bB^{*}) \leq (m\vcap p) \leq n/2$, we can further relax~\eqref{eq:succ_snr_condition_appendix} to
\begin{align*}
\log (m\cdot \snr)\geq \
(571 + 380\epsilon)\log n,
\end{align*}
which reduces to the form given in Theorem~\ref{thm:succ_recover}.
Before we proceed, we give an outline of our
proof.

%\begin{itemize}
%\item
\noindent\textbf{Stage I:}
%The incorrect recovery can be written as $\set{\hat{\bPi} \neq \bPi^{*}}$.
%We first transform it to be a union of sub-events, which are written as
We decompose the event $\set{\wh{\bPi} \neq \bPi^{*}}$ as
\begin{equation}
\label{eq:wrong_recover_decompose}
\left\{\wh{\bPi}\neq \bPi^{*}\right\} = \bigcup_{\bPi \neq \bPi^{*}}\
\set{\fnorm{P_{\bPi \bX}^{\perp} \bY}^2 \leq  \fnorm{P_{\bPi^{*}\bX}^{\perp} \bY}^2}
= \bigcup_{h\geq 2}\calE_3(h),
\end{equation}
and bound the probability of each individual event in
\eqref{eq:wrong_recover_decompose}.

%\item
\noindent\textbf{Stage II:}
For fixed Hamming distance $\dh(\bPi; \bPi^{*}) = h$, we will prove
$\Psi(\calE_3(h))\leq  \sum_{i=1}^3 T_i(t, h) + r\exp\left(-n\log \frac{n}{2}\right)$,
where $r$ denotes the rank of $\bB^{*}$, and $t > 0$ is an arbitrary positive number.

%\item
\noindent\textbf{Stage III:}
Under the condition specified by~\eqref{eq:succ_snr_condition} and
$\snr\cdot mn^{-\frac{2n}{n-p}} \geq 323$, we set
$t$ as \newline
${\sqrt{m}h \log\left(\snr\cdot mn^{-\frac{2n}{n-p}}\right)}/{\snr}$
and show that
\begin{align}
\label{eq:succ_prac_stage_ii}
\Psi(\calE_3(h)) \leq  9n^{-(1+\epsilon)h} + r\exp\bracket{-n\log\frac{n}{2}}
\end{align}

%\item
\noindent\textbf{Stage IV:}
We prove that
\[
\begin{aligned}
\Prob(\wh{\bPi}\neq \bPi^{*})
\leq 10.36\bracket{\dfrac{1}{n^{\epsilon} \left(n^{\epsilon} - 1\right)}\vcup \dfrac{1}{n^{\epsilon}}},
\end{aligned}
\]
when $n$ is large,
where $\epsilon > 0$ is some positive constant.
%\end{itemize}

%=========================================
\subsection{Proof details}
\begin{proof}
As the outline of our proof, we start with
providing the details of
Stage I and Stage IV, while the proofs of Stage II and Stage III
are given in Lemma~\ref{lemma:succ_final_bound} and
Lemma~\ref{lemma:succ_p12_bound}, respectively.

\noindent\textbf{Stage I:}
From the  definition of ML estimator in~\eqref{eq:sys_ml_estimator},
failure of recovery requires at least one pair $\bracket{\bPi, \bB}$ distinct
from  $\bracket{\bPi^{*}, \bB^{*}}$ such that
\begin{equation*}
\fnorm{\bY - \bPi \bX \bB}^2 \leq \fnorm{\bY - \bPi^{*} \bX \bB^{*}}^2.
\end{equation*}
Note that the optimal $\bB$ corresponding to  ${\bPi}$
can be expressed as $\bB =({\bPi} \bX)^{\dagger} \bY$, where
$\bracket{\bPi \bX}^{\dagger} \defequal \bracket{\bX^{\rmt}\bX}^{-1}\bX^{\rmt}\bPi^{\rmt} $. Back-substitution yields
\begin{align*}
\fnorm{ \bY - \bPi\bX({\bPi} \bX)^{+}\bY}^2 = \fnorm{P_{\bPi \bX}^{\perp} \bY}^2,
\end{align*}
which proves the claim.

\noindent\textbf{Stage II and Stage III:}
As stated above, the detailed proof can be
found in  Lemma~\ref{lemma:succ_final_bound} and
Lemma~\ref{lemma:succ_p12_bound}.

\noindent\textbf{Stage IV:}
We have
\[
\begin{aligned}
\Prob(\hat{\bPi}\neq \bPi^{*}) %\\
\leq &~\sum_{h\geq 2}{n\choose h}h! \Psi(\calE_3(h))
%\Prob\bracket{\fnorm{\Proj_{\bPi \bX}^{\perp}\bY}^2 \leq \fnorm{\Proj_{\bPi^{*} \bX}^{\perp}\bY }^2,~\dh(\bPi, \bPi^{*}) = h } \\
\stackrel{\circled{1}}{\leq} %&~
\sum_{h\geq 2}{n\choose h}h!
\bracket{9n^{-(1+\epsilon) h} +
r\exp\left(-n\log\dfrac{n}{2}\right)} \\
\stackrel{\circled{2}}{\leq}&~  9\sum_{h\geq 2}n^h n^{-(1+\epsilon)h} +
r\sum_{h\geq 2}n! \exp\left(-n\log \dfrac{n}{2}\right) \\
\stackrel{\circled{3}}{\leq}&~9\sum_{h\geq 2}n^{-\epsilon h} + \
r\sum_{h\geq 2} e\sqrt{n}\
\exp \left(n\log n - n\log\bracket{\frac{n}{2}}-n\right) \\
\leq&~ \dfrac{9}{n^{\epsilon}\left(n^{\epsilon} - 1\right)} +
e\sum_h r n^{\frac{1}{2}}
\exp\left(-n \log\left(\frac{e}{2}\right)\right) \\
\stackrel{\circled{4}}{\leq}&~\dfrac{9}{n^{\epsilon}\left(n^{\epsilon} - 1\right)} +
\dfrac{e}{2}\sum_h  n^{\frac{3}{2}}
\exp\left(-n \log\left(\frac{e}{2}\right)\right)\\
\leq&~\dfrac{9}{n^{\epsilon}\left(n^{\epsilon} - 1\right)}+
\dfrac{e}{2} n^{\frac{5}{2}}
\exp\left(-n \log\left(\frac{e}{2}\right)\right)\\
\stackrel{\circled{5}}{\leq}&~
\dfrac{9}{n^{\epsilon}\left(n^{\epsilon} - 1\right)}+
\dfrac{e}{2}\exp\left(-\epsilon \log n\right)% \\
\leq~10.36 \bracket{\dfrac{1}{n^{\epsilon} \left(n^{\epsilon} - 1\right)}\vcup \dfrac{1}{n^{\epsilon}}},
\end{aligned}
\]
where in $\circled{1}$ we use~\eqref{eq:succ_prac_stage_ii},
in $\circled{2}$ we use $\frac{n!}{(n-h)!}\leq n^h$ and
$\frac{n!}{(n-h)!} \leq n!$, in $\circled{3}$ we use
\emph{Stirling's approximation} in the form $n! \leq en^{n + 0.5}e^{-n}$,
in $\circled{4}$ we use $r\leq \min(m, p)$ and $p \leq \frac{n}{2}$ (according to
our assumption in Section~\ref{sec:sys_mdl}),
and in $\circled{5}$, we use
$n\log(\frac{e}{2}) >\left(\frac{5}{2} + \epsilon\right)\log n$
when $n$ is sufficiently large
(e.g., when $\epsilon = 0.5$, we require
$n \geq 36$; when $\epsilon = 1$, we require
$n\geq 44$). %A more detailed list can be found in
%Tab.~\ref{tab:large_n_require}).
The proof is hence complete.
\end{proof}

\subsection{Supporting lemmas}
\begin{lemma}
\label{lemma:succ_final_bound}
We have
\begin{align*}
\Psi(\calE_3(h))\leq &~
 \sum_{i=1}^3 T_i(t, h) +
 r\bracket{\frac{2}{n}}^n, %\exp\left(-n\log \bracket{\frac{n}{2}}\right),
%\exp\bracket{-\frac{t\times \textup{snr}}{72}}\
%+ 6r \left[\
%\dfrac{tn^{\frac{2n}{n-p}}}{mh} \exp \left(1 - \dfrac{tn^{\frac{2n}{n-p}}}{mh} \right)\right]^{\frac{h}{10}} \\
%+& r\exp\left(-n\log \bracket{\frac{n}{2}}\right) +
%2\exp\bracket{-\frac{1}{288}\bracket{ \frac{t^2\times \snr^2}{mh}\vcap \bracket{t\times \snr}}},
\end{align*}
where $\Psi(\cdot)$ denotes $\Expc \Ind(\cdot)$, and $t < mh$ is an arbitrary
positive number.
\end{lemma}

\begin{proof}
The proof is completed by the following decomposition, which reads
\begin{align*}
\Psi(\calE_3(h)) \leq~& \Psi(\calE_4(t, h)) + \
\Psi\bracket{\calE_3(h)\bigcap \br{\calE}_4(t, h)} \\
\stackrel{\circled{1}}{\leq}~& \Psi(\calE_4(t, h)) + \Psi\bracket{\br{\calE}_4(t, h) \bigcap  \calE_5(t, h)} + \Psi\bracket{\br{\calE}_4(t, h) \bigcap\calE_6(t, h)},
\end{align*}
where $\circled{1}$ is due to the
relation
$\calE_3(h) \subseteq \calE_5(t, h)\bigcup \calE_6(t, h)$.
A detailed explanation is given as follows.
Conditional on $\br{\calE_5}(t, h)\bigcap \br{\calE_6}(t, h)$,
we have
\[
\begin{aligned}
\fnorm{P_{{\bPi}\bX}^{\perp}\bY}^2 -
\fnorm{P_{\bPi^{*} \bX}^{\perp}\bY}^2 %\\
\stackrel{\circled{2}}{=}~& \fnorm{P_{{\bPi}\bX}^{\perp}\bY}^2 \hspace{-0.5mm}-
\hspace{-0.5mm}\
\fnorm{P_{\bPi \bX}^{\perp} \bW}^2 +
\fnorm{P_{\bPi \bX}^{\perp} \bW}^2 \hspace{-0.5mm}-\hspace{-0.5mm}
\fnorm{P_{\bPi^{*} \bX}^{\perp}\bW}^2 \\
\geq~& \fnorm{P_{{\bPi}\bX}^{\perp}\bY}^2 \hspace{-0.5mm}-\hspace{-0.5mm}
\fnorm{P_{\bPi \bX}^{\perp} \bW}^2 \hspace{-0.5mm}-\hspace{-0.5mm}
\left|\fnorm{P_{\bPi \bX}^{\perp} \bW}^2\hspace{-0.5mm} -
\hspace{-0.5mm}
\fnorm{P_{\bPi^{*} \bX}^{\perp}\bY}^2 \right| \\
\stackrel{\circled{3}}{>}~&\frac{2T_{\bPi}}{3} - \frac{T_{\bPi}}{3} > 0,
\end{aligned}
\]
where in $\circled{2}$ we use the fact
$P_{\bPi^{*} \bX}^{\perp}\bY = P_{\bPi^{*} \bX}^{\perp}\bW$, and in $\circled{3}$ we use the definitions of $\br{\calE_5}(t, h)$ and $\br{\calE_6}(t, h)$.
This suggests that
$\br{\calE_5}(t, h) \bigcap \br{\calE_6}(t, h) \subseteq \br{\calE}_3(h)$, which
is equivalent to $\calE_3(h) \subseteq \calE_5(t, h)\bigcup \calE_6(t, h)$.
We then separately bound the above terms.

\noindent \textbf{Term $\Psi(\calE_4(t, h))$.}
We first perform $\textup{SVD}\bracket{\bB^{*}} = \bU \bSigma
\bV^{\rmt}$ (as defined in Appendix~\ref{sec:appendix_notation}), such that
$\bSigma = \diag\left(\beta_1, ~\beta_2,~\cdots, \beta_r, 0,\cdots \right)$, where
 $r$ denotes the rank of $\bB^{*}$ ($r \leq \min(m, p)$),
and $\beta_i$ denotes the
corresponding singular values.

Due to the rotational invariance of the Gaussian distribution and
$\bV$ being unitary, it is easy to check that
$T_{\bPi}$ has the same distribution
as $\|P_{\bX}^{\perp} \bPi \bX \bSigma\|_F^2$.
Therefore, we have
\begin{align}
\label{eq:psi4}
\Psi(\calE_4(t, h))
\leq~&\sum_{i=1}^r \Prob\bracket{\fnorm{P_{\bPi\bX}^{\perp}\bPi^{*} \bX\beta_i \be_i}^2 \leq \dfrac{t\beta_i^2}{m}} \notag \\
\stackrel{\circled{4}}{\leq}~& r\Bracket{\bracket{\frac{2}{n}}^n +
6 \left[\dfrac{tn^{\frac{2n}{n-p}}}{mh} \exp \left(1 - \dfrac{tn^{\frac{2n}{n-p}}}{mh} \right)\right]^{\frac{h}{10}}},
\end{align}
where $\circled{4}$ follows from Lemma 5 in~\cite{pananjady2016linear}.

\noindent\textbf{Term $\Psi\bracket{\br{\calE}_4(t, h) \bigcap  \calE_5(t, h)}$.}
We expand
\[
\begin{aligned}
&\fnorm{ P_{\bPi \bX}^{\perp} \bY }^2 -
\fnorm{P_{\bPi \bX}^{\perp} \bW }^2 %\\
=~\fnorm{P_{\bPi \bX}^{\perp} \bPi^{*}\bX\bB^{*}}^2 +
2 \la P_{\bPi \bX }^{\perp} \bPi^{*}\bX\bB^{*},~P_{\bPi \bX}^{\perp} \bW\ra.
\end{aligned}
\]
Conditional on the sensing matrix $\bX$, we have that
$\fnorm{ P_{\bPi \bX}^{\perp} \bY }^2 -
\fnorm{P_{\bPi \bX}^{\perp} \bW }^2$ follows a Gaussian distribution, namely,
$\normdist(T_{\bPi}, 4\sigma^2 T_{\bPi})$. Therefore, we obtain
\begin{align}
\Psi\bracket{\calE_5(t, h)}%\\
\stackrel{\circled{5}}{=}&~\Expc_{\bX}
\bigg[\Ind\bracket{ T_{\bPi} > \frac{t\fnorm{\bB^{*}}^2}{m}}
\times \Expc_{\bW}\Ind\bracket{\fnorm{P_{\bPi \bX}^{\perp} \bY}^2 -
\fnorm{P_{\bPi \bX}^{\perp} \bW}^2
\leq  \frac{2T_{\bPi}}{3}} \bigg]\notag\\
\stackrel{\circled{6}}{\leq} &~
\Expc_{\bX}\left[\Ind\bracket{ T_{\bPi} > \frac{t\fnorm{\bB^{*}}^2}{m}}
\times \exp\bracket{-\frac{T_{\bPi}}{72\sigma^2}}\right]
%\leq &~ \exp\bracket{-\frac{t\fnorm{\bB^{*}}^2}{72m\sigma^2}} =
\leq \exp \bracket{-\frac{t\times \snr}{72}},\label{eq:noise_diff_e1}
\end{align}
where $\circled{5}$ results from  independence
of $\bX$ and $\bW$, and in $\circled{6}$ we use a standard tail bound for Gaussian random variables.

\noindent\textbf{Term $\Psi\bracket{\br{\calE}_4(t, h) \bigcap  \calE_6(t, h)}$.}
 We have
\[
\begin{aligned}
\fnorm{\Proj_{\bPi^{*} \bX}^{\perp}\bW}^2 -
\fnorm{\Proj_{\bPi \bX}^{\perp}\bW}^2 = \fnorm{\Proj_{\bPi \bX}\bW}^2-
\fnorm{\Proj_{\bPi^{*}\bX} \bW}^2  %\\
=\fnorm{\Proj_{\bPi \bX \setminus \bPi^{*}\bX}\bW}^2 -
\fnorm{\Proj_{\bPi^{*}\bX \setminus \bPi \bX }\bW }^2,
\end{aligned}
\]
where $\bPi \bX\setminus \bPi^{*}\bX$ ($\bPi^{*}\bX\setminus \bPi\bX$)
is the short-hand for $\textup{range}\bracket{\bPi\bX}\setminus \textup{range}\bracket{\bPi^{*}\bX}$
($\textup{range}\bracket{\bPi^{*}\bX}\setminus \textup{range}\bracket{\bPi\bX}$).
Setting $k = p\vcap h$, we have that
$\fnorm{\Proj_{\bPi \bX \setminus \bPi^{*}\bX}\bW}^2/\sigma^2$
is $\chi^2$-RV with $mk$ degrees of freedom according to Appendix B.1 in~\cite{pananjady2016linear}.

We conclude that
\begin{align}\label{eq:noise_diff_e2}
 \Psi\bracket{\calE_6(t, h)} %notag \\
\leq~& 2\Prob\bracket{
\left|\fnorm{\Proj_{\bPi \bX \setminus \bPi^{*}\bX}\bW}^2 - mk\sigma^2\right|
\geq \frac{T_{\bPi}}{6},~T_{\bPi}>\frac{t\fnorm{\bB^{*}}^2}{m}} \notag \\
\stackrel{\circled{7}}{\leq}~&2\exp\bracket{-\frac{1}{8}\bracket{ \frac{t^2\times \snr^2}{36mk}\vcap \frac{t\times \snr}{6}}} \notag \\
\leq~& 2\exp\bracket{-\frac{1}{8}\bracket{ \frac{t^2\times \snr^2}{36mh}\vcap \frac{t\times \snr}{6}}},
\end{align}
where in $\circled{7}$ we use the concentration inequality for $\chi^2$-RVs given
in Appendix~\ref{sec:appendix_prob}, Lemma~\ref{lemma:appendix_chi_square}.
We complete the proof by
combing~\eqref{eq:psi4},~\eqref{eq:noise_diff_e1} and
\eqref{eq:noise_diff_e2}.
\end{proof}

\begin{lemma}\label{lemma:succ_p12_bound}
Given that $\snr \cdot n^{-\frac{2n}{n-p}} \geq 1$
and $\log(m\cdot \snr) \geq
380\bracket{1 + \epsilon + \frac{n\log n}{190(n-p)} + \frac{1}{2}\log r(\bB^{*})}$,
where  $\epsilon > 0$ is a constant, we have
one positive $0 < t < mh$ such that
$\sum_{i=1}^3 T_{i}(t, h)\leq 9 n^{-(1 + \epsilon)h}$.
\end{lemma}

\begin{proof} We complete the proof by choosing $t$ as
${\sqrt{m}h\log\bracket{\snr\cdot mn^{-\frac{2n}{n-p}}}}/{\snr}$
and separately bounding $T_i(t, h)$, $1\leq i \leq 3$. Before proceed,
we first check that $t < mh$, which can be easily verified.

\noindent\textbf{Term $T_1(t, h)$:}
We have
\begin{align}
\label{eq:succ_prac_t1_bound}
\exp\bracket{-\frac{t\times \snr}{72}} = & \exp\bracket{-\frac{\sqrt{m}h}{72} \log\bracket{\snr\cdot mn^{-\frac{2n}{n-p}}}} %\notag  \\
\leq & \exp\bracket{-\frac{h}{72} \log\bracket{\snr\cdot mn^{-\frac{2n}{n-p}}}}.
\end{align}

\noindent\textbf{Term $T_2(t, h)$:}
Provided that $\bracket{{t^2\times \snr^2}/(mh)}\vcap \bracket{t\times \snr} = t\times \snr$,
the term $T_2(t, h)$ is of a similar form as $T_1(t, h)$ in
\eqref{eq:succ_prac_t1_bound}.
Here we focus on the case in which
$\frac{t^2\times \snr^2}{mh}\vcap \bracket{t\times \snr} = \frac{t^2\times \snr^2}{mh}$.
The right hand side of this equality can be expanded as
\begin{align*}
\frac{t^2\times \snr^2}{mh} =
h\log^2\bracket{\snr\cdot mn^{-\frac{2n}{n-p}}} \stackrel{\circled{1}}{\geq} h\log\bracket{\snr\cdot mn^{-\frac{2n}{n-p}}},
\end{align*}

\noindent
where in $\circled{1}$ we use the fact $\snr \cdot mn^{-\frac{2n}{n-p}} \geq 323$,
which can be verified by~\eqref{eq:succ_snr_condition}.
We then obtain
\begin{align}\label{eq:succ_prac_t2_bound}
T_2(t, h) \leq 2\exp\Bracket{-\frac{h}{288}\log\bracket{\snr\cdot mn^{-\frac{2n}{n-p}}}}.
\end{align}

\noindent\textbf{Term $T_3(t, h)$:}
We have
\begin{align}
r \left[\
\dfrac{tn^{\frac{2n}{n-p}}}{mh} \exp \left(1 - \dfrac{tn^{\frac{2n}{n-p}}}{mh} \right)\right]^{\frac{h}{10}} %\notag \\
=~& r \exp \left[-\dfrac{h}{10}
\left(\log \dfrac{mh}{tn^{\frac{2n}{n-p}}}+ \dfrac{tn^{\frac{2n}{n-p}}}{mh} - 1\right)\right] \notag \\
\stackrel{\circled{2}}{=}~&r\exp\Bracket{-\frac{h}{10}
\bracket{-\frac{1}{2}\log m
-\log\frac{\log z}{z} + \frac{\sqrt{m}\log z}{z} - 1
}}\notag \\
\leq~&r\exp\Bracket{-\frac{h}{10}
\bracket{-\frac{1}{2}\log m
-\log\frac{\log z}{z} + \frac{\log z}{z} - 1}}\notag \\
\stackrel{\circled{3}}{\leq}~&r \exp\Bracket{-\frac{h}{10}
\bracket{\frac{\log z}{1.9}-\frac{\log m}{2}}} \notag \\
\stackrel{\circled{4}}{\leq}~&r\exp\Bracket{-\frac{h}{380}\log(\snr \cdot mn^{-\frac{2n}{n-p}})},
\label{eq:succ_prac_t3_bound}
\end{align}
where in $\circled{2}$ we
set $z = \snr \cdot mn^{-\frac{2n}{n-p}} \geq 323$,
in $\circled{3}$
we use the fact
$
\frac{\log z}{z} - 1 - \log \frac{\log z}{z} \geq
\frac{\log z}{1.9}$
for $z \geq 323$, and in $\circled{4}$ we use the fact $\snr\cdot n^{-\frac{2n}{n-p}}\geq 1$.

Combining~\eqref{eq:succ_prac_t1_bound}, ~\eqref{eq:succ_prac_t2_bound}
and ~\eqref{eq:succ_prac_t3_bound}, we conclude that
$
\sum_{i=1}^3 T_i(t, h) \leq 9r\exp\left[-\frac{h}{380}\log \left(\snr\cdot mn^{-\frac{2n}{n-p}} \right)\right]$.
Under the condition specified by~\eqref{eq:succ_snr_condition}, we have
\[
\begin{aligned}
&\frac{\log \left(\snr\cdot mn^{-\frac{2n}{n-p}}\right)}{380} = \
\frac{\log\left(m\cdot \snr \right)}{380} - \dfrac{n\log n}{190(n-p)}% \\
\geq~ \left(1 + \epsilon \right)\log n + \dfrac{1}{2}\log r.
\end{aligned}
\]
Hence, we have
\[
\begin{aligned}
r\exp\Bracket{-\frac{h}{380}\log \left(\snr\cdot mn^{-\frac{2n}{n-p}}\right)}
\leq r\exp\left[-h\left(1 + \epsilon \right)\log n - \dfrac{h}{2}\log r\right] \stackrel{\circled{5}}{\leq} n^{-\left(1 + \epsilon\right) h},
\end{aligned}
\]
where in $\circled{5}$ we have $r^{1-\frac{h}{2}} \leq 1$ since $h \geq 2$. This completes the proof.
\end{proof}

\section{Proof of Theorem~\ref{thm:succ_recover_refine}}
\label{thm_proof:succ_recover_refine}

\subsection{Notations}
We define events $\calE_7(h), \calE_8(t, h)$ as
\begin{align*}
\calE_7(h) \defequal &
\set{0 < \fnorm{(\bI - \bPi)\bX\bB^{*}}^2 \leq h \fnorm{\bB^{*}}^2,~\dh(\bI;\bPi)=h}; \\
\calE_8(t, h) \defequal & \set{\fnorm{\Proj_{\bX}^{\perp}\frac{(\bI-\bPi)\bX\bB^{*}}{\fnorm{(\bI - \bPi)\bX\bB^{*}}}}^2
< \frac{t}{h},~\dh(\bI; \bPi) = h} .
\end{align*}

Furthermore, we define the terms $T_4$ as
\begin{align*}
T_4(t, h) &\defequal \exp\bracket{\frac{rh}{2}\bracket{\log\bracket{\frac{t}{h} } - \frac{t}{h} + 1} + 4.18rh},
%\\T_5(t, h)&\defequal  \exp\bracket{-mt\times \snr/72}; \\
%T_6(t, h) &\defequal 2\exp\bracket{-\frac{mt\times \snr}{288}\bracket{\frac{m\times\snr}{h}  \vcap 1}},
\end{align*}
where $r$ denotes the rank of $\bB^{*}$.

\subsection{Proof outline}
We first restate Theorem~\ref{thm:succ_recover_refine} as the following,
where the specific values of $c_i$, $0\leq i \leq 5$ are given.
Notice that our proof focuses on the order
and hence  these values are not sharpened to their limits.

\begin{theorem*}
Suppose that $\dh(\bI; \bPi^{*}) \leq h_{\max}$ with $h_{\max}$ satisfying the relation
$h_{\mathsf{max}}r(\bB^{*}) \leq n/8$. Let further $\epsilon > 0$ be arbitrary, and suppose that
$n > N_2(\epsilon)$, where $N_2(\epsilon) > 0$ is a positive constant depending
only on $\epsilon$.
In addition, suppose that the following conditions hold:
\begin{align*}
(i)\;\snr > 26.2, \quad (ii)\;\varrho(\bB^{*})\geq  5(1+\epsilon)\log n/c_0, \quad (iii)\;\log\bracket{\snr}\geq
\frac{288(1+\epsilon)\log n}{\varrho(\bB^{*})} + 33.44.
\end{align*}
Then the ML estimator \eqref{eq:sys_ml_estimator} subject to the constraint $\dh(\bI; \bPi)\leq h_{\mathsf{max}}$
equals $\bPi^{*}$ with probability at least
$1 - 10n^{-\epsilon}\Bracket{(n^{\epsilon} - 1)^{-1}\vcup 1}$,
where $c_0, \ldots, c_4 > 0$ are some positive constants.
\end{theorem*}

\noindent
Here we adopt the same proof strategy as in
Theorem~\ref{thm:succ_recover}.
For the sake of brevity, we
only present the parts that
are different compared with the proof of Theorem~\ref{thm:succ_recover}.

\noindent\textbf{Stage I:}
Given the requirement $\dh(\bI; \bPi^{*}) \leq h_{\mathsf{max}}$, the triangle inequality implies that
\[
\dh(\wh{\bPi}; \bPi^{*}) \leq \dh(\bI; \wh{\bPi}) +
\dh(\bI; \bPi^{*}) \leq 2h_{\mathsf{max}}.
\]
Hence, we can confine ourselves to the
case in which $\dh(\bPi; \bPi^{*}) \leq  2h_{\mathsf{max}}$.

\noindent\textbf{Stage II}:
We replace Lemma~\ref{lemma:succ_final_bound} with
Lemma~\ref{lemma:succ_final_bound_refine}.

\noindent\textbf{Stage III}:
We replace Lemma~\ref{lemma:succ_p12_bound} with
Lemma~\ref{lemma:succ_p12_bound_refine}.

\noindent\textbf{Stage IV}:
We use the same argument as Stage IV in proving
Theorem~\ref{thm:succ_recover} and complete the proof.

\subsection{Supporting lemmas}
\begin{lemma}
\label{lemma:succ_final_bound_refine}
Given that $rh \leq n/4$ and $t\leq 0.125h$, we have
$\Psi(\calE_3(h)) \leq T_1(mt, h) + T_2(mt, h) + T_4(t, h) +
6\exp\bracket{-\frac{c_0h\varrho(\bB^{*})}{5}}$,
where $h\geq 2$, $\calE_3$ is defined in~\eqref{eq:event3456},
and $T_1(\cdot, \cdot), T_2(\cdot,\cdot)$ are defined in~\eqref{eq:t123}.
\end{lemma}
%===============================
\begin{proof}
Similar to the proof of
Lemma~\ref{lemma:succ_final_bound},
we bound $\Psi(\calE_3(h))$ by decomposing it as
\begin{align*}
\Psi(\calE_3(h)) \leq ~&
\Psi(\calE_3(h)\bigcap \br{\calE}_4(mt, h)) + \Psi(\calE_4(mt, h)) \\
\stackrel{\circled{1}}{\leq}~&\Psi(\calE_3(h)\bigcap \br{\calE}_4(mt, h)) +
\Psi(\calE_7(h)) + \Psi(\calE_8(t, h)),
\end{align*}
where $\calE_3(h), \calE_4(mt, h)$ are defined in~\eqref{eq:event3456},
$\circled{1}$ is due to
\begin{align*}
\calE_4(mt) = \set{\fnorm{\Proj_{\bX}^{\perp}\bPi \bX \bB^{*}}^2\leq t \fnorm{\bB^{*}}^2,\dh(\bI; \bPi) = h},
\end{align*}
event $\set{\fnorm{(\bI-\bPi)\bX\bB^{*}} = 0}$ being with measure zero,
and the relation $\fnorm{\Proj_{\bX}^{\perp}\bPi \bX \bB^{*}} = \fnorm{\Proj_{\bX}^{\perp}(\bI - \bPi)\bX\bB^{*}}$.

Using the same argument as in
proving Lemma.~\ref{lemma:succ_final_bound}, we can prove
$
\Psi(\calE_3(h)\bigcap \br{\calE}_4(mt, h))  %\\
\leq T_1(mt, h) + T_2(mt, h).
$
For the clarify of presentation, we leave
the proof for $\Psi(\calE_7(h))$ and
$\Psi(\calE_8(t, h))$ to Lemma~\ref{lemma:event7}
and Lemma~\ref{lemma:event8}, respectively.
The proof is completed by summarizing the upper upper-bounds.
\end{proof}

% =============================================
\begin{lemma}
\label{lemma:succ_p12_bound_refine}
Given that $\snr > 26.2$,
$rh \leq n/4$, $t\leq 0.125h$, $\varrho(\bB^{*})\geq 5(1+\epsilon)\log n/c_0$,
and
\[
\log\bracket{\snr}\geq \
\frac{288(1+\epsilon)\log n}{\varrho(\bB^{*})} + 33.44,
\]
we have one positive number $0 < t < 0.125h$ such that
\begin{align*}
T_1(mt, h) + T_2(mt, h) + T_4(t, h)+ 6\exp\bracket{-\frac{c_0 h\varrho(\bB^{*})}{5}}
\leq
10 n^{-(1 + \epsilon)h},
\end{align*}	
where $c_0, \epsilon > 0$ are positive constants.
\end{lemma}

\begin{proof}
We complete the proof by choosing
$t ={h \log\bracket{\snr}}/{\snr}$.
Note that if $\snr > 26.2$, we have
$t < 0.125h$.
Given~\eqref{eq:snr_refine_cond},
we have
\begin{equation}
\label{eq:snr_log_inequal}
\log\bracket{\snr}\geq
\frac{288(1+\epsilon)\log n}{\varrho(\bB^{*})}
\stackrel{\circled{1}}{\geq}
\frac{288(1+\epsilon)\log n}{m},
\end{equation}
where in $\circled{1}$ we use $\varrho^{*}(\bB) \leq r(\bB^{*}) \leq m$.
First we verify
\begin{equation}
\label{eq:succ_prac_t1_refine_bound}
e^{-{c_0h\varrho(\bB^{*})}/{5}}\leq n^{-(1+\epsilon)h},
\end{equation}
if $\varrho(\bB^{*})$ satisfies
$\varrho(\bB^{*}) \geq 5(1 + \epsilon)\log n/c_0$.
In the sequel we will separately bound the
terms.

\noindent\textbf{Term $T_1(mt, h)$:}
We have
\begin{align}
\exp\bracket{-\frac{mt\times \snr}{72} } =
\exp\bracket{-\frac{mh}{72}\log(\snr)} \stackrel{\circled{2}}{\leq} n^{-\bracket{1+\epsilon}h},
\label{eq:succ_prac_t4_refine_bound}
\end{align}
where in $\circled{2}$ we use
\eqref{eq:snr_log_inequal}.

\noindent\textbf{Term $T_2(mt, h)$:}
Since we have $\snr \geq  26.2$, we obtain
$\bracket{\frac{mt^2\times \snr^2}{h}\vcap \bracket{mt\times \snr}}
\geq mh\log(\snr)$.
We then have
\begin{align}
T_2(mt, h)\leq2\exp\bracket{-\frac{mh}{288}\times \log(\snr)}
\stackrel{\circled{3}}{\leq} 2n^{-\bracket{1+\epsilon}h},
\label{eq:succ_prac_t3_refine_bound}
\end{align}
where in $\circled{3}$ we use~\eqref{eq:snr_log_inequal}.

\noindent\textbf{Term $T_4(t, h)$:}
We have
\begin{align}
T_4(t, h) \stackrel{\circled{4}}{\leq} \exp\bracket{-\frac{rh}{8}\log\bracket{\snr} +
4.18rh} \stackrel{\circled{5}}{\leq} n^{-\bracket{1+\epsilon}h},
\label{eq:succ_prac_t2_refine_bound}
\end{align}
where in $\circled{4}$ we use
$\frac{\log z}{z} - 1 - \log \frac{\log z}{z} \geq
\frac{\log z}{4}$,
for $z\geq 1.5$,
and in $\circled{5}$ we use
the assumption such that
\begin{align*}
\log\bracket{\snr}\geq \
\frac{8(1+\epsilon)\log n}{\varrho(\bB^{*})} + 33.44.
\end{align*}
We finish the
proof by combining~\eqref{eq:succ_prac_t1_refine_bound},
\eqref{eq:succ_prac_t4_refine_bound},~\eqref{eq:succ_prac_t3_refine_bound},
and\eqref{eq:succ_prac_t2_refine_bound}.
\end{proof}

\begin{lemma}
\label{lemma:event7}
We bound $\Psi(\calE_7(h)) \leq 6\exp\bracket{-\frac{c_0h\varrho(\bB^{*})}{5}}$
for $2\leq h \leq n$. 	
\end{lemma}

\begin{proof}
With $\textup{SVD}\bracket{\bB^{*}} = \bU \bSigma \bV^{\rmt}$,
we first verify
$\fnorm{\bracket{\bI-\bPi}\bX\bB^{*}}=
\fnorm{\bracket{\bI-\bPi}\wt{\bX}\bSigma}$,
where $\wt{\bX} \defequal \bX \bU$. Due to the rotational invariance
of the Gaussian distribution, $\wt{\bX}$ has the same distribution $\bX$.
We separately consider the cases where $h = 2$ and $h\geq 3$.

For $h = 2$, we assume w.l.o.g.~that the first row and second row are permuted.
Then we have
\begin{align}
\Prob\bracket{\|(\bI-\bPi)\wt{\bX}\bSigma\|_{\textup{F}}^2 \leq 2\fnorm{\bB^{*}}^2 } %\\
=~&\Prob\Bracket{2\sum_{i=1}^r \beta_i^2 \bracket{\wt{X}_{1,i} - \wt{X}_{2,i}}^2
\leq 2\bracket{\sum_{i=1}^r \beta_i^2}} \notag \\
\stackrel{\circled{1}}{=}~&\
\Prob\Bracket{\sum_{i=1}^r \beta_i^2 \wt{z}^2_{1,i}
\leq \frac{\sum_{i=1}^r \beta_i^2}{2}}\notag \\
\stackrel{\circled{2}}{=}~&
\Prob\Bracket{ \la \wt{\bz},~\bSigma^2 \wt{\bz}\ra
\leq \frac{\sum_{i=1}^r \beta_i^2}{2}} \stackrel{\circled{3}}{\leq}
2\exp\bracket{-c_0\varrho(\bB^{*})},
\label{eq:succ_prac_h2_dev_bound}
\end{align}
where $\bSigma = \diag\bracket{\beta_1, \cdots, \beta_r, 0, \cdots}$,
$\beta_i$ denotes the
$i^{\mathsf{th}}$ singular values of $\bB^{*}$,
$\wt{X}_{i,j}$ denotes the
$(i, j)$ element of $\wt{\bX}$,
in $\circled{1}$ we define
$\wt{z}_{1, i}= (\wt{X}_{1, i} - \wt{X}_{2, i})/\sqrt{2}$,
in $\circled{2}$ we define $\wt{\bz}$ as the
vectorized version, and
$\Expc\la \wt{\bz}, \bSigma^2 \wt{\bz}\ra = \sum_{i=1}^r \beta_i^2$,
and in $\circled{3}$ we use Theorem~2.5 in~\cite{latala2007banach}
(c.f.~also Appendix~\ref{sec:appendix_prob}) and
$c_0$ is the corresponding constant.

Then we consider the case where $h\geq 3$, by studying the index set $I \defequal \set{j:\bpi(j)\neq j}$,
where $\bpi(\cdot)$ is the permutation corresponding to the permutation matrix
$\bPi$. %, since other indices are canceled out.
Adopting the same argument as in Lemma~8 in~\cite{pananjady2016linear},
we  decompose the index set $I$
into $3$ subsets $\set{I_1, I_2, I_3}$, such that
\begin{itemize}
\item $\sum_{i=1}^3 |I_i| = h$ with
$|I_i|\geq \lfloor h/3\rfloor, 1\leq i \leq 3$.
\item For arbitrary $j$, the indices $j$ and $\bpi(j)$ will not
be in the same index set $I_i$, $(1\leq i \leq 3)$ at the same time.
\end{itemize}

We define a matrix $\wt{\bZ}_i$
which consists of the rows
$(\bI-\bPi)\wt{\bX}\bSigma$ corresponding to
indices in $I_i$.
Accordingly, we can verify that $\|(\bI-\bPi)\wt{\bX}\bSigma\|_{\textup{F}}^2 =
\sum_{i=1}^3\|\wt{\bZ}_i\|_{\textup{F}}^2$. %follows different rows of
%$(\bI-\bPi)\wt{\bX}\bSigma$.
Let $h_i$ denote the corresponding cardinality of $|I_i|$, $i = 1, 2, 3$. We have
\begin{align*}
\Prob\bracket{\fnorm{(\bI-\bPi)\bX\bB^{*}}^2 \hspace{-0.2mm} \leq \hspace{-0.2mm} h\fnorm{\bB^{*}}^2} \hspace{-0.5mm}\leq \sum_{i=1}^3 \Prob\bracket{\|\wt{\bZ}_i\|_{\textup{F}}^2 \hspace{-0.2mm}\leq \hspace{-0.2mm} h_i \fnorm{\bB}^2}.
\end{align*}
In the sequel, we bound
$\Prob\bracket{\|\wt{\bZ}_1\|_{\textup{F}}^2 \leq h_1 \fnorm{\bB}^2}$;
the other two probabilities can be bounded similarly.
Since $j$ and $\bpi(j)$ cannot be in $I_1$ simultaneously,
we define
$\wt{z}_{j,k} = (\wt{X}_{j,k} - \wt{X}_{\bpi(j),k})/\sqrt{2}$, $j\in I_1, 1\leq k \leq r$,
and can treat the $\{ \wt{z}_{jk} \}$ as independent
$\normdist(0, 1)$-RVs. Similar to the case $h=2$,
we have
\begin{align}
\Prob\bracket{\|\wt{\bZ}_1\|_{\textup{F}}^2 \leq h_1 \fnorm{\bB}^2} %\\
\stackrel{\circled{4}}{=}~& \Prob\Bracket{\big\langle \wt{\bz}, \diag(\bSigma^2, \cdots, \bSigma^2) \wt{\bz}\big\rangle \
\stackrel{\circled{5}}{\leq} \frac{h_1(\sum_{i=1}^r \beta_i^2)}{2} } \notag \\
\stackrel{\circled{6}}{\leq}~&2\exp\bracket{-c_0h_1\varrho(\bB^{*})} \stackrel{\circled{7}}{\leq} \
2\exp\bracket{-\frac{c_0h\varrho(\bB^{*})}{5}} \label{eq:succ_prac_hmany_dev_bound},
\end{align}
where the diagonal matrix $ \diag(\bSigma^2, \cdots, \bSigma^2)$ in $\circled{4}$
consists of $h_1$ terms,
in $\circled{5}$ we define $\wt{\bz}$ as the vectorization of $\wt{\bZ}_1$,
in $\circled{6}$ we use Theorem~$2.5$ in~\cite{latala2007banach} (also listed in Appendix~\ref{sec:appendix_prob}),
and in $\circled{7}$ we use the fact $h_i\geq \lfloor h/3 \rfloor$.
We hence bound $\Psi(\calE_7(h))$ by
combing the above cases in~\eqref{eq:succ_prac_h2_dev_bound} and
\eqref{eq:succ_prac_hmany_dev_bound}.
\end{proof}

\begin{lemma}
\label{lemma:event8}	
We bound $\Psi(\calE_8(t, h)) \leq  \exp\bracket{\frac{rh}{2}\bracket{\log\bracket{\frac{t}{h} } - \frac{t}{h} + 1} + 4.18rh}$
for $rh\leq \frac{n}{4}$ and $t\leq 0.125h$.
\end{lemma}

\begin{proof}
For ease of notation, we define
$\bTheta = (\bI-\bPi)\bX\bB^{*}/\fnorm{(\bI-\bPi)\bX\bB^{*}}$.
Then the probability of the event $\calE_8$ can
be bounded as
\begin{align*}
\Psi(\calE_8)
= &~\
\Prob\bracket{\fnorm{\Proj_X^{\perp}\bTheta}^2  <
\frac{t}{h}\fnorm{\bTheta}^2}
\stackrel{}{=}
\Prob\bracket{\sum_{i=1}^r \fnorm{\Proj_{ \bX}^{\perp} \bTheta_{:, i}}^2 \leq
\frac{t}{h}\fnorm{\bTheta_{:, i}}^2 } \\
\leq &~\sum_{i=1}^r \Prob\bracket{\fnorm{\Proj_{\bX}^{\perp} \bTheta_{:, i} }^2 \leq
\frac{t}{h}\fnorm{\bTheta_{:, i}}^2 }
\stackrel{\circled{1}}{=}\sum_{i=1}^r \Prob\bracket{\norm{\Proj_{ \bX}^{\perp}\btheta_i}{2}^2 \leq \frac{t}{h} },
\end{align*}
where in $\circled{1}$ we define
$\btheta_i$ as the normalized version of $\bTheta_{:, i}$, namely,
$\bTheta_{:, i}/\norm{\bTheta_{:, i}}{2}$. Here, we define the set $\Theta_{h}$ by
\begin{align*}
\Theta_{h} = \
\big\{\btheta \in \mathbb{R}^{n}~|~\norm{\btheta}{2} = 1,
\btheta \; \textup{has at most $h$ non-zero elements}\big\}.
\end{align*}

We can verify that $\btheta_i \in \Theta_{h}$ for $1\leq i \leq r$,
since $\dh(\bI; \bPi) = h \geq 2$.
Before delving into detailed calculations,
we first  summarize our proof strategy:

\begin{itemize}
\item
\noindent\textbf{Step I}:
We cover the set $\Theta_{h}$ with an $\varepsilon$-net $\+N_{\varepsilon}$
such that for arbitrary $\btheta \in \Theta_{h}$, there
exists a $\btheta_0 \in \+N_{\varepsilon}$ such that
$\norm{\btheta_0 - \btheta}{2} \leq \varepsilon$.
\item
\noindent\textbf{Step II}:
Setting $\varepsilon = \sqrt{{t}/{h}}$,
we define events $\calE_{\Theta}$ and $\calE_{\+N_{\varepsilon}}$ by
\begin{align*}
\calE_{\Theta} &\defequal
\set{\exists~\btheta \in \Theta_h~\St~\norm{\Proj_{\bX}^{\perp}\btheta}{2} < \varepsilon = \sqrt{{t}/{h}}}  \\
\calE_{\+N_{\varepsilon}} &\defequal
\set{\exists~\btheta_0 \in \+N_{\varepsilon}~\St~\norm{\Proj_{\bX}^{\perp} \btheta_0}{2} < 2\varepsilon = 2\sqrt{{t}/{h}}}.
\end{align*}
Then we will prove
\begin{align*}
&\Prob\bracket{\fnorm{\Proj_{\bX}^{\perp}\frac{(\bI-\bPi)\bX\bB^{*}}{\fnorm{(\bI - \bPi)\bX\bB^{*}}}}^2 \
< \frac{t}{h}} %\\
\leq~ r\Prob\bracket{\calE_{\Theta}} \leq
r\Prob\bracket{\calE_{\+N_{\varepsilon}}}.
\notag
\end{align*}
\item
\noindent\textbf{Step III}:
We consider an arbitrary fixed element $\btheta_0 \in \+N_{\varepsilon}$,
and study
$\Prob\bracket{\norm{\Proj^{\perp}_{\bX}\btheta_0}{2} \leq 2\varepsilon}$.
Adopting the union bound
\begin{align*}
\Prob\bracket{\calE_{\+N_{\varepsilon}}} \leq
|\+N_{\varepsilon}|\times
\Prob\bracket{\norm{\Proj^{\perp}_{\bX}\btheta_0}{2}\leq 2\varepsilon},
\end{align*}
we finish the bound of $\Psi(\calE_8)$.
\end{itemize}

The following analysis fills in the details.

\noindent\textbf{Stage I:}
We cover the set $\Theta_{h}$ with
an $\varepsilon$-net $\+N_{\varepsilon}$. Its cardinality can be bounded
as
\begin{align*}
|\+N_{\varepsilon}| \stackrel{\circled{2}}{\leq} \bracket{1 + \frac{2}{\varepsilon}}^{h}
\stackrel{\circled{3}}{\leq}\
\bracket{\frac{3}{\varepsilon}}^{h},
\end{align*}
where in $\circled{2}$ we use that elements of $\Theta_h$ have at least $(n-h)$ zero elements,
and accordingly we cover the sphere $\mathbb{S}^{h-1}$ with an $\varepsilon$-net $\+N_{\varepsilon}$,
whose cardinality can be bounded as in~\cite{wainwright_2019},
and in $\circled{3}$ we assume that $\varepsilon \leq 1$.

\noindent\textbf{Stage II:}
We will prove the relation
\begin{align*}
&\Prob\bracket{\fnorm{\Proj_{\bX}^{\perp}\frac{(\bI-\bPi)\bX\bB^{*}}{\fnorm{(\bI - \bPi)\bX\bB^{*}}}}^2 < \frac{t}{h}} \stackrel{\circled{4}}{\leq}  r\Prob\bracket{\calE_{\Theta}}
\stackrel{\circled{5}}{\leq} r\Prob\bracket{\calE_{\+N_{\varepsilon}}},
\notag
\end{align*}
when $\varepsilon = \sqrt{{t}/{h}}$ and $\circled{4}$ follows from
the definition of $\calE_{\Theta}$. We here focus on proving inequality $\circled{5}$, which is done by
\begin{align*}
\Prob\bracket{\calE_{\Theta}}=
\Prob\bracket{\calE_{\Theta}\bigcap \calE_{\+N_{\varepsilon}} } +
\Prob\bracket{\calE_{\Theta}\bigcap \br{\calE}_{\+N_{\varepsilon}} } %\\
\leq &~\Prob\bracket{\calE_{\+N_{\varepsilon}}} +
\Prob\bracket{\calE_{\Theta}\bigcap \br{\calE}_{\+N_{\varepsilon}}}
\stackrel{\circled{6}}{=}\Prob\bracket{\calE_{\+N_{\varepsilon}}},
\end{align*}
where $\circled{6}$ is due to the fact
$\Prob\bracket{\calE_{\Theta}\bigcap \br{\calE}_{\+N_{\varepsilon}}} = 0$.
A detailed explanation is given as follows. Note that, given $\br{\calE}_{\+N_{\varepsilon}}$, it holds that
 for all $\btheta_0 \in \+N_{\varepsilon}$, we have
$\norm{\Proj_{\bX}^{\perp}\btheta_0 }{2} \geq 2\varepsilon$.
Then for arbitrary $\btheta \in \Theta_{h}$,
we consider an element $\btheta_0 \in \+N_{\varepsilon}$ such that
$\norm{\btheta - \btheta_0}{2} \leq \varepsilon$ and consequently
\begin{align*}
\norm{\Proj^{\perp}_{\bX}\btheta}{2} \geq~&
\norm{\Proj^{\perp}_{\bX}\btheta_0}{2} -
\norm{\Proj^{\perp}_{\bX}\bracket{\btheta - \btheta_0}}{2} %\\
\geq ~ 2\varepsilon - \norm{\Proj^{\perp}_{\bX}\bracket{\btheta - \btheta_0}}{2}%\\
\stackrel{\circled{7}}{\geq} ~2\varepsilon - \norm{\btheta - \btheta_0}{2} \stackrel{\circled{8}}{\geq}
\sqrt{\frac{t}{h}},
\end{align*}
where in $\circled{7}$ we use
the contraction property of projections,
and in $\circled{8}$ the
fact $\norm{\btheta - \btheta_0}{2} \leq \varepsilon = \sqrt{{t}/{h}}$.

\noindent\textbf{Stage III:}
We study the probability
$\Prob\bracket{\norm{\Proj_{\bX}^{\perp}\btheta_0}{2}^2 \leq \frac{4t}{h}}$
for fixed $\btheta_0 \in \+N_{\varepsilon}$.
In virtue of results in~\cite{dasgupta2003elementary}, we have
\begin{align*}
\Prob\bracket{\norm{\Proj^{\perp}_{\bX}\btheta_0}{2}^2 \leq
\frac{\alpha (n-p)}{n}\norm{\btheta_0}{2}^2}
\stackrel{}{\leq} \exp\bracket{\frac{n-p}{2}\bracket{\log \alpha- \alpha + 1}},
~~\alpha \leq 1.
\end{align*}
We can set $\alpha = {4nt}/((n-p)h)$ $(< 1)$ and
obtain
\begin{align*}
\Prob\bracket{\norm{\Proj^{\perp}_{\bX}\btheta_0}{2}^2 \leq
\frac{4t }{h}}%\\
=~& \Prob\bracket{\norm{\Proj^{\perp}_{\bX}\btheta_0}{2}^2 \leq
\frac{\alpha (n-p)}{n}}\\
\leq ~& \exp\bracket{\frac{n-p}{2}\bracket{\log\bracket{\frac{4nt}{(n-p)h}}- \frac{4nt}{(n-p)h} + 1}}  \\
\stackrel{\circled{9}}{\leq}~&\exp\bracket{\frac{n}{4}\bracket{\log\bracket{\frac{8t}{h}}- \frac{8t}{h} + 1}},
\end{align*}
where in $\circled{9}$ we use that
$(a)$ $n\geq 2p$, $(b)$ $\log x - x + 1$ is increasing
in range $(0, 1)$, and $(c)$ $\log x+ 1\leq x$.

In the end, we bound $\Psi(\calE_8)$ as
\begin{align}
\Psi(\calE_8) %\notag \\
\leq~&
r\bracket{\frac{3}{\sqrt{{t}/{h}}}}^{h}
\exp\bracket{\frac{n}{4}\bracket{\log \bracket{\frac{8t}{h}} - \frac{8t}{h} + 1}} \notag\\
=~&
\exp\bigg(h\log(3) - \frac{h}{2}\log\bracket{\frac{t}{h}}
+ \log r +\frac{n}{4}\bracket{\log \bracket{\frac{8t}{h}} - \frac{8t}{h} + 1}\bigg)\notag \\
\stackrel{\circled{A}}{\leq}~&
\exp\bigg(\frac{rh}{2}\bracket{\log\bracket{\frac{t}{h}} - \frac{16t}{h} + 1}
+ 3.68rh + \log r \bigg) \notag \\
\stackrel{\circled{B}}{\leq}~&\exp\bracket{\frac{rh}{2}\bracket{\log\bracket{\frac{t}{h}} - \frac{t}{h} + 1} + 4.18rh}\label{eq:prob_omega_1_omega_2c_bound}
\end{align}
where in $\circled{A}$ we use the assumption that $n\geq 4rh$,
and in $\circled{B}$ we use that $rh \geq 2r \geq 2\log(r)$.
Combining~\eqref{eq:succ_prac_h2_dev_bound},
\eqref{eq:succ_prac_hmany_dev_bound} and~\eqref{eq:prob_omega_1_omega_2c_bound}, we finish the proof.
\end{proof}

\begin{remark}
\label{remark:h_constraint}
Note that we cannot improve $h$ from
$\Omega\bracket{\frac{n}{\log n}}$ to $n$ in
general, since there is an inherent problem
when dealing with the case $h\rightarrow n$.
A detailed explanation is given as the following.
The key ingredient in bounding
$\Psi(\calE_8)$ is based on the step
\begin{align*}
\Psi(\calE_8) \leq
\Prob\bracket{\norm{\Proj_{\bX}^{\perp}\btheta}{2} \leq \sqrt{\frac{t}{h}},~~
\exists~\btheta \in \Theta_{h}}
\stackrel{}{\leq}
|\+N_{\varepsilon}|\cdot
\Prob\bracket{\norm{\Proj_{\bX}^{\perp}\btheta_0}{2} \leq \sqrt{\frac{t}{h}} + \varepsilon,~
\exists~\btheta \in \+N_{\varepsilon}} < 1.
\end{align*}
For the extreme case $h = n$, we
cannot have $|\+N_{\varepsilon}|\Prob\bracket{\norm{\Proj_{\bX}^{\perp}\btheta_0}{2} \leq \sqrt{{t}/{h}} + \varepsilon,~
\exists \btheta \in \+N_{\varepsilon}} < 1$
since
\begin{align*}
&\Prob\bracket{\norm{\Proj_{\bX}^{\perp}\btheta}{2} \leq \sqrt{\frac{t}{h}},~~
\exists~\btheta \in \Theta_{n}}
\geq \Prob\bracket{\fnorm{\Proj_{\bX}^{\perp}\frac{\bX\bB^{*}}{\fnorm{\bX\bB^{*}}}} \leq \sqrt{\frac{t}{h}}} = 1.
\end{align*}
The reason behind this is that
we lose control of the cardinality
$|\+N_{\varepsilon}| \lsim  \bracket{C/\varepsilon}^{rh}$ when
$h\rightarrow n$.
% A more detailed calculation shows that
% this problem arises when $h\geq n/2$.
\end{remark}

\section{Useful Probability inequalities}
\label{sec:appendix_prob}

\begin{lemma}[Thm. 2.5 in~\cite{latala2007banach}]
\label{lemma:gauss_proj}
Let $\bA\in \RR^{n\times n}$ be a non-zero matrix
and let $\bxi = (\xi_i)_{i=1}^n$ be a random vector with independent sub-Gaussian
entries such that $(i)$ $\Var(\xi_i)\geq 1$, $1\leq i \leq n$, and
$(ii)$ the sub-Gaussian constant of the $\{ \xi_i \}$ is at most $\beta$. Then $\forall \by \in \RR^n$, there exists a $c_0 > 0$ such that
\[
\Prob\bracket{\norm{\by - \bA\bxi}{2} \leq \frac{\fnorm{\bA}}{2}}
\leq 2\exp\bracket{-\frac{c_0}{\beta^4} \varrho(\bA)}.
\]

\end{lemma}

\begin{lemma}[Lemma 2.6 in~\cite{latala2007banach}]
\label{lemma:subgauss_proj}
Let $\bA\in \RR^{n\times n}$ be a non-zero matrix
and $\bg$ be Gaussian $\normdist(\b0, \bI_{n\times n})$. Then
we have
\[
\Prob\left(\norm{\by - \bA\bg}{2} \leq \alpha{\fnorm{\bA}}\right)
\leq \exp\left(\kappa \log(\alpha)\varrho(\bA)\right),
\]
for any $\alpha \in (0, \alpha_0)$,
where $\by\in \RR^n$ is an arbitrary fixed vector,
$\alpha_0 \in (0, 1)$ and $\kappa > 0$ are
universal constants.
\end{lemma}

\begin{lemma}[\cite{wainwright_2019} (Example 2.11, P29)]
\label{lemma:appendix_chi_square}
For a $\chi^2$-RV $Y$ with $\ell$ degrees of freedom, we have
\[
\Prob\bracket{|Y - \ell| \geq t} \leq
2\exp\bracket{-\bracket{\frac{t^2}{8\ell} \vcap \frac{t}{8}}},
~~\forall~t\geq 0.
\]
\end{lemma}

\end{appendices}

\end{document}